\documentclass[12pt]{article}
\usepackage[colorlinks,citecolor=blue,urlcolor=blue,breaklinks]{hyperref}
\usepackage{xcolor, latexsym,amsmath,amssymb,amsfonts,amsthm,bbm,bm,enumitem,xcolor,soul,graphicx, dsfont,enumitem,natbib,threeparttable,caption,subcaption, comment}
\usepackage[normalem]{ulem}
\usepackage[ruled, vlined]{algorithm2e}
\usepackage{multirow}
\SetKwInput{KwSet}{Set}

\newtheorem{thm}{Theorem}
\newtheorem{prop}[thm]{Proposition}
\newtheorem{lemma}[thm]{Lemma}

\DeclareMathOperator*{\argmax}{argmax}

\DeclareMathOperator*{\sargmax}{sargmax}

\DeclareMathOperator*{\esssup}{ess\,sup}

\newcommand{\bdot}{\bm{\cdot}}

\title{High-dimensional, multiscale online changepoint detection}
\author{Yudong Chen, Tengyao Wang and Richard J. Samworth}
\date{(\today)}

\begin{document}
\maketitle

\begin{abstract}
	We introduce a new method for high-dimensional, online changepoint detection in settings where a $p$-variate Gaussian data stream may undergo a change in mean.  The procedure works by performing likelihood ratio tests against simple alternatives of different scales in each coordinate, and then aggregating test statistics across scales and coordinates.  The algorithm is online in the sense that both its storage requirements and worst-case computational complexity per new observation are independent of the number of previous observations; in practice, it may even be significantly faster than this.  We prove that the patience, or average run length under the null, of our procedure is at least at the desired nominal level, and provide guarantees on its response delay under the alternative that depend on the sparsity of the vector of mean change.  Simulations confirm the practical effectiveness of our proposal, which is implemented in the \texttt{R} package \texttt{ocd}, and we also demonstrate its utility on a seismology data set.
\end{abstract}

\section{Introduction}
\label{Sec:Intro}
Modern technology has not only allowed the collection of data sets of unprecedented size, but has also facilitated the real-time monitoring of many types of evolving processes of interest.  Wearable health devices, astronomical survey telescopes, self-driving cars and transport network load-tracking systems are just a few examples of new technologies that collect large quantities of streaming data, and that provide new challenges and opportunities for statisticians.

Very often, a key feature of interest in the monitoring of a data stream is a \emph{changepoint}; that is, a moment in time at which the data generating mechanism undergoes a change.  Such times often represent events of interest, e.g.~a change in heart function, and moreover, the accurate identification of changepoints often facilitates the decomposition of a data stream into stationary segments.

Historically, it has tended to be univariate time series that have been monitored and studied, within the well-established field of statistical process control \citep[e.g.][]{Duncan1952,Page1954,Barnard1959,FearnheadLiu2007,Oakland2007,TNB2014}.  These days, however, it is frequently the case that many data processes are measured simultaneously.  In the context of changepoint detection, this introduces the new challenge of borrowing strength across the different component series in an attempt to detect much smaller changes than would be possible through the observation of any individual series alone.

The field of changepoint detection and estimation also has a long history \citep[e.g.][]{Page1955}, but has been undergoing a marked renaissance in recent years; entry points to the field include \citet{CsorgoHorvath1997} and \citet{HorvathRice2014}.  However, the vast majority of this ever-growing literature has focused on the offline changpoint problem, where, after the entire data stream is observed, the statistician is asked to identify any changepoints retrospectively.  For univariate, offline changepoint estimation, state-of-the-art methods include the Pruned Exact Linear Time method (PELT) \citep{KFE2012}, Narrowest-Over-Threshold (NOT) \citep{BCF2019}, Simultaneous Multiscale Changepoint Estimator (SMUCE) \citep{FMS2014} and $\ell_0$-penalisation \citep{WYR2018}, while work on multivariate and high-dimensional offline changepoints includes the double CUSUM method of \citet{Cho2016}, the \texttt{inspect} algorithm of \citet{WangSamworth2018}, as well as \citet{EnikeevaHarchaoui2019}, \citet{LGS2019} and \citet{PYWR2019}.

Despite this rich literature on offline changepoint problems, it is the online version of the problem that is arguably the more important for many applications: one would like to be able to detect a change as soon as possible after it has occurred.  Of course, one option here is to apply an offline method after seeing every new observation (or batch of observations).  However, this is unlikely to be a successful strategy: not only is there a difficult and highly dependent multiple testing issue to handle when using the method repeatedly on data sets of increasing size (see also \citet{CSW1996} for further discussion of this point), but moreover, the storage and running time costs may frequently be prohibitive.

In this work, we are interested in algorithms for detecting changepoints in high-dimensional data that are observed sequentially.  In order to avoid the trap mentioned in the previous paragraph and ensure that any methods we consider can be applied to large data streams, we will focus our attention on \emph{online algorithms}. By this, we mean that the computational complexity for processing a new observation, as well as the storage requirements, depend only on the number of bits needed to represent the new observation\footnote{For the purpose of this definition, we ignore the errors in rounding real numbers to machine precision.  Thus, when we later work with observations having Gaussian (or other absolutely continuous) distributions, we do not distinguish between these distributions and quantised versions where the data have been rounded to machine precision.}. Importantly, they are not allowed to depend on the number of previously observed data points. This turns out to be a very stringent requirement, in the sense that finding online algorithms with good statistical performance is typically extremely challenging. Online algorithms must necessarily store only compact summaries of the historical observations, so the class of all possible procedures is severely restricted.

To set the scene for our contributions, let $X_1, X_2, \ldots$ be a sequence of independent random vectors in $\mathbb{R}^p$.  Assume that for some unknown, deterministic time $z \in \mathbb{N} \cup \{0\}$, the sequence is generated according to
\begin{equation}
\label{Eq:DataGeneration}
X_1,\ldots,X_z  \sim \mathcal{N}_p(\mu_{-}, I_p)  \quad \text{and} \quad X_{z+1}, X_{z+2}, \ldots \sim \mathcal{N}_p(\mu_{+}, I_p),
\end{equation}
for some $\mu_-, \mu_+\in\mathbb{R}^p$. When $\mu_{+}\neq \mu_{-}$, we say that there is a changepoint at time $z$.  In many applications, such as in industrial quality control where the distribution of relevant properties of goods in a manufacturing process under regular conditions may be well understood, we may assume that the mean before the change is known (or at least can be estimated to high accuracy using historical data). However, the vector of change, $\theta:=\mu_+-\mu_-$, is typically unknown. Thus, for simplicity, we will work in the setting where $\mu_- = 0$ and $\mu_+ = \theta$. Let $\mathbb{P}_{z,\theta}$ denote the joint distribution of $(X_n)_{n=1}^\infty$ under~\eqref{Eq:DataGeneration} and $\mathbb{E}_{z,\theta}$ the expectation under this distribution. Note that when $\theta = 0$, the joint distribution of the data does not depend on $z$, and we therefore let $\mathbb{P}_0 = \mathbb{P}_{z,0}$ denote this joint distribution (with corresponding expectation $\mathbb{E}_0$).  We will then say that the data is generated \emph{under the null}.  By contrast, if $\theta \neq 0$, we will say that the data is generated \emph{under the alternative}, though we emphasise that in fact the alternative is composite, being indexed by $z \in \mathbb{N}$ and $\theta \in \mathbb{R}^p \setminus \{0\}$.  In practice, in order for our procedure to have uniformly non-trivial power, it will be necessary to work with a subset of the alternative hypothesis parameter space that is well-separated from the null, in the sense that the $\ell_2$-norm of the vector of mean change, $\vartheta:=\|\theta\|_2$, is at least a known lower bound $\beta > 0$.

A sequential changepoint procedure is an extended stopping time\footnote{A random variable $\tau$ taking values in $\mathbb{N} \cup\{\infty\}$ is an \emph{extended stopping time} with respect to the filtration $(\mathcal{F}_n)_{n \in \mathbb{N}}$, if $\{\tau = n\} \in \mathcal{F}_n$ for all $n \in \mathbb{N}$.} $N$ (with respect to the natural filtration) taking values in $\mathbb{N} \cup \{\infty\}$.  Equivalently, we can think of it as a family of $\{0,1\}$-valued estimators $(\hat{H}_n)_{n=1}^\infty$, where $\hat{H}_n = \hat{H}_n(X_1,\ldots,X_n)$, and where the sequence is increasing in the sense that $\hat{H}_m(X_1,\ldots,X_m) \leq \hat{H}_n(X_1,\ldots,X_n)$ for $m \leq n$.  Here, the correspondence arises from $\hat{H}_n = \mathbbm{1}_{\{N \leq n\}}$ and $N = \inf\{n \in \mathbb{N}:\hat{H}_n = 1\}$, with the usual convention that $\inf \emptyset := \infty$.

We measure the performance of a sequential changepoint procedure via its responsiveness subject to a given upper bound on the false alarm rate, or equivalently, a lower bound on the average run length in the absence of change.  Specifically, following the concepts introduced by~\citet{Lorden1971}, we define the \emph{patience} (this is sometimes referred to as the average run length under the null or average run length to false alarm in the literature) of a sequential changepoint procedure $N$ to be $\mathbb{E}_0(N)$, and its \emph{worst-case response delay} (likewise, this is sometimes referred to as the worst-worst-case average detection delay) to be
\[
\bar{\mathbb{E}}_\theta^{\mathrm{wc}}(N):=\sup_{z\in\mathbb{N}} \esssup \mathbb{E}_{z, \theta}\bigl\{(N-z)\vee 0 \mid X_1, \ldots, X_{z}\bigr\}.
\]
While controlling the worst-case response delay provides a very strong theoretical guarantee of the average detection delay of the procedure, even under the worst possible pre-change data sequence, obtaining a good bound for this quantity is often difficult.  We therefore also consider the average-case response delay, or simply the \emph{response delay} of a procedure $N$, defined as
\[
\bar{\mathbb{E}}_\theta(N):=\sup_{z\in\mathbb{N}} \mathbb{E}_{z, \theta}\bigl\{(N-z)\vee 0\bigr\}.
\]
We note that $\bar{\mathbb{E}}_\theta(N)\leq \bar{\mathbb{E}}_\theta^{\mathrm{wc}}(N)$.  A good sequential changepoint procedure should have small worst- and average-case response delays, uniformly over the relevant class of alternatives $\{\mathbb{P}_{z,\theta}: (z,\theta) \in (\mathbb{N} \cup \{0\}) \times \mathbb{R}^p, \|\theta\|_2\geq \beta\}$, subject to its patience being at least some suitably large, pre-determined $\gamma > 0$.  Finally, as mentioned above, we are interested in sequential changepoint procedures that are online, so that the computational complexity per additional observation should be a function of $p$ only.

Our main contribution in this work is to propose, in Section~\ref{Sec:Method}, a new algorithm called \texttt{ocd} (short for \underline{o}nline \underline{c}hangepoint \underline{d}etection), for high-dimensional, online changepoint detection in the above setting.  The procedure works by performing likelihood ratio tests against simple alternatives of different scales in each coordinate, and then aggregating test statistics across scales and coordinates for changepoint detection.  The \texttt{ocd} algorithm has worst-case computational complexity $O\bigl(p^2 \log (ep)\bigr)$ per new observation, so satisfies our requirement for being an online algorithm.  In fact, as we explain in Section~\ref{SubSec:ocd}, the algorithmic complexity is often even better than this.  Moreover, as we illustrate in Section~\ref{Sec:Simulation}, it has extremely effective empirical performance.  In terms of theoretical guarantees, it turns out to be more convenient to analyse a slight variant of our initial algorithm, which we refer to as \texttt{ocd}$'$.  This has the same order of computational complexity per new observation as \texttt{ocd}, but enables us to ensure that whenever we are yet to declare that a change has occurred, only post-change observations contribute to the running test statistics.  In practice, the original \texttt{ocd} algorithm also appears to have this property for typical pre-change sequences, and we argue heuristically that there is a sense in which it is more efficient than \texttt{ocd}$'$ by a factor of at most~2.

Our theoretical analysis in Section~\ref{Sec:Theory} initially considers separately versions of the \texttt{ocd}$'$ algorithm best tuned towards settings where the vector $\theta$ of change is dense, and where it is sparse in an appropriate sense.  We then present results for a combined, adaptive procedure that seeks the best of both worlds.  In all cases, the appropriate version of \texttt{ocd$'$} has guaranteed patience, at least at the desired nominal level.  In the (small-change) regime of primary interest, and when $\vartheta$ is of the same order as $\beta$, the response delay of \texttt{ocd$'$} is of order at most $\sqrt{p}/\vartheta^2$ in the dense case, up to a poly-logarithmic factor; this can be improved to order $s/\vartheta^2$, again up to a poly-logarithmic factor, when the effective sparsity of $\theta$ is $s < \sqrt{p}$. 

As alluded to above, there is a paucity of prior literature on multivariate, online changepoint problems, though exceptions include \citet{TRBK2006}, \citet{Mei2010} and \citet{ZWZJ2015}.   These works focus either on the case where both the pre- and post-change distributions are exactly known, or where, for each coordinate, both the sign and a lower bound on the magnitude of change, are known in advance.  A number of methods have also been proposed that involve scanning a moving window of fixed size for changes \citep{XieSiegmund2013,SohChandrasekaran2017,Chan2017}.  Such methods can be effective when the signal-to-noise ratio is large enough that the change can be detected within the prescribed window, but may experience excessive response delay in other cases.  Of course, the window size may be increased to compensate, but this correspondingly increases the computational complexity and storage requirements, so allowing the window size to vary with the signal strength would fail to satisfy our definition of an online algorithm.  We also mention that online changepoint detection has been studied in the econometrics literature, where the problem is often referred to as that of monitoring structural breaks \citep{CSW1996,LHK2000,ZLKH2005}.  These works have studied low-dimensional regression settings, and asymptotic theory has been provided on the probability of eventual declaration of change.

Numerical results illustrate the performance of our \texttt{ocd} algorithm in Section~\ref{Sec:Simulation}.  Proofs of our main results are given in Section~\ref{Sec:Proofs}. All the auxiliary lemmas and their proofs are provided in Section~\ref{Sec:Auxiliary}.



\subsection{Notation}
We write $\mathbb{N}_0$ for the set of all non-negative integers. For $d \in \mathbb{N}$, we write $[d]:=\{1, \ldots, d\}$. Given $a, b \in \mathbb{R}$, we denote $a \vee b := \max(a, b)$ and $a \wedge b := \min(a, b)$. For a set $S$, we use $\mathbbm{1}_S$ and $|S|$ to denote its indicator function and cardinality respectively. For a real-valued function $f$ on a totally ordered set $S$, we write $\sargmax_{x \in S} f(x) := \min \argmax_{x\in S} f(x)$, the smallest maximiser of $f$ in $S$. For a vector $v = \bigl(v^1, \ldots, v^M\bigr)^\top \in \mathbb{R}^M$, we define $\|v\|_0 := \sum_{i=1}^M \mathbbm{1}_{\{v^i \neq 0\}}$, $\|v\|_2:=\bigl\{\sum_{i=1}^{M} (v^i)^2\bigr\}^{1/2}$ and $\|v\|_\infty := \max_{i \in [M]} |v^i|$. In addition, we define $v^{-j} := (v^1,\ldots,v^{j-1},v^{j+1},\ldots,v^M)^\top \in \mathbb{R}^{M-1}$. For a matrix $A=(A^{i,j}) \in \mathbb{R}^{d_1\times d_2}$ and $j \in [d_2]$, we write $A^{\bdot, j} := \bigl(A^{1, j}, \ldots, A^{d_1, j}\bigr)^\top \in \mathbb{R}^{d_1}$ and $A^{-j, j} := \bigl(A^{1, j}, \ldots, A^{j-1, j}, A^{j+1, j}\ldots, A^{d_1, j}\bigr)^\top \in \mathbb{R}^{d_1 - 1}$. We use $\Phi(\cdot)$ and $\phi(\cdot)$ to denote the distribution function and density function of the standard normal distribution respectively. For two real-valued random variables $U$ and $V$, we write $U \geq_{\mathrm{st}} V$ if $\mathbb{P}(U \leq x) \leq \mathbb{P}(V \leq x)$ for all $x\in \mathbb{R}$. We adopt the convention that an empty sum is $0$.

\section{An online changepoint procedure}
\label{Sec:Method}
\subsection{The \texttt{ocd} algorithm}
\label{SubSec:ocd}
In this section, we describe our online changepoint procedure, \texttt{ocd}, in more detail. As mentioned in the introduction, the procedure aggregates likelihood ratio test statistics against simple alternatives of different scales in different coordinates.  For $i \in [n]$ and $j \in [p]$, we write $X_i^j$ for the $j$th coordinate of~$X_i$.  If we want to test a null of $\mathcal{N}(0,1)$ against a simple post-change alternative distribution of $\mathcal{N}(b,1)$ for some $b \neq 0$ in coordinate $j\in[p]$, by \citet{Page1954}, the optimal online changepoint procedure is to declare that a change has occurred by time $n$ when the test statistic
\begin{equation}
\label{Eq:MaxLR}
R_{n,b}^j := \max_{0\leq h\leq n} \sum_{i=n-h+1}^n b(X_i^j - b/2)
\end{equation}
exceeds a certain threshold. Note that $\sum_{i=n-h+1}^n b(X_i^j - b/2)$ can be viewed as the likelihood ratio test statistic between the null and this simple alternative using the tail sequence $X_{n-h+1},\ldots,X_n$. Thus $R_{n,b}^j$ can be regarded as the most extreme of these likelihood ratio statistics, over all possible starting points for the tail sequence.  Write 
\begin{equation}
\label{Eq:TailLength}
t_{n,b}^j := \sargmax_{0\leq h\leq n} \sum_{i=n-h+1}^n b(X_i^j - b/2)
\end{equation}
for the length of the tail sequence in which the associated likelihood ratio statistic (in the $j$th coordinate) is maximised. One way to aggregate across the $p$ coordinates would be to use $\sum_{j=1}^p R_{n,b}^j$ as a test statistic. However, this approach is not ideal for two reasons. Firstly, the exact distribution of the tail likelihood ratio statistic $R_{n,b}^j$ is hard to obtain, making it difficult to analyse the aggregated statistic under the null. More importantly, this aggregated statistic uses the same simple alternative $\mathcal{N}(b,1)$ in all coordinates, and so even after varying the magnitude of $b$, it is only effective against a very limited set of alternative distributions in $\{\mathbb{P}_{z,\theta}: z\in\mathbb{N}, \|\theta\|_2\geq \beta\}$, namely those for which the change is of very similar magnitude in all coordinates.  In order to overcome these problems, our procedure uses the coordinate-wise statistics $(R_{n,b}^j:j\in[p])$, which we call `diagonal statistics', to detect changes that have a large proportion of their signal concentrated in one coordinate.  To detect denser changes, for each $j\in[p]$, we also compute tail partial sums of length $t_{n,b}^j$ in all other coordinates $j'\neq j$, given by
\[
A_{n,b}^{j',j}:=\sum_{i=n-t_{n,b}^j+1}^n X_i^{j'},
\]
and aggregate them to form an `off-diagonal statistic' anchored at coordinate $j$.  Note that the number of summands in $A_{n,b}^{j',j}$ depends only on the observed data in the $j$th coordinate, and not on the data being aggregated in the $j'$th coordinate.  These off-diagonal statistics are used to detect changes whose signal is not concentrated in a single coordinate. Intuitively, if a change has occurred and $\theta^j/b \geq 1$, then we can expect the tail length in coordinate $j$ to be roughly of order $n-z$ for sufficiently large $n$, and this will ensure that the off-diagonal statistic anchored at coordinate $j$ is close to the generalised likelihood ratio test statistic between the null and the composite alternative $\{\mathbb{P}_{z,\theta}:\|\theta\|_2 \neq 0\}$. If, in addition, a non-trivial proportion of the signal is contained in coordinates $[p]\setminus \{j\}$, then this statistic will be powerful for detecting the change. 

The full description of the \texttt{ocd} procedure is given in Algorithm~\ref{Alg:changepoint2}. Note that for notational simplicity, we have suppressed the time dependence of many variables as they are updated recursively in the algorithm. In the following, when necessary, we will make this dependence explicit by writing $A_{n,b}, t_{n,b}, Q_{n, b}, S^{\mathrm{diag}}_{n}$ and $S^{\mathrm{off}}_n$ for the relevant quantities at the end of the $n$th iteration of the repeat loop. 

By Lemma~\ref{Lemma:DefRProcess}, $bA_{n,b}^{j,j} - b^2 t_{n,b}^j/2$, as defined in the algorithm, is equal to the quantity $R_{n,b}^j$ defined in~\eqref{Eq:MaxLR} (we will also suppress its $n$ dependence when it is clear from the context). Moreover, by Lemma~\ref{Lemma:AltDefTailLength}, the two definitions of $t_{n,b}^j$ from Algorithm~\ref{Alg:changepoint2} and~\eqref{Eq:TailLength} coincide. In the algorithm, we allow $b$ to range over the (signed) dyadic grid $\mathcal{B}\cup \mathcal{B}_0$, since the maximal signal strength in individual coordinates, $\|\theta\|_\infty$, can range from $\vartheta/\sqrt{p}$ to $\vartheta$. In this way, the algorithm automatically adapts to different signal strengths in each coordinate.  Here, the inclusion of $\mathcal{B}_0$ and the extra logarithmic factors in the denominators of elements of $\mathcal{B}\cup\mathcal{B}_0$ appear due to technical reasons in the theoretical analysis of the algorithm.

Algorithm~\ref{Alg:changepoint2} uses $S^{\mathrm{diag}}$ and $S^{\mathrm{off}}$ to aggregate diagonal and off-diagonal statistics respectively as mentioned above, and declares that a change has occurred as soon as either of these quantities exceeds its own pre-determined threshold.  As mentioned previously, $S^{\mathrm{diag}}$ tracks the maximum of $R_{b}^j$ over all scales $b$ and coordinates $j$. 
Before introducing $S^{\mathrm{off}}$, we first discuss the off-diagonal statistics $Q_b^j$ in Algorithm~\ref{Alg:changepoint2}, which are $\ell_2$ aggregations of normalised tail sums $A_b^{j', j}/\sqrt{t_b^j\vee 1}$, each hard-thresholded at level $a$.  The hard thresholding level can be chosen to detect dense or sparse signals $\theta$; in the sparse case a non-zero $a$ facilitates an aggregation that aims to exclude coordinates with negligible change (thereby reducing the variance of the normalised tail sums).  
Finally, $S^{\mathrm{off}}$ is computed as the maximum of the $Q_{b}^j$ over all anchoring coordinates $j \in [p]$ and scales $b\in\mathcal{B}$.

Although the off-diagonal statistics described in the previous paragraph are effective for detecting changes when the signal sparsity is known, it is desirable to the practitioner to have a combined procedure that adapts to the sparsity level.  This may be computed straightforwardly by tracking $S^{\mathrm{off}}$ for $a=a^{\mathrm{dense}}$ and $a=a^{\mathrm{sparse}}$, as well as $S^{\mathrm{diag}}$, and declaring a change when any of these three statistics exceeds a suitable threshold. Figure~\ref{fig:Null_All_trigger} illustrates the performance of this adaptive procedure, together with the time evolution of normalised versions of all three statistics tracked, in synthetic datasets both with and without a change. This adaptive procedure is analysed theoretically in Section~\ref{SubSec:AdaptiveTheory} and empirically in Section~\ref{Sec:Simulation}.

\begin{figure}[htbp]
	\begin{center}
		\begin{tabular}{cc}
			\includegraphics[width=0.45\textwidth]{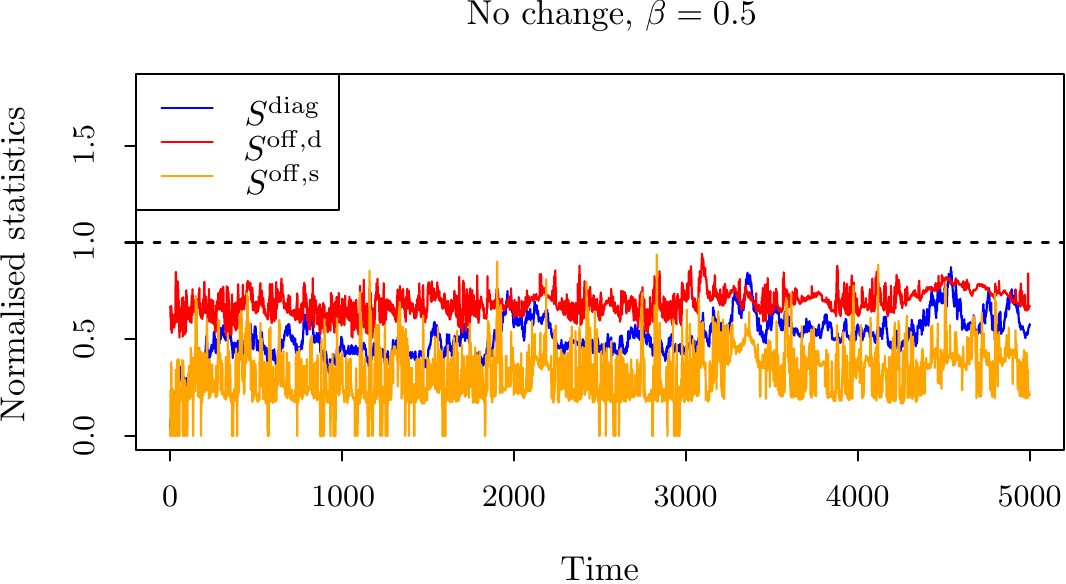}
			&
			\includegraphics[width=0.45\textwidth]{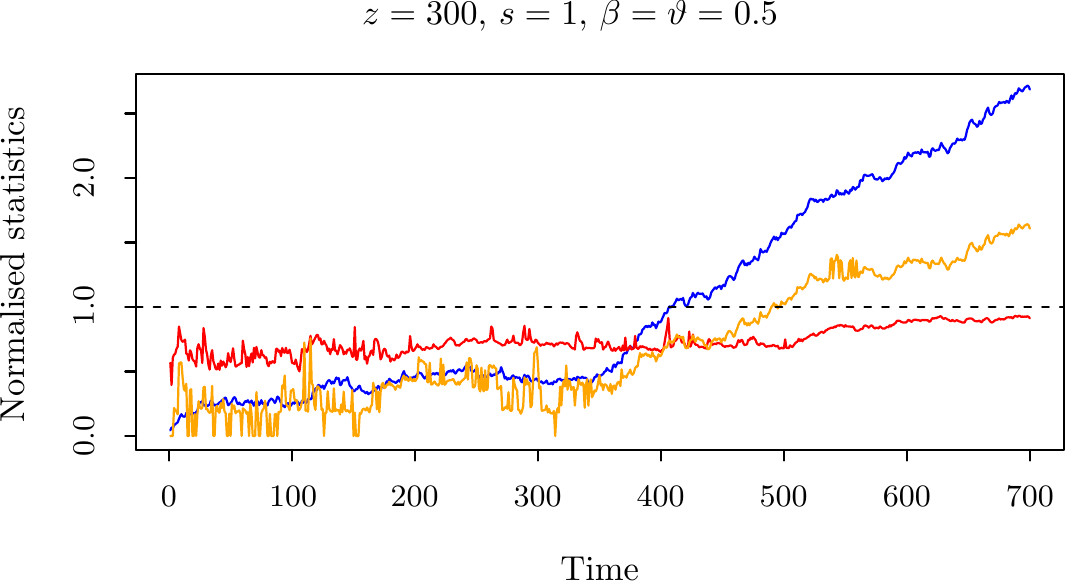} \\ \\
			\includegraphics[width=0.45\textwidth]{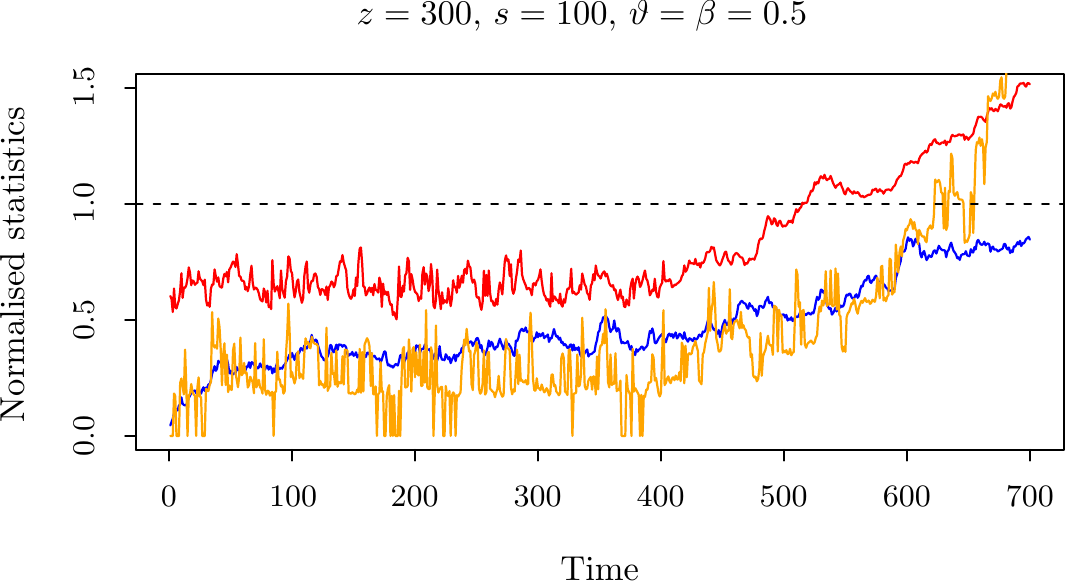}
			&
			\includegraphics[width=0.45\textwidth]{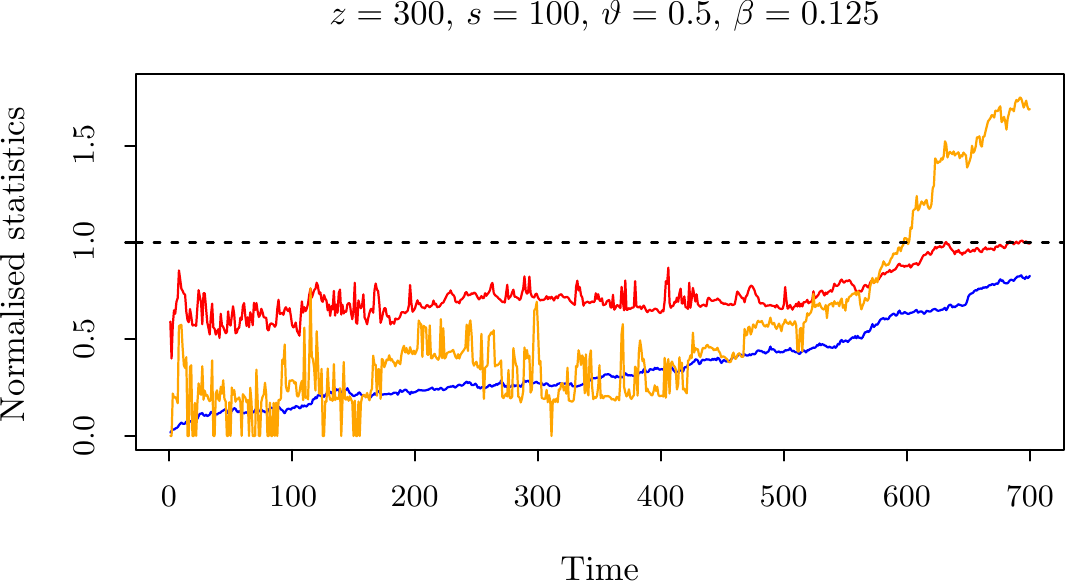}
		\end{tabular}
	\caption{\label{fig:Null_All_trigger}Behaviour of the three normalised statistics in $\texttt{ocd}$ under the null and under the alternative with different signal strength, sparsity level and assumed lower bound. A change is declared as soon as one of these three normalised statistics exceeds $1$.  The data were generated in the top-left panel according to $\mathbb{P}_0$, and, in the other panels, according to $\mathbb{P}_{z,\theta}$, with $p=100$, $z=300$ and $\theta = \vartheta U$, where $U$ is uniformly distributed on the union of all $s$-sparse unit spheres in $\mathbb{R}^p$ (see Section~\ref{SubSec:ocdnumerics} for a more detailed description).}
	\end{center}
\end{figure}


The \texttt{ocd} procedure satisfies our definition of an online algorithm. Indeed, for each new observation $X_n$, \texttt{ocd} updates $t_{n,b} \in \mathbb{R}^p$ and $A_{n,b} \in \mathbb{R}^{p\times p}$ for $O\bigl(\log (ep)\bigr)$ different values of~$b$. It then computes $S_{n}^{\mathrm{diag}}$ and $S_{n}^{\mathrm{off}}$ via $A_{n,b}$. These steps require $O\bigl(p^2 \log (ep)\bigr)$ operations. Moreover, the total storage used is $O\bigl(p^2 \log (ep)\bigr)$ throughout the algorithm.

In fact, the computational complexity of \texttt{ocd} can often be reduced, because typically $\mathcal{T} := \{t_b^j:j\in [p],b \in \mathcal{B}\}$ has cardinality much less than $p|\mathcal{B}|$ (which is the worst case, when all elements are distinct).  Correspondingly, at each time step, we need only store the $p \times |\mathcal{T}|$ matrix $(B^{k,t})_{k \in [p], t \in \mathcal{T}}$ given by $B^{k,t_b^j} := A_b^{k,j}$, resulting in an improved per-iteration computational complexity and storage for \texttt{ocd} of $O(p|\mathcal{T}|)$.  For simplicity of exposition, we have not presented this computational speed-up in Algorithm~\ref{Alg:changepoint2}, and it appears to be difficult to provide theoretical guarantees on $|\mathcal{T}|$.  Nevertheless we have implemented the algorithm in this form in the \texttt{R} package \texttt{ocd} \citep{CWS2020}, and have found it to provide substantial computational savings in practice.  

\begin{algorithm}[htbp!]
	\KwIn{$X_1, X_2 \ldots \in \mathbb{R}^p$ observed sequentially, $\beta>0$, $a\geq 0$,  $T^{\mathrm{diag}}>0$ and $T^{\mathrm{off}}>0$}
	\KwSet{$\mathcal{B} = \Bigl\{\pm\frac{\beta}{\sqrt{2^\ell \log_2(2p)}}: \ell = 0,\ldots,\lfloor \log_2 p\rfloor\Bigr\}$, $\mathcal{B}_0=\Bigl\{ \pm\frac{\beta}{\sqrt{2^{\lfloor \log_2 p\rfloor+1} \log_2(2p)}} \Bigr\}$, $n = 0$, $A_b = \mathbf{0} \in \mathbb{R}^{p\times p}$ and $t_b = 0\in\mathbb{R}^p$ for all $b \in \mathcal{B}\cup \mathcal{B}_0$}
	\Repeat{$S^{\mathrm{diag}} \geq T^{\mathrm{diag}}$ \textup{or} $S^{\mathrm{off}} \geq T^{\mathrm{off}}$}{
		$n \leftarrow n+1$\\
		observe new data vector $X_n$ \\
		\For{$(j, b) \in [p]\times (\mathcal{B} \cup \mathcal{B}_0)$}{
			$t^j_b \leftarrow t^j_b +1$ \\ 
			$A^{\bdot,j}_b \leftarrow A^{\bdot,j}_b + X_n$ \\
			\If{$bA^{j,j}_{b} - b^2t^j_b/2 \leq 0$}{                 
				$t^j_b \leftarrow 0$ and  $A^{\bdot,j}_b \leftarrow 0$
			} 
			compute $Q^j_b \leftarrow \sum_{j' \in [p]: j' \neq j}\frac{(A^{j', j}_b)^2}{t^j_b \vee 1} \mathbbm{1}_{\bigl\{|A^{j', j}_b| \geq a\sqrt{t_b^j}\bigr\}}$
		}
		$S^{\mathrm{diag}}\leftarrow \max_{(j, b) \in [p]\times (\mathcal{B} \cup \mathcal{B}_0)} \bigl(b A^{j,j}_b - b^2t^j_b/2\bigr)$\\ 
		$S^{\mathrm{off}} \leftarrow \max_{(j,b)\in[p]\times \mathcal{B}} Q_{b}^j$
	}
	\KwOut{$N=n$}
	\caption{Pseudo-code of the \texttt{ocd} algorithm}
	\label{Alg:changepoint2}
\end{algorithm}

\subsection{A slight variant of \texttt{ocd}}
\label{SubSec:ocdvariant}

While the \texttt{ocd} algorithm performs very well numerically, it turns out to be easier theoretically to analyse a slight variant, which we call \texttt{ocd}$'$, and describe in Algorithm~\ref{Alg:changepoint3}. Again, we have suppressed the time dependence $n$ of many variables including $\tau_{n,b}, \tilde{\tau}_{n, b}, \Lambda_{n, b}$ and $\tilde{\Lambda}_{n, b}$ in the algorithm. The main difference between these two algorithms is that in \texttt{ocd}$'$, the off-diagonal statistics $Q_b^j$ are computed using tail partial sums of length $\tau_b^j$ instead of $t_b^j$. These new tail partial sums are recorded in $\Lambda_{b}\in\mathbb{R}^{p\times p}$.

By Lemma~\ref{Lemma:AuxiliaryTailStats}, we always have
\begin{equation}
\label{Eq:ttaubound}
t_b^j/2 \leq \tau_b^j < 3t_b^j/4
\end{equation}
whenever $t_b^j \geq 2$. In this sense, the tail sample size used by \texttt{ocd}$'$ is smaller than that of \texttt{ocd} by a factor of at most $2$. The benefit of using a shorter tail in \texttt{ocd}$'$ is that when $n$ exceeds a known, deterministic threshold, we can be sure that whenever we have not declared that a change has occurred by time $z$, the tail partial sum consists exclusively of post-change observations. In practice, we observe that even in Algorithm~\ref{Alg:changepoint2}, the tail lengths $t_{z,b}^j$ at the changepoint are generally very short for many coordinates, so the inclusion of a few pre-change observations in the tail partial sum calculation does not significantly affect the efficacy of the changepoint detection procedure. 
The practical performance of Algorithm~\ref{Alg:changepoint2} is statistically more efficient than Algorithm~\ref{Alg:changepoint3} in many settings by a factor of between $4/3$ and $2$, as suggested by~\eqref{Eq:ttaubound}.  By construction, $\tau_b^j$ and $\Lambda_b^{\bdot,j}$ are computable online, through auxiliary variables $\tilde\tau_b^j$ and $\tilde\Lambda_b^{\bdot,j}$.  Indeed, Algorithm~\ref{Alg:changepoint3} is also an online algorithm, with overall computational complexity per observation and storage remaining at $O\bigl(p^2\log (ep)\bigr)$ in the worst case; similar computational improvements to those mentioned for \texttt{ocd} at the end of Section~\ref{SubSec:ocd} are also possible here. 

\begin{algorithm}[htbp!]
	\KwIn{$X_1, X_2 \ldots \in \mathbb{R}^p$ observed sequentially, $\beta>0$, $a \geq 0$, $T^{\mathrm{diag}}>0$ and $T^{\mathrm{off}}>0$.}
	\KwSet{$\mathcal{B} = \Bigl\{\pm\frac{\beta}{\sqrt{2^\ell \log_2(2p)}}: \ell = 0,\ldots,\lfloor \log_2 p\rfloor\Bigr\}$, $\mathcal{B}_0=\Bigl\{ \pm\frac{\beta}{\sqrt{2^{\lfloor \log_2 p\rfloor+1} \log_2(2p)}} \Bigr\}$, $n=0$, $A_b = \Lambda_b = \tilde{\Lambda}_b = \mathbf{0} \in \mathbb{R}^{p\times p}$ and $t_b = \tau_b = \tilde{\tau}_b = 0\in\mathbb{R}^p$ for all $b \in \mathcal{B}\cup \mathcal{B}_0$}
	\Repeat{$S^{\mathrm{diag}} \geq T^{\mathrm{diag}}$ \textup{or} $S^{\mathrm{off}} \geq T^{\mathrm{off}}$}{
		$n \leftarrow n+1$\\
		observe new data vector $X_n$ \\
		\For{$(j, b) \in [p]\times (\mathcal{B} \cup \mathcal{B}_0)$}{
			$t^j_b\leftarrow t^j_b +1 \;\; \text{and}\;\; A^{\bdot,j}_b \leftarrow  A^{\bdot,j}_b + X_n$\\
			\medskip
			set $\delta = 0$ if $t_b^j$ is a power of 2 and $\delta = 1$ otherwise.\\
			$\tau_b^j \leftarrow \tau_b^j\delta + \tilde \tau_b^j (1-\delta) + 1 \;\; \text{and} \;\; \Lambda_b^{\bdot, j} \leftarrow \Lambda_b^{\bdot, j}\delta + \tilde \Lambda_b^{\bdot, j} (1-\delta) + X_n$\\
			$\tilde \tau_b^j \leftarrow (\tilde\tau_b^j+1)\delta \;\; \text{and} \;\; \tilde \Lambda_b^{\bdot, j} \leftarrow (\tilde \Lambda_b^{\bdot, j} + X_n) \delta.$\\
			\medskip
			\If{$bA^{j,j}_{b} - b^2t^j_b/2 \leq 0$}{                 
				$t_b^j \leftarrow \tau_b^j \leftarrow \tilde{\tau}_b^j \leftarrow 0$\\
				$A_b^{\bdot,j} \leftarrow \Lambda_b^{\bdot,j} \leftarrow \tilde{\Lambda}_b^{\bdot,j} \leftarrow 0$
			} 
			compute $Q^j_b \leftarrow \sum_{j' \in [p]: j' \neq j}\frac{(\Lambda^{j', j}_b)^2}{\tau^j_b \vee 1} \mathbbm{1}_{\bigl\{|\Lambda^{j', j}_b| \geq a\sqrt{\tau_b^j}\bigr\}} $
		}
		$S^{\mathrm{diag}} \leftarrow \max_{(j, b) \in [p]\times (\mathcal{B} \cup \mathcal{B}_0)} \bigl(b A^{j,j}_b - b^2t^j_b/2\bigr)$\\ 
		$S^{\mathrm{off}} \leftarrow \max_{(j,b)\in[p]\times \mathcal{B}} Q_{b}^j$\\
	}
	\KwOut{$N=n$}
	\caption{Pseudo-code of the \texttt{ocd}$'$ algorithm, a slight variant of \texttt{ocd}}
	\label{Alg:changepoint3}
\end{algorithm}

\section{Theoretical analysis}
\label{Sec:Theory}
As mentioned in Section~\ref{Sec:Method}, the input $a$ in Algorithms~\ref{Alg:changepoint2} and~\ref{Alg:changepoint3} allows users to detect changepoints of different sparsity levels. More precisely, for any $\theta \in \mathbb{R}^p$, we have by Lemma~\ref{Lemma:LogGridSignal} that there exists a smallest $s(\theta) \in\{2^0, 2^1 ,\ldots,2^{\lfloor\log_2 p\rfloor}\}$ such that the set
\[
\mathcal{S}(\theta):= \biggl\{j\in[p]: |\theta^j| \geq \frac{\|\theta\|_2}{\sqrt{s(\theta)\log_2(2p)}}\biggr\}
\]
has cardinality at least $s(\theta)$. On the other hand, we also have $|\mathcal{S}(\theta)| \leq s(\theta)\log_2(2p)$. We call $s(\theta)$ the \emph{effective sparsity} of the vector $\theta$ and $\mathcal{S}(\theta)$ its \emph{effective support}. Intuitively, the sum of squares of coordinates in the effective support of $\theta$ has the same order of magnitude as $\|\theta\|_2^2$, up to logarithmic factors. Moreover, if $\theta$ is an $s$-sparse vector in the sense that $\|\theta\|_0\leq s$, then $s(\theta) \leq s$, and the equality is attained when, for example, all non-zero coordinates have the same magnitude. 

In this section, we initially analyse the theoretical performance of Algorithm~\ref{Alg:changepoint3} for two different choices of $a$ in $S^{\mathrm{off}} = S^{\mathrm{off}}(a)$, namely $a=0$ and $a = \sqrt{8\log(p-1)}$.  We then present our combined, adaptive procedure and its performance guarantees.


Define $N^{\mathrm{diag}}:= \inf\{n: S^{\mathrm{diag}}_{n} \geq T^{\mathrm{diag}}\}$ and $N^{\mathrm{off}} = N^{\mathrm{off}}(a) :=\inf\{n: S^{\mathrm{off}}_{n}(a) \geq T^{\mathrm{off}} \}$. Then the stopping time for our changepoint detection procedure is simply $N = N(a) = N^{\mathrm{diag}}  \wedge N^{\mathrm{off}}(a)$. 

\subsection{Dense case}
\label{SubSec:DenseAnalysis}
Here, we analyse the changepoint detection procedure $N=N(0)$, which, as we will see, is most suitable for detecting dense mean changes in the sense that $s(\theta)\geq \sqrt{p}$ (though we do not assume this in our theory). In this case, when $p \geq 2$ and conditionally on $\tau_b^j$, the quantity $Q_b^j$ follows a chi-squared distribution with $p-1$ degrees of freedom under the null, provided that $\tau_b^j$ is positive (When $p=1$, we have that $Q_b^j = 0$ for all $j \in [p]$ and $b \in \mathcal{B}$, so $S^{\mathrm{off}} = 0$ and the off-diagonal statistic never triggers the declaration of a change.  Similarly, if $p \geq 2$ but $\tau_{n,b}^j = 0$, then we also have $Q_{n,b}^j = 0$.).  Motivated by the chi-squared tail bound of \citet[Lemma~1]{LaurentMassart2000}, we choose a threshold of the form
\begin{equation}
\label{Eq:Toff}
T^{\mathrm{off}} := p-1+ \tilde{T}^{\mathrm{off}} + \sqrt{2(p-1)\tilde{T}^{\mathrm{off}}} =: \psi(\tilde{T}^{\mathrm{off}}),
\end{equation}
say, for some $\tilde{T}^{\mathrm{off}} > 0$.

The following theorem provides control of the patience of \texttt{ocd}$'$.
\begin{thm}
	\label{Thm:Null} 
	Let $X_1,X_2,\ldots $ be generated according to $\mathbb{P}_0$. For any $\gamma\geq 1$, let $(X_t)_{t\in\mathbb{N}}$, $\beta>0$, $a = 0$, $T^{\mathrm{diag}}= \log\{16p\gamma\log_2(4p)\}$ and $T^{\mathrm{off}}= \psi\bigl(\tilde{T}^{\mathrm{off}})$ with $\tilde{T}^{\mathrm{off}} = 2\log\{16p\gamma\log_2(2p)\}\bigr)$ be the inputs of Algorithm~\ref{Alg:changepoint3}, with corresponding output $N$. Then $\mathbb{E}_0(N) \ge \gamma$.
\end{thm}

We note that either of the two statistics $S^{\mathrm{diag}}$ and $S^{\mathrm{off}}$ may trigger a false alarm under the null. The two threshold levels $T^{\mathrm{diag}}$ and $T^{\mathrm{off}}$ are chosen so that $\mathbb{E}_0(N^{\mathrm{diag}})$ and $\mathbb{E}_0(N^{\mathrm{off}})$ have comparable upper bounds.  We also remark that although Theorem~\ref{Thm:Null} as stated only controls the expected value of $N$ under the null, careful examination of the proof reveals that we can also control $\mathbb{P}_0(N \leq m)$ for every $m \in \mathbb{N}$.  More precisely, from~\eqref{Eq:Noffprob_d} and~\eqref{Eq:NdiagProb} in the proof, we can deduce that
\[
  \mathbb{P}_0(N \leq m) \leq \frac{m}{4\gamma}
\]
for every $m \in \mathbb{N}$.  The same bound holds for our other patience control results below, though we omit formal statements for brevity.


Our next result controls the response delay of \texttt{ocd}$'$ in both worst-case and average senses.
\begin{thm}
	\label{Thm:DenseMain}
	Assume that $X_1,X_2,\ldots$ are generated according to $\mathbb{P}_{z,\theta}$ for some $z$ and $\theta$ such that $\|\theta\|_2 = \vartheta \geq \beta > 0$ and that $\theta$ has an effective sparsity of $s:=s(\theta) $.  Then there exists a universal constant $C > 0$, such that the output $N$ from Algorithm~\ref{Alg:changepoint3}, with inputs $(X_t)_{t\in\mathbb{N}}$, $\beta \in \mathbb{R}^p$,  $a=0$, $T^{\mathrm{diag}}= \log\{16p\gamma\log_2(4p)\}$ and $T^{\mathrm{off}}= \psi\bigl(\tilde{T}^{\mathrm{off}})$ with $\tilde{T}^{\mathrm{off}} = 2\log\{16p\gamma\log_2(2p)\}\bigr)$, satisfies
	\begin{equation} \label{Eq:DenseMain1}
	\bar{\mathbb{E}}_\theta^{\mathrm{wc}}(N) \leq C\biggl\{\frac{\sqrt{p}\log (ep\gamma)}{\vartheta^2} \vee \frac{s\log(ep\gamma)\log(ep)}{\beta^2} \vee 1\biggr\}.
	\end{equation}
	Furthermore, there exists $\beta_0(s)>0$, depending only on $s$, such that for all $\beta \leq \beta_0(s)$, the output $N$ satisfies
	\begin{equation} \label{Eq:DenseMain2}
	\bar{\mathbb{E}}_\theta(N) \leq C\biggl\{ \frac{\sqrt{p}  \log(ep\gamma)}{\vartheta^2} \vee  \frac{ \sqrt{s}\log(ep/\beta)\log(ep)}{\beta^2}  \vee 1 \biggl\},
	\end{equation}
	for $s \geq 2$, and
	\begin{equation} \label{Eq:DenseMain3}
	\bar{\mathbb{E}}_\theta(N) \leq C\biggl\{\frac{ \log(ep\gamma)\log(ep)}{\beta\vartheta}  \vee 1 \biggl\},
	\end{equation}
	for $s=1$.
\end{thm}

We defer detailed discussion of our response delay bounds until after we have presented our adaptive procedure in Section~\ref{SubSec:AdaptiveTheory}.

\subsection{Sparse case}
\label{SubSec:SparseAnalysis}
We now assume that $p\geq 2$, and analyse the performance of $N=N\bigl(\sqrt{8\log(p-1)}\bigr)$; in other words, we choose $a = \sqrt{8\log(p-1)}$.  This choice turns out to work particularly well when the vector of mean change is sparse in the sense that $s(\theta) \leq \sqrt{p}$, though again we do not assume this in our theory.  The motivation for this choice of $a$ comes from the fact that, for fixed $b$ and $j$, we have $\Lambda_b^{j', j} \bigm | \tau_b^j \stackrel{\mathrm{iid}}{\sim} \mathcal{N}(0, \tau_b^j)$ for $j' \in [p] \setminus \{j\}$ under the null.  Since $a$ is the threshold level for $\bigl|\Lambda_b^{j', j}\bigr|/\sqrt{\tau_b^j}$, it is therefore natural to choose $a$ to be of the same order as the maximum absolute value of $p-1$ independent and identically distributed $\mathcal{N}(0, 1)$ random variables. The declaration threshold $T^{\mathrm{off}}$ is determined based on Lemma~\ref{Lemma:ThreshChi}.  Theorem~\ref{Thm:NullThresh} below shows that, in the sparse case, the patience of our procedure is also guaranteed to be at least at the nominal level $\gamma > 0$. In addition, as in the dense case, we can also control the response delay of \texttt{ocd}$'$ according to Theorem~\ref{Thm:AlternativeThresh}.
\begin{thm}
	\label{Thm:NullThresh} 
	Let $X_1,X_2,\ldots $ be generated according to $\mathbb{P}_0$. For any $\gamma\geq 1$, let $(X_t)_{t\in\mathbb{N}}$, $\beta>0$, $a = \sqrt{8\log(p-1)}$, $T^{\mathrm{diag}}= \log\{16p\gamma\log_2(4p)\}$ and $T^{\mathrm{off}}=8\log\{16p\gamma\log_2(2p)\}$ be the inputs of Algorithm~\ref{Alg:changepoint3}, with corresponding output $N$. Then $\mathbb{E}_0(N) \ge \gamma$.
\end{thm}
\begin{thm} \label{Thm:AlternativeThresh}
	Assume that $X_1,X_2,\ldots$ are generated according to $\mathbb{P}_{z,\theta}$ for some $z$ and $\theta$ such that $\|\theta\|_2 = \vartheta \geq \beta > 0$ and that $\theta$ has an effective sparsity of $s:=s(\theta) $. Then there exists a universal constant $C > 0$, such that the output $N$ from Algorithm~\ref{Alg:changepoint3}, with inputs $(X_t)_{t\in\mathbb{N}}$, $\beta \in \mathbb{R}^p$,  $a=\sqrt{8\log(p-1)}$, $T^{\mathrm{diag}}= \log\{16p\gamma\log_2(4p)\}$ and $T^{\mathrm{off}}=8\log\{16p\gamma\log_2(2p)\}$, satisfies
	\begin{equation} \label{Eq:SparseMain}
	\bar{\mathbb{E}}_\theta(N)\leq \bar{\mathbb{E}}_\theta^{\mathrm{wc}}(N) \leq C\biggl\{\frac{s\log (ep\gamma)\log(ep)}{\beta^2} \vee 1\biggr\}.
	\end{equation}
\end{thm}

Comparing Theorems~\ref{Thm:DenseMain} and~\ref{Thm:AlternativeThresh}, we see that the thresholding induced by the non-zero choice of $a=\sqrt{8\log(p-1)}$ in Theorem~\ref{Thm:AlternativeThresh} facilitates an improved dependence on the effective sparsity~$s$ in the bound on the response delay, whenever $s$ is of smaller order than $\sqrt{p}$. 

\subsection{Adaptive procedure}
\label{SubSec:AdaptiveTheory}
To adapt to different sparsity levels $s$, we can run \texttt{ocd} (or \texttt{ocd}$'$) with two values of $a$ simultaneously: we choose $a=a^{\mathrm{dense}}=0$ to form the off-diagonal dense statistic $S^{\mathrm{off,d}} = S^{\mathrm{off}}(a^{\mathrm{dense}})$, and $a=a^{\mathrm{sparse}}=\sqrt{8\log(p-1)}$ to form the off-diagonal sparse statistic $S^{\mathrm{off,s}} = S^{\mathrm{off}}(a^{\mathrm{sparse}})$.  We recall that the diagonal statistic $S^{\mathrm{diag}}$ does not depend on the choice of $a$. For clarity, we redefine the three stopping times here: $N^{\mathrm{diag}}:=\inf\{n: S_n^{\mathrm{diag}} \geq T^{\mathrm{diag}}\}$, $N^{\mathrm{off,d}}:=\inf\{n: S_n^{\mathrm{off,d}} \geq T^{\mathrm{off,d}}\}$ and $N^{\mathrm{off,s}}:=\inf\{n: S_n^{\mathrm{off,s}} \geq T^{\mathrm{off,s}}\}$, for appropriately-chosen thresholds $T^{\mathrm{diag}}$, $T^{\mathrm{off,d}}$ and  $T^{\mathrm{off,s}}$. The output of this adaptive procedure is thus $N=N^{\mathrm{diag}} \wedge N^{\mathrm{off,d}} \wedge N^{\mathrm{off,s}}$. 

The following results provide patience and response delay guarantees for this adaptive procedure.
\begin{thm} \label{Thm:AdaptiveNull}
	Let $X_1,X_2,\ldots $ be generated according to $\mathbb{P}_0$. For any $\gamma\geq 1$, let $(X_t)_{t\in\mathbb{N}}$, $\beta>0$, $T^{\mathrm{diag}}= \log\{24p\gamma\log_2(4p)\}$, $T^{\mathrm{off,d}}=\psi\bigl(\tilde{T}^{\mathrm{off,d}}\bigr)$ with $\tilde{T}^{\mathrm{off,d}} = 2\log\{24p\gamma\log_2(2p)\}$ and $T^{\mathrm{off,s}}=8\log\{24p\gamma\log_2(2p)\}$ be the inputs of the adaptive version of Algorithm~\ref{Alg:changepoint3}, with corresponding output $N$. Then $\mathbb{E}_0(N) \ge \gamma$.
\end{thm}
\begin{thm} \label{Thm:AdaptiveAlt}
	Assume that $X_1,X_2,\ldots$ are generated according to $\mathbb{P}_{z,\theta}$ for some $z$ and $\theta$ such that $\|\theta\|_2 = \vartheta \geq \beta > 0$ and that $\theta$ has an effective sparsity of $s:=s(\theta) $. Then there exists a universal constant $C > 0$, such that the output $N$ from the adaptive version of Algorithm~\ref{Alg:changepoint3}, with inputs $(X_t)_{t\in\mathbb{N}}$, $\beta \in \mathbb{R}^p$, $T^{\mathrm{diag}}= \log\{24p\gamma\log_2(4p)\}$, $T^{\mathrm{off,d}}=\psi\bigl(\tilde{T}^{\mathrm{off,d}}\bigr)$ with $\tilde{T}^{\mathrm{off,d}} = 2\log\{24p\gamma\log_2(2p)\}$ and $T^{\mathrm{off,s}}=8\log\{24p\gamma\log_2(2p)\}$, satisfies
	\begin{equation} \label{Eq:AdaptiveMain1}
	\bar{\mathbb{E}}_\theta^{\mathrm{wc}}(N) \leq C\biggl\{\frac{s\log (ep\gamma)\log(ep)}{\beta^2} \vee 1\biggr\}.
	\end{equation}
	Furthermore, there exists $\beta_0(s)\in(0,1/2]$, depending only on $s$, such that for all $\beta \leq \beta_0(s)$, the output $N$ satisfies
	\begin{equation} \label{Eq:AdaptiveMain2}
	\bar{\mathbb{E}}_\theta(N)\leq  C\biggl\{\biggl(\frac{\sqrt{p}  \log(ep\gamma)}{\vartheta^2} \vee  \frac{ \sqrt{s}\log(ep\beta^{-1})\log(ep)}{\beta^2}\biggr) \wedge \frac{s\log(ep\gamma)\log(ep)}{\beta^2} \biggr\},
	\end{equation}
	for $s\geq 2$, and
	\begin{equation} \label{Eq:AdaptiveMain3}
	\bar{\mathbb{E}}_\theta(N) \leq  \frac{C \log(ep\gamma)\log(ep)}{\beta^2}  ,
	\end{equation}
	for $s=1$.	
\end{thm}
Comparing these two results with the corresponding theorems in Sections~\ref{SubSec:DenseAnalysis} and~\ref{SubSec:SparseAnalysis}, we see that by choosing slightly more conservative thresholds, the adaptive procedure retains the nominal patience control while (up to constant factors) achieving the best of both worlds in terms of its response delay guarantees under different sparsity regimes.

To better understand the worst-case and average-case response delay bounds in Theorem~\ref{Thm:AdaptiveAlt}, it is helpful to assume that $\vartheta/C_1 \leq \beta\leq \vartheta\leq C_1$ and $\log(\gamma/\beta)\leq C_2\log p$ for some $C_1, C_2 > 0$. Under these additional assumptions, the result of Theorem~\ref{Thm:AdaptiveAlt} takes the simpler form that for some $C> 0$, depending only on $C_1$ and $C_2$, we have
\[
\bar{\mathbb{E}}_\theta^{\mathrm{wc}}(N) \leq  \frac{Cs\log^2(ep)}{\vartheta^2}
\quad \text{and} \quad
\bar{\mathbb{E}}_\theta(N) \leq  \frac{C(s\wedge p^{1/2})\log^2(ep)}{\vartheta^2}.
\]
In particular, the average-case response delay upper bound exhibits a phase transition when the effective sparsity level $s$ is of order $\sqrt{p}$, which is the boundary between the sparse and dense cases.  Similar sparsity-related elbow effects have been observed in the minimax rate for high-dimensional Gaussian mean testing \citep{CCT2017} and the corresponding offline changepoint detection problem \citep{LGS2019}. On the other hand, we note that quadratic dependence on $\vartheta$ in the denominator, and the logarithmic dependence on $\gamma$ in the numerator, are known to be optimal in the case when $p=1$ \citep[Theorem~3]{Lorden1971}.  The different dependencies on sparsity of the worst-case and average-case response delays for the dense, sparse and adaptive versions of \texttt{ocd}$'$ are illustrated in Figure~\ref{Fig:Sparsity}.
\begin{figure}[htbp]
	\begin{center}
		\begin{tabular}{cc}
		\includegraphics[width=0.45\textwidth]{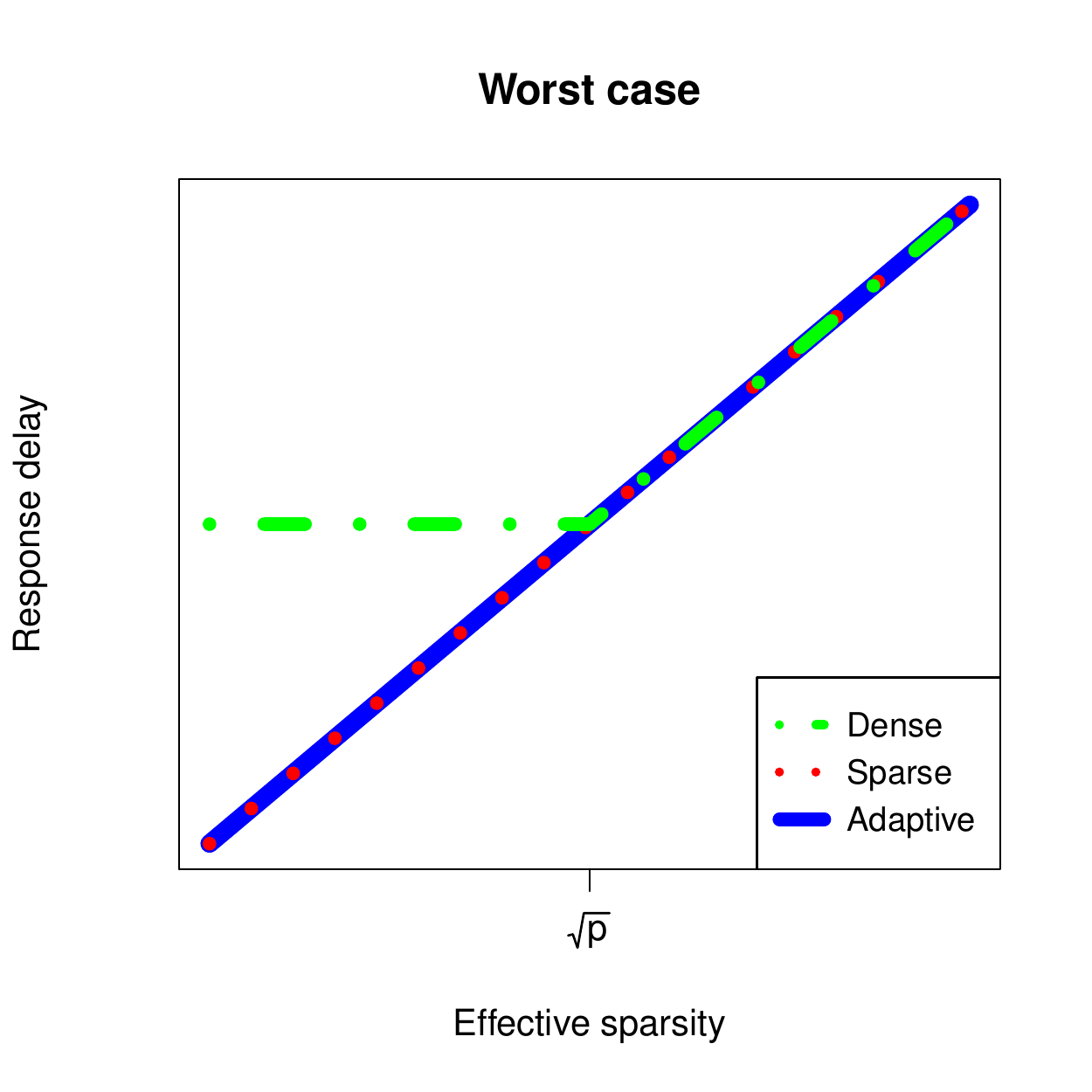}
		&
		\includegraphics[width=0.45\textwidth]{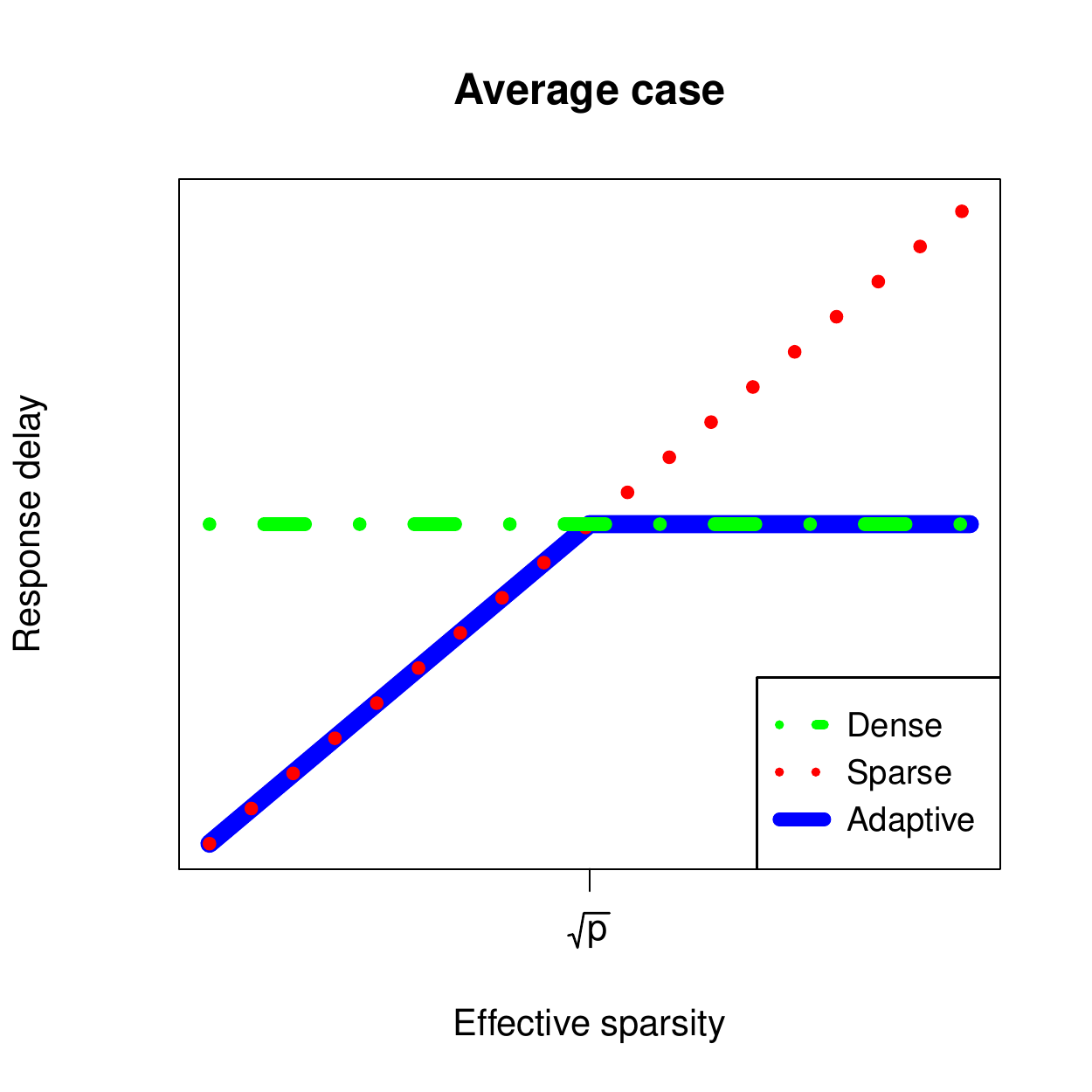}
		\end{tabular}
		\caption{\label{Fig:Sparsity}Illustration of the dependencies on sparsity of the worst-case and average-case response delays for the dense, sparse and adaptive versions of \texttt{ocd}$'$, as given by Theorems~\ref{Thm:DenseMain},~\ref{Thm:AlternativeThresh} and~\ref{Thm:AdaptiveAlt}.}
	\end{center}
      \end{figure}

      \subsection{Relaxation of assumptions}

      The setting we consider for our theoretical results, with independent Gaussian observations having identity covariance matrix, is convenient for facilitating a relatively clean presentation and to clarify the main ideas behind the \texttt{ocd} procedure.  Nevertheless, it is of interest to consider more general data generating mechanisms, where these assumptions are relaxed.  Focusing on the dense case for simplicity of exposition, the Gaussianity assumption ensures that our aggregated statistics have chi-squared distributions (under the null) or non-central chi-squared distributions (under the alternative), so we can apply existing sharp tail bounds.  If, instead, our observations have sub-Gaussian distributions, then the corresponding statistics would have sub-Gamma distributions, in the terminology of \citet{BLM2013}, so Bernstein's inequality could be applied to give alternative bounds in this setting.  Another place where we make use of the Gaussianity assumption is in comparing the trajectories of our test statistics with a Brownian motion with drift (see, for instance, the proof of Lemma~\ref{Lemma:DriftedRandomWalk}).  Since we can view these trajectories as discrete Gaussian random walks, we can establish direct inequalities in this comparison.  If we were to relax the Gaussianity, then we would need to rely on Donsker's invariance principle, or preferably its finite-sample version given by the Hungarian embedding \citep{KMT1976}.

      In cases where the covariance matrix of the observations were unknown, it may be possible to estimate this using a training sample, known to come from the null hypothesis, and use this to pre-whiten the data.  The form of the estimator to be used should be chosen to exploit any known dependence structure (e.g.~banding, Toeplitz or tapering) between the different coordinates.  Similar remarks apply when there is short-range serial (temporal) dependence between successive observations.  In Section~\ref{Sec:RealData}, we demonstrate one way of handling temporal dependence with real data, by studying the residuals of the fit of an autoregressive model.      

\section{Numerical studies}
\label{Sec:Simulation}

In this section, we study the empirical performance of the \texttt{ocd} algorithm and compare it with other online changepoint detection methods.  Recall that the (adaptive) \texttt{ocd} algorithm declares a change when any of the three statistics $S^{\mathrm{diag}}$, $S^{\mathrm{off,d}}$ and $S^{\mathrm{off,s}}$ exceeds their respective thresholds $T^{\mathrm{diag}}$, $T^{\mathrm{off,d}}$ and $T^{\mathrm{off,s}}$.  If a priori knowledge about the signal sparsity is available, it may be slightly preferable to use $N^{\mathrm{diag}} \wedge N^{\mathrm{off,d}}$ in the dense case, and $N^{\mathrm{diag}} \wedge N^{\mathrm{off,s}}$ in the sparse case, but for simplicity of exposition, we will focus on the adaptive version of our \texttt{ocd} procedure throughout the remainder of this section.  While the threshold choices given in Theorem~\ref{Thm:AdaptiveNull} guarantee that the patience of (adaptive) \texttt{ocd} will be at least at the nominal level, in practice, they may be conservative.  We therefore describe a scheme for practical choice of thresholds in Section~\ref{SubSec:PracThresh}.  Recall that, in order to form $S^{\mathrm{off,d}}$ and $S^{\mathrm{off,s}}$, two different entrywise hard thresholds for $A^{j', j}_b/\sqrt{t_b^j \vee 1}$ need to be specified. For $S^{\mathrm{off,d}}$, we choose $a = 0$ for both theoretical analysis and practical usage. For $S^{\mathrm{off,s}}$, the theoretical choice is $a=\sqrt{8\log(p-1)}$, but since this is also slightly conservative, the choice of $a=\sqrt{2\log p}$ is used in our practical implementation of the algorithm, and our numerical simulations below.

\subsection{Practical choice of declaration thresholds}
\label{SubSec:PracThresh}
The purpose of this section is to introduce an alternative to using the theoretical thresholds $T^{\mathrm{diag}}$, $T^{\mathrm{off,d}}$ and $T^{\mathrm{off,s}}$ provided by Theorem~\ref{Thm:AdaptiveNull}, namely to determine the thresholds through Monte Carlo simulation. The basic idea is that since the null distribution is known, we can simulate from it to determine the patience for any given choice of thresholds.  A complicating issue is the fact that the choices of the three thresholds $T^{\mathrm{diag}}$, $T^{\mathrm{off},\mathrm{d}}$ and $T^{\mathrm{off},\mathrm{s}}$ are related, so that we may be able to achieve the same patience by increasing $T^{\mathrm{diag}}$ and decreasing $T^{\mathrm{off},\mathrm{d}}$, for example.  To handle this, we first argue that the renewal nature of the processes involved means that, at least for moderately large thresholds, the times to exceedence for each of the three statistics $S^{\mathrm{diag}}$, $S^{\mathrm{off},\mathrm{d}}$ and $S^{\mathrm{off},\mathrm{s}}$ are approximately exponentially distributed.  Evidence to support this is provided by Figure~\ref{Fig:Exponential}, where we present QQ-plots of $N^{\mathrm{diag}}/m({N}^{\mathrm{diag}})$, $N^{\mathrm{off},\mathrm{d}}/m({N}^{\mathrm{off},\mathrm{d}})$ and $N^{\mathrm{off},\mathrm{s}}/m({N}^{\mathrm{off},\mathrm{s}})$, where the $m(N)$ statistics are empirical medians of the corresponding $N$ statistics (divided by $\log 2$) over 200 repetitions.

We can therefore set an individual Monte Carlo threshold for $S^{\mathrm{diag}}$ as follows (the other two statistics can be handled in identical fashion): for $r\in[B]$, simulate $X_1^{(r)},\ldots,X_\gamma^{(r)} \stackrel{\mathrm{iid}}{\sim} \mathcal{N}_p(0,I_p)$ and for each $n \in [\gamma]$, compute the diagonal statistic $S_n^{\mathrm{diag},(r)}$ on the $r$th sample.  Now compute $V^{(r)} := \max_{1 \leq n \leq \gamma} S_n^{\mathrm{diag},(r)}$, and take $\tilde{T}^{\mathrm{diag}}$ to be the $(1/e)$th quantile of $\{V^{(r)}:r\in [B]\}$.  The rationale for the final step here is that if $\mathbb{P}_0(V^{(1)} < \tilde{T}^{\mathrm{diag}}) = 1/e$, then $\mathbb{P}_0(\tilde{N}^{\mathrm{diag}} > \gamma) = 1/e$, where $\tilde{N}^{\mathrm{diag}}:=\min\{n: S_n^{\mathrm{diag}} \geq \tilde{T}^{\mathrm{diag}}\}$. Thus, under an exponential distribution for $\tilde{N}^{\mathrm{diag}}$, we have that $\tilde{N}^{\mathrm{diag}}$ has individual patience $\gamma$.

Having determined appropriate thresholds $\tilde{T}^{\mathrm{diag}}$, $\tilde{T}^{\mathrm{off},\mathrm{d}}$ and $\tilde{T}^{\mathrm{off},\mathrm{s}}$, we can then use similar ideas to set a suitable combined threshold $T^{\mathrm{comb}}$.  In particular, we also argue that $N^{\mathrm{diag}} \wedge N^{\mathrm{off},\mathrm{d}} \wedge N^{\mathrm{off},\mathrm{s}}$ has an approximate exponential distribution; see Figure~\ref{Fig:Exponential} for supporting evidence.  We therefore proceed as follows: for $r\in [B]$, simulate $\tilde{X}_1^{(r)},\ldots,\tilde{X}_\gamma^{(r)} \stackrel{\mathrm{iid}}{\sim} \mathcal{N}_p(0,I_p)$ and use this new data to compute $\tilde{S}_n^{\mathrm{diag},(r)} := S_n^{\mathrm{diag},(r)}/\tilde{T}^{\mathrm{diag}}$, $\tilde{S}_n^{\mathrm{off,d},(r)} := S_n^{\mathrm{off,d},(r)}/\tilde{T}^{\mathrm{off,d}}$ and $\tilde{S}_n^{\mathrm{off,s},(r)} := S_n^{\mathrm{off,s},(r)}/\tilde{T}^{\mathrm{off,s}}$ for each $n \in [\gamma]$, and set $W^{(r)} := \max\bigl\{\tilde{S}_n^{\mathrm{diag},(r)} \vee \tilde{S}_n^{\mathrm{off,d},(r)} \vee \tilde{S}_n^{\mathrm{off,s},(r)}:n\in [\gamma]\bigr\}$ on the $r$th sample.  Now take $T^{\mathrm{comb}}$ to be the $(1/e)$th quantile of $\{W^{(r)}:r\in [B]\}$.  Similar to before, our reasoning here is that if $\mathbb{P}_0(W^{(1)} < T^{\mathrm{comb}}) = 1/e$, then ${N}^{\mathrm{diag}} := \min\bigl\{n:S_n^{\mathrm{diag}} \geq \tilde{T}^{\mathrm{diag}}T^{\mathrm{comb}}\bigr\}$, ${N}^{\mathrm{off,d}} := \min\bigl\{n:S_n^{\mathrm{off,d}} \geq \tilde{T}^{\mathrm{off,d}}T^{\mathrm{comb}}\bigr\}$ and ${N}^{\mathrm{off,s}} := \min\bigl\{n:S_n^{\mathrm{off,s}} \geq \tilde{T}^{\mathrm{off,s}}T^{\mathrm{comb}}\bigr\}$ satisfy
\[
\mathbb{P}_0\Bigl({N}^{\mathrm{diag}} \wedge {N}^{\mathrm{off,d}} \wedge {N}^{\mathrm{off,s}} > \gamma\Bigr) = 1/e.
\]
Thus, under an exponential distribution for ${N}^{\mathrm{diag}} \wedge {N}^{\mathrm{off,d}} \wedge {N}^{\mathrm{off,s}}$, it again has the desired nominal patience.  Our practical thresholds, therefore, are $T^{\mathrm{diag}}=\tilde{T}^{\mathrm{diag}}T^{\mathrm{comb}}$, $T^{\mathrm{off,d}}=\tilde{T}^{\mathrm{off,d}}T^{\mathrm{comb}}$ and $T^{\mathrm{off,s}}=\tilde{T}^{\mathrm{off,s}}T^{\mathrm{comb}}$ for $S^{\mathrm{diag}}$, $S^{\mathrm{off,d}}$ and $S^{\mathrm{off,s}}$ respectively. Table~\ref{Tab:Patience} confirms that, with these choices of Monte Carlo thresholds, the patience of the adaptive \texttt{ocd} algorithm remains at approximately the desired nominal level.
\begin{figure}[htbp]
	\begin{center}
		\begin{tabular}{cc}
			\includegraphics[width=0.45\textwidth]{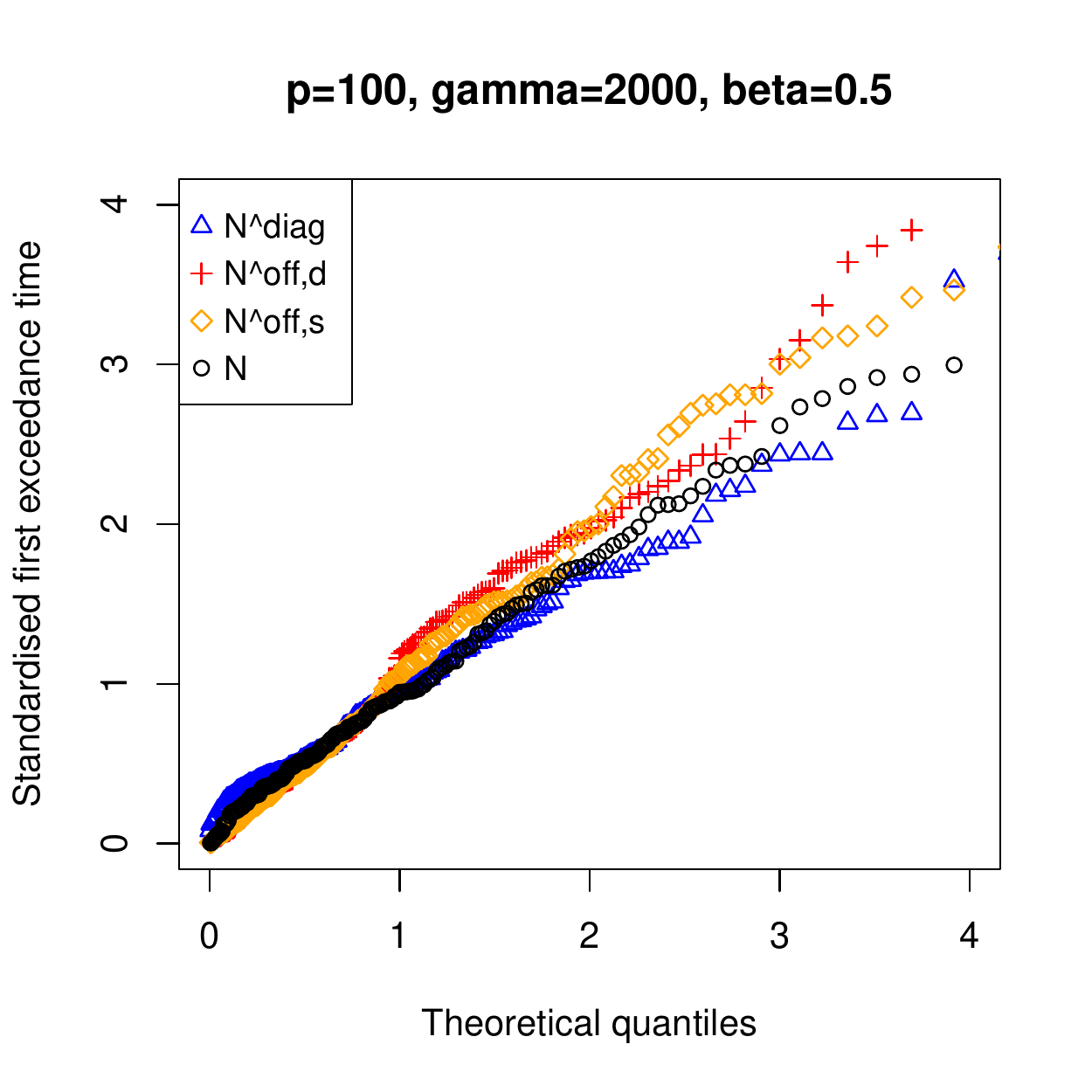} & 
			\includegraphics[width=0.45\textwidth]{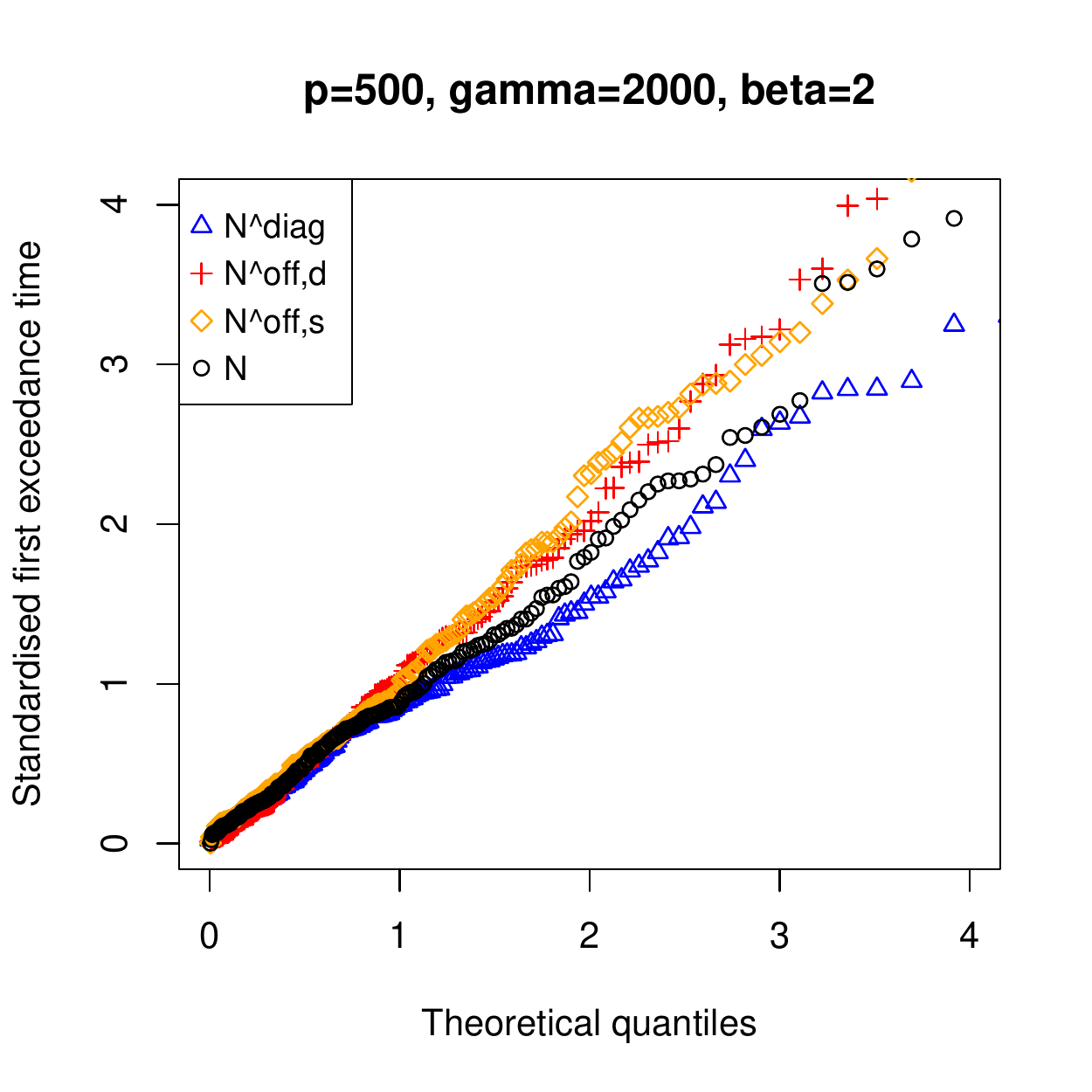}
		\end{tabular}
	\caption{\label{Fig:Exponential}QQ-plots of standardised versions of $N^{\mathrm{diag}}$, $N^{\mathrm{off,d}}$ and $N^{\mathrm{off,s}}$, as well as $N = N^{\mathrm{diag}} \wedge N^{\mathrm{off,d}} \wedge N^{\mathrm{off,s}}$, against theoretical $\mathrm{Exp}(1)$ quantiles.}
	\end{center}
\end{figure}

\begin{table}[htbp]
\begin{center}
	\caption{\label{Tab:Patience}Estimated run lengths under the null using the Monte Carlo thresholds described in Section~\ref{SubSec:PracThresh} over 500 repetitions, with desired patience level $\gamma=5000$. Algorithm is terminated after $20000$ data points for each repetition. Each reported value is the average run length taken over the repetitions which have already declared prior to time $20000$.  For reference, $\mathbb{E}(X \mid X < 20000) \approx 4626.9$ when $X \sim \mathrm{Exp}(1/5000)$.}
		\begin{tabular}{c|cc}
			 & $p=100$ & $p=1000$ \\
			\hline
			$\beta=2$ & $4606.2$ & $4480.8$ \\
			$\beta=1/2$ & $5291.5$ & $4383.6$ \\
	\end{tabular}
\end{center}
\end{table}



\subsection{Numerical performance of \texttt{ocd}}
\label{SubSec:ocdnumerics}

\begin{table}[htbp]
\begin{center}
	\caption{\label{Tab:ocdnumerics}Estimated response delays over 200 repetitions for $N^{\mathrm{diag}}$, $N^{\mathrm{off,d}}$ and $N^{\mathrm{off,s}}$ and the response delay of the combined declaration time $N$ for \texttt{ocd}, with the percentages of repetitions on which each statistics triggers the declaration first (or equal first) shown in parentheses. The quickest response in each setting is given in bold. Other parameters: $p=100$, $\gamma=5000$, $z=0$ and $\theta = \vartheta U$, where the distribution of $U$ is described in Section~\ref{SubSec:ocdnumerics}.}
		\begin{tabular}{cc|cccc}
			& 	 & 	\multicolumn{4}{|c}{$\beta = \vartheta$} \\
			\hline
			$s$ & $\vartheta$ & $N^{\mathrm{diag}}$ & $N^{\mathrm{off,d}}$ & $N^{\mathrm{off,s}}$ & $N$ \\
			\hline
			$1$ & $2$ & $\mathbf{11.5}$ $\mathbf{(83.5)}$ & $19.4$ $(1.5)$ & $13.0$ $(35)$  & $11.2$ \\
			$1$ & $1$ & $\mathbf{40.6}$ $\mathbf{(79.5)}$ & $74.4$ $(1.5)$ & $47.4$ $(19)$ & $39.1$ \\
			$1$ & $0.5$ & $\mathbf{136.3}$ $\mathbf{(82)}$ & $305.2$ $(1)$ & $169.2$ $(17)$ & $129.7$ \\ 
			$1$ & $0.25$ & $\mathbf{455.4}$ $\mathbf{(83)}$ & $1124.5$ $(1)$ & $635.0$ $(16)$ & $433.6$ \\
			$10$ & $2$ & $20.1$ $(9.5)$ & $19.2$ $(9.5)$ & $\mathbf{14.7}$ $\mathbf{(88)}$ & $14.3$ \\
			$10$ & $1$ & $69.7$ $(15.5)$ & $72.6$ $(12)$ & $\mathbf{52.4}$ $\mathbf{(73.5)}$ & $50.4$ \\
			$10$ & $0.5$ & $240.4$ $(29.5)$ & $308.0$ $(3)$ & $\mathbf{207.7}$ $\mathbf{(68)}$ & $197.1$ \\
			$10$ & $0.25$ & $\mathbf{723.3}$ $\mathbf{(56.5)}$ & $1124.3$ $(6)$ & $760.7$ $(37.5)$ & $648.4$  \\
			$100$ & $2$ & $53.3$ $(0.5)$ & $\mathbf{19.7}$ $\mathbf{(92)}$ & $27.4$ $(10)$ & $19.5$ \\
			$100$ & $1$ & $169.9$ $(2)$ & $\mathbf{75.2}$ $\mathbf{(85)}$ & $94.9$ $(14.5)$ & $73.1$ \\
			$100$ & $0.5$ & $544.1$ $(9)$ & $\mathbf{300.6}$ $\mathbf{(75.5)}$ & $345.1$ $(15.5)$ & $278.9$  \\
			$100$ & $0.25$ & $1493.6$ $(28.5)$ & $\mathbf{1206.0}$ $\mathbf{(51.5)}$ & $1420.2$ $(20)$ & $1065.4$  \\
	\end{tabular}
\end{center}
\end{table}

\begin{table}[htbp]
\begin{center}
	\caption{\label{Tab:ocdnumericsmis}Estimated response delays over 200 repetitions for $N^{\mathrm{diag}}$, $N^{\mathrm{off,d}}$ and $N^{\mathrm{off,s}}$ and the response delay of the combined declaration time $N$ for \texttt{ocd}. Settings where $\beta$ is both over- and under-specified are given.  The quickest response in each setting is given in bold. Other parameters: $p=100$, $\gamma=5000$, $z=0$ and $\theta = \vartheta U$, where the distribution of $U$ is described in Section~\ref{SubSec:ocdnumerics}.}
		\begin{tabular}{cc|cccc|cccc}
			  &    & 	\multicolumn{4}{|c|}{$\beta = 4\vartheta$} & 	\multicolumn{4}{c}{$\beta = \vartheta/4$}\\
			\hline
			$s$ & $\vartheta$ & $N^{\mathrm{diag}}$ & $N^{\mathrm{off,d}}$ & $N^{\mathrm{off,s}}$ & $N$ & $N^{\mathrm{diag}}$ & $N^{\mathrm{off,d}}$ & $N^{\mathrm{off,s}}$ & $N$ \\
			\hline
			$1$ & $2$ & $\mathbf{7.7}$ & $19.5$ & $12.8$ & $7.6$ & $30.3$ & $19.5$ & $\mathbf{12.6}$ & $12.6$ \\
			$1$ & $1$ & $\mathbf{27.8}$  & $77.7$  & $48.3$  & $27.6$ & $98.3$ & $73.7$ & $\mathbf{45.2}$  & $45.1$\\
			$1$ & $0.5$ & $\mathbf{92.9}$  & $288.9$  & $162.0$  & $92.3$ & $304.8$ & $304.9$  & $\mathbf{171.8}$  & $171.1$ \\    
			$1$ & $0.25$ & $\mathbf{351.7}$ & $1148.7$ & $657.2$ & $342.8$ & $746.7$  & $1158.1$ & $\mathbf{614.0}$ & $586.7$ \\
			$10$ & $2$ & $16.7$ & $19.0$ & $\mathbf{14.9}$ & $13.7$ & $50.0$ & $20.4$ & $\mathbf{15.1}$ & $15.0$\\
			$10$ & $1$ & $57.6$ & $72.9$ & $\mathbf{51.2}$ & $46.5$ & $161.9$  & $76.5$  & $\mathbf{54.7}$  & $54.5$ \\
			$10$ & $0.5$ & $228.3$  & $286.4$ & $\mathbf{201.0}$  & $180.5$ & $509.0$  & $314.7$ & $\mathbf{203.6}$  & $201.8$ \\
			$10$ & $0.25$ & $\mathbf{739.3}$  & $1175.1$  & $787.9$ & $645.1$ & $1208.2$ & $1189.6$ & $\mathbf{725.1}$  &$715.9$ \\
			$100$ & $2$ & $59.2$ & $\mathbf{18.9}$ & $25.3$ & $18.7$ & $110.8$ & $\mathbf{21.2}$ & $27.2$ & $20.5$ \\
			$100$ & $1$ & $213.9$ & $\mathbf{73.0}$  & $92.4$ & $71.0$ & $347.4$  & $\mathbf{76.8}$  & $95.5$  & $74.2$ \\
			$100$ & $0.5$ & $696.5$  & $\mathbf{307.0}$   & $385.0$ & $284.8$ & $1029.0$  & $\mathbf{310.2}$ & $352.5$ & $289.3$ \\
			$100$ & $0.25$ & $1811.5$  & $\mathbf{1218.1}$  & $1327.4$  & $967.1$ & $2149.9$  &  $\mathbf{1091.9}$  & $1175.9$  & $957.8$ \\
	\end{tabular}
\end{center}
\end{table}
In this section, we study the empirical performance of \texttt{ocd}. As shown in Figure~\ref{fig:Null_All_trigger}, under the alternative, all three statistics $S^{\mathrm{diag}}$, $S^{\mathrm{off,d}}$ and $S^{\mathrm{off,s}}$ in \texttt{ocd} can be the first to trigger a declaration that a mean change has occurred. We thus examine different settings under which each of these three statistics can respectively be the quickest to react to a change. Our simulations were run for $p = 100, s \in \{1,\lfloor p^{1/2} \rfloor, p\}, z \in \{0, 1000\}, \gamma = 5000, \vartheta \in \{2, 1, 0.5, 0.25\}, \beta \in \{\vartheta, 4\vartheta, \vartheta/4\}$.  In all cases, $\theta$ was generated as $\vartheta U$, where $U$ is uniformly distributed on the union of all $s$-sparse unit spheres in $\mathbb{R}^p$.  By this, we mean that we first generate a uniformly random subset $S$ of $[p]$ of cardinality $s$, then set $U := Z/\|Z\|_2$, where $Z = (Z^1,\ldots,Z^p)^\top$ has independent components satisfying $Z^j \sim \mathcal{N}(0,1)\mathbbm{1}_{\{j \in S\}}$.  Instead of terminating the \texttt{ocd} procedure once one of the three statistics declares a change (as we would in practice), we run the procedure until all three statistics have exceeded their respective thresholds. Tables~\ref{Tab:ocdnumerics} and~\ref{Tab:ocdnumericsmis} summarise the performance of the three statistics for $z=0$. Simulation results for $z=1000$ were similar, and are therefore not included here.

We first discuss the case when $\beta$ is correctly specified (Table~\ref{Tab:ocdnumerics}). When the sparsity $s$ is small or moderate and $\vartheta$ is small, the diagonal statistic $S^{\mathrm{diag}}$ is likely to be the first to declare a change.  The response delay of $S^{\mathrm{diag}}$ increases with $s$, which means that the off-diagonal sparse statistic $S^{\mathrm{off,s}}$ typically reacts quickest to a change when the $s$ is moderate to large and $\vartheta$ is not too small.  On the other hand, the stopping time $N^{\mathrm{off,d}}$, which is driven by the off-diagonal dense statistic, is not significantly affected by $s$ (in agreement with our average-case bound in Theorem~\ref{Thm:DenseMain}), and is usually the dominant statistic when the signal is dense.  A further observation is that the three individual response delays, as well as the combined response delay, are all approximately proportional to $\vartheta^{-2}$, a phenomenon which is supported by Theorem~\ref{Thm:AdaptiveAlt}.

Table~\ref{Tab:ocdnumericsmis} presents corresponding results when $\beta$ is both over- and under-specified.  We note that both $N^{\mathrm{off,d}}$ and $N^{\mathrm{off,s}}$ are almost unaffected by either type of misspecification.  For $N^{\mathrm{diag}}$, a mild over-misspecification of $\beta$ helps it to react faster, while an under-misspecification causes it to have increased response delay.  However, since we can also observe that $N^{\mathrm{diag}}$ rarely declares first by a large margin, the performance of \texttt{ocd} is highly robust to misspecification of~$\beta$, especially when $s$ is not too small. 

\subsection{Comparison with other methods}

We now compare our adaptive \texttt{ocd} algorithm with other online changepoint detection algorithms proposed in the literature, namely those of \citet{Mei2010}, \citet{XieSiegmund2013} and \citet{Chan2017}.  Since we were unable to find publicly-available implementations of any of these algorithms, we briefly describe below their methodology and the small adaptations that we made in order to allow them to be used in our settings.

\citet{Mei2010} assumes knowledge of $\theta$, and, on observing each new data point, aggregates likelihood ratio tests in each coordinate of the null $\mathcal{N}(0,1)$ against an alternative of $\mathcal{N}(\theta^j,1)$ in the $j$th coordinate.  More precisely, in the notation of~\eqref{Eq:MaxLR}, the algorithm declares a change when either $\sum_{j \in [p]} R_{n,\theta^j}^j$ or $\max_{j \in [p]} R_{n,\theta^j}^j$ exceeds given thresholds.  In our setting where we do not know $\theta$ and only assume that $\|\theta\|_2 \geq \beta$, we replace $\sum_{j \in [p]} R_{n,\theta^j}^j$ and $\max_{j \in [p]} R_{n,\theta^j}^j$ with
\[
\max\biggl\{\sum_{j=1}^p R_{n,\beta/\sqrt{p}}^j , \sum_{j=1}^p R_{n,-\beta/\sqrt{p}}^j\biggr\} \quad \text{and} \quad \max\biggl\{\max_{j \in [p]} R_{n,\beta/\sqrt{p}}^j,\max_{j \in [p]} R_{n,-\beta/\sqrt{p}}^j\biggr\}
\]
respectively.

The algorithms of \citet{XieSiegmund2013} and \citet{Chan2017} have a similar flavour.  The idea is to test the null $\mathcal{N}_p(0,I_p)$ distribution against an alternative where the $j$th coordinate has a $(1-p_0)\mathcal{N}(0,1) + p_0\mathcal{N}(\mu^j,1)$ mixture distribution, for some known $p_0 \in [0,1]$ and unknown $\mu^j \in \mathbb{R}$.  After specifying a window size $w$, both algorithms search for the strongest evidence against the null from the past $r \in [w]$ observations.  Specifically, writing $Z_{n,r}^j := r^{-1/2}\sum_{i=n-r+1}^n X_i^j$ for $n \in \mathbb{N}$, $r \in [n]$ and $j \in [p]$, the test statistics are of the form
\[
S_{\mathrm{XS,C}}^+(p_0,\lambda,\kappa,w) := \max_{r \in [w \wedge n]} \sum_{j=1}^p \log \biggl(1-p_0 + \lambda p_0 e^{(Z_{n,r}^j \vee 0)^2/\kappa}\biggr),
\]
where \citet{XieSiegmund2013} take $(\lambda,\kappa,w) = (1,2,200)$ and \citet{Chan2017} takes $(\lambda,\kappa,w) = (2\sqrt{2}-2,4,200)$.  Since such a test statistic is only effective when $\sum_{j \in [p]} (\mu^j \vee 0)^2$ is large, we considered statistics of the form $S_{\mathrm{XS,C}}^+(p_0,\lambda,\kappa,w) \vee S_{\mathrm{XS,C}}^-(p_0,\lambda,\kappa,w)$, where $S_{\mathrm{XS,C}}^-(p_0,\lambda,\kappa,w)$ replaces the exponent $Z_{n,r}^j \vee 0$ with $Z_{n,r}^j \wedge 0$.  An adaptive choice of $p_0$ is not provided by the authors, but the choices $p_0 \in \{1/\sqrt{p},0.1,1\}$ have been considered; we found the choice $p_0 = 1/\sqrt{p}$ to be the most competitive overall, so for simplicity of exposition, present only that choice in our results.

For each of the \citet{Mei2010}, \citet{XieSiegmund2013} and \citet{Chan2017} algorithms, we determined appropriate thresholds using Monte Carlo simulation, as suggested by the authors, and in a similar fashion to the way in which we set the \texttt{ocd} thresholds as described in Section~\ref{SubSec:PracThresh}.  This guarantees that the algorithms have approximately the nominal patience, and so allows us to compare the methods by means of the response delay.

Table~\ref{Tab:BakeOff} displays the response delays for the \texttt{ocd} algorithm, as well as the alternative methods described above, for $p \in \{100,2000\}$, $s \in \{5,\lfloor\sqrt{p}\rfloor,p\}$ and $\vartheta \in \{2, 1,0.5,0.25\}$.  In fact, we also ran simulations for $p=1000$, $s \in \{1,p/2\}$ and $\vartheta = 0.125$, but the results are qualitatively similar and are therefore omitted.  Overall, the results reveal that \texttt{ocd} performs very well in comparison with existing methods, across a wide range of scenarios; in several cases it is by far the most responsive procedure, and in none of the settings considered is it outperformed by much.  The \citet{XieSiegmund2013} and \citet{Chan2017} algorithms perform similarly to each other, and in most settings are both more competitive than the \citet{Mei2010} method described above.  We note that the performance of the \citet{XieSiegmund2013} and \citet{Chan2017} algorithms is relatively better when the signal-to-noise ratio $\vartheta$ is high; in these scenarios, the default window size $w=200$ is large enough that sufficient evidence against the null can typically be accumulated within the moving window.  For lower signal-to-noise ratios, this ceases to be the case, and from time $z+w$ onwards, the test statistic has the same marginal distribution (with no positive drift).  This explains the relative deterioration in performance for those algorithms in the harder settings considered.  As mentioned in the introduction, if the change in mean were known to be small, then the window size could be increased to compensate, but at additional computational expense; a further advantage of \texttt{ocd}, then, is that the computational time only depends on $p$ (and not on $\beta$ or other problem parameters).

\begin{table}[htbp]
\begin{center}
\caption{\label{Tab:BakeOff}Estimated response delay for \texttt{ocd}, as well as the algorithms of \citet{Mei2010} (\texttt{Mei}), \citet{XieSiegmund2013} (\texttt{XS}) and \citet{Chan2017} (\texttt{Chan}) over 200 repetitions, with $z=0$, $\gamma=5000$ and $\theta$ generated as described in Section~\ref{SubSec:ocdnumerics}.  The smallest response delay is given in bold.}
		\begin{tabular}{ccc|cccc}
			$p$ & $s$ & $\vartheta$ & \texttt{ocd} & \texttt{Mei} & \texttt{XS} & \texttt{Chan} \\
			\hline
			$100$ & $5$ & $2$ & $13.7$ & $36.3$ & $13.1$ & $\mathbf{11.9}$ \\ 
			$100$ & $5$ & $1$ & $46.9$ & $125.9$& $47.3$  & $\mathbf{42.0}$ \\
			$100$ & $5$ & $0.5$ & $174.8$ & $383.1$  & $194.3$ & $\mathbf{163.7}$\\
			$100$ & $5$ & $0.25$ & $\mathbf{583.5}$ & $970.4$  & $2147$ & $1888.8$\\
			$100$ & $10$ & $2$ & $14.9$ & $44.1$ & $15.2$ & $\mathbf{14.5}$ \\
			$100$ & $10$ & $1$ & $53.8$ & $150.1$  & $52.9$ & $\mathbf{51.5}$\\
			$100$ & $10$ & $0.5$ & $\mathbf{194.4}$ & $458.2$ & $255.8$  & $245.6$\\
			$100$ & $10$ & $0.25$ & $\mathbf{629.7}$ & $1171.3$  & $2730.7$ & $2484.9$\\
			$100$ & $100$ & $2$ & $\mathbf{19.4}$ & $72.7$ & $23.6$ & $27.5$ \\
			$100$ & $100$ & $1$ & $\mathbf{74.4}$ & $268.3$ & $89.6$  & $102.1$ \\
			$100$ & $100$ & $0.5$ & $\mathbf{287.9}$ & $834.9$ & $526.8$ & $756.0$ \\
			$100$ & $100$ & $0.25$ & $\mathbf{1005.8}$ & $1912.9$ & $3598.3$  & $3406.6$\\
			\hline
			$2000$ & $5$ & $2$ & $19.0$ & $130.5$ & $20.8$ & $\mathbf{15.6}$ \\
			$2000$ & $5$ & $1$ & $67.3$ & $316.7$  & $79.5$ & $\mathbf{59.5}$\\
			$2000$ & $5$ & $0.5$ & $\mathbf{247.3}$ & $680.2$ & $607.7$ & $285.0$ \\
			$2000$ & $5$ & $0.25$ & $\mathbf{851.3}$ & $1384.8$  & $4459.2$ & $3856.9$\\
			$2000$ & $44$ & $2$ & $\mathbf{37.5}$ & $247.7$ & $40.2$ & $37.7$ \\
			$2000$ & $44$ & $1$ & $\mathbf{136.0}$ & $596.1$ & $149.1$ & $145.0$ \\
			$2000$ & $44$ & $0.5$ & $\mathbf{479.1}$ & $1270.8$ & $2945.5$ & $2751.4$ \\
			$2000$ & $44$ & $0.25$ & $\mathbf{1584.2}$ & $2428.8$ & $4457.8$ & $5049.7$ \\
			$2000$ & $2000$ & $2$ & $\mathbf{97.1}$ & $949.9$ & $103.2$ & $136.7$ \\
			$2000$ & $2000$ & $1$ & $\mathbf{360.7}$ & $2126.5$ & $1020.0$ & $2074.7$ \\
			$2000$ & $2000$ & $0.5$ & $\mathbf{1296.0}$ & $3428.1$ & $4669.3$ & $4672.7$ \\
			$2000$ & $2000$ & $0.25$ & $\mathbf{3436.7}$ & $4140.4$  & $5063.7$ & $5233.5$ \\
		\end{tabular}
        \end{center}
\end{table}

\subsection{Real data example}
\label{Sec:RealData}

We consider a seismic signal detection problem, using a dataset from the High Resolution Seismic Network, operated by the Berkeley Seismological Laboratory.  Ground motion sensor measurements were recorded using geophones at a frequency of 250 Hz in three mutually perpendicular directions, at 13 stations near Parkfield, California for a total of 740 seconds from 2am on 23 December 2004.  This dataset was also studied by~\citet{XXM2019}, and was obtained from \url{http://service.ncedc.org/fdsnws/dataselect/1/}.
To begin, we removed the linear trend in each coordinate and applied a 2--16 Hz bandpass filter to the data using the GISMO toolbox\footnote{Available at: http://geoscience-community-codes.github.io/GISMO/}; these are standard pre-processing steps in the seismology literature \citep[e.g.][]{Cetal2018,XXM2019}.  In order to reduce the effects of temporal dependence, we computed a root mean square amplitude envelope, downsampled to 16 Hz, and then extracted the residuals from the fit of an autoregressive model of order 1.  The processed data are available as a built-in dataset in the \texttt{ocd} \texttt{R} package.   The first four minutes of the series were used to estimate the baseline mean and variance for each sensor, and we plot the standardised data from 2:04am onwards in Figure~\ref{Fig:seismicdata}.  When applying our $\texttt{ocd}$ algorithm to this data, we specified the patience level to be $\gamma = 1.35\times 10^6$, corresponding to a patience of one day, and $\beta = 150$. The \texttt{ocd} algorithm declared a change at 02:10:03.84, and was triggered by $S^{\text{off,d}}$.  According to the Northern California Earthquake Catalog\footnote{Available at: \url{http://www.ncedc.org/ncedc/catalog-search.html.}}, an earthquake of magnitude $1.47$ Md hit near Atascadero, California ($50$ km away from Parkfield) at 02:09:54.01, so the delay was $9.8$ seconds.  It is known\footnote{One source for this information is \url{https://www.usgs.gov/natural-hazards/earthquake-hazards/science/seismographs-keeping-track-earthquakes.}} that P waves, which are the primary preliminary wave and arrive first after an earthquake, travel at up to 6 km/s in the Earth's crust, which is consistent with this delay.
\begin{figure}[htbp]
	\begin{center}
		\includegraphics[width=0.9\textwidth]{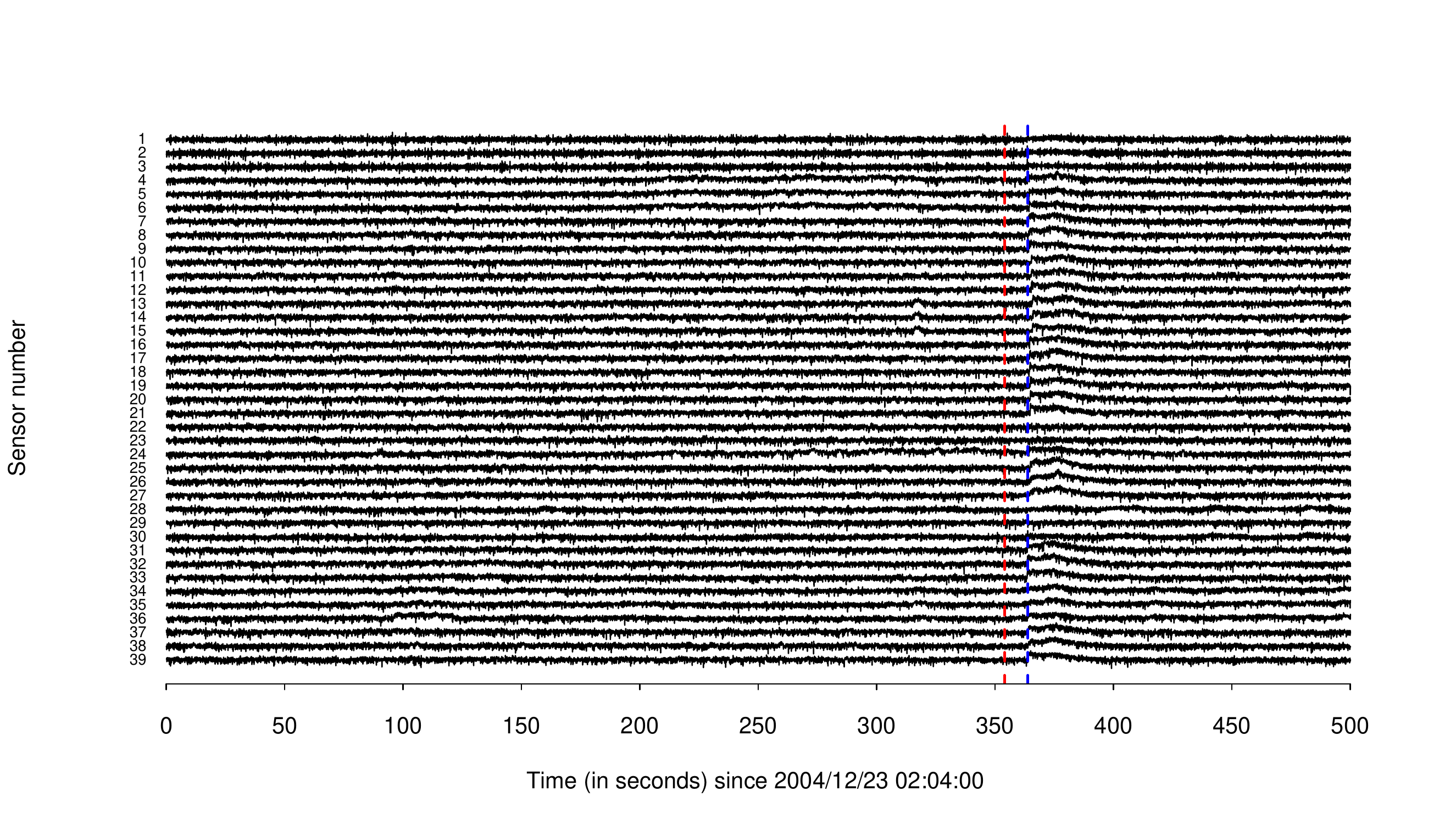}
		\caption{\label{Fig:seismicdata} Standardised, pre-processed earthquake data from 39 sensors.  The time of the 1.47 Md earthquake is given by the vertical red dashed line, while time of $\texttt{ocd}$ declaration of change is given as a blue dashed line.}
	\end{center}
\end{figure}

\section{Proofs of main results}
\label{Sec:Proofs}

\subsection{Proofs from Section~\ref{SubSec:DenseAnalysis}}

\begin{proof}[of Theorem~\ref{Thm:Null}]
	Define $m:=\lfloor 2\gamma \rfloor$. It suffices to prove that
	(a) $\mathbb{P}_0(N^{\mathrm{off}}\le m )\leq 1/4$ and 
	(b) $\mathbb{P}_0(N^{\mathrm{diag}}\le m) \leq 1/4$, since then we have
	\begin{align*}
	\mathbb{E}_0(N) &= \mathbb{E}_0(N^{\mathrm{off}}\wedge N^{\mathrm{diag}}) \geq 2\gamma\mathbb{P}_0(N^{\mathrm{off}}  \wedge N^{\mathrm{diag}} > 2\gamma) \\
	&\geq 2\gamma\big\{1 - \mathbb{P}_0(N^{\mathrm{off}}\le m) - \mathbb{P}_0(N^{\mathrm{diag}}\le m) \bigr\} \geq \gamma.
	\end{align*}
	We prove the two claims below.
	
	(a) By~\eqref{Eq:Toff} and a union bound, we have
	\begin{align}
	\mathbb{P}_0(N^{\mathrm{off}}\le m)  \le  \sum_{\substack{n \in [m], j \in [p]\\ b\in\mathcal{B}}}\mathbb{P}_0\Bigl(Q_{n,b}^j \ge T^{\mathrm{off}}\Bigr) =  \sum_{\substack{n \in [m], j \in [p]\\ b\in\mathcal{B}}} \mathbb{E}_0\biggl[\mathbb{P}_0\Bigl(Q_{n,b}^j \ge T^{\mathrm{off}} \Bigm | \tau_{n,b}^j\Bigr)\biggr]. \label{Eq:NoffUnionBound}
	\end{align}
	Recall that under the null, $\Lambda^{k, j}_b \mid \tau_b^j \stackrel{\mathrm{iid}}{\sim} \mathcal{N}(0, \tau_b^j)$ for all $b \in \mathcal{B}, j \in [p]$ and $k \in [p]\backslash \{j\}$, which implies that $Q_b^j \mid \tau_b^j \sim \chi_{p-1}^2\mathbbm{1}_{\{\tau_b^j > 0\}}$. Thus, we have by \citet[][Lemma 1]{LaurentMassart2000} that for all $n \in [m], j \in [p]$ and $b \in \mathcal{B}$,
	\begin{align} \label{Eq:LMtail}
	\mathbb{P}_0\Bigl(Q_{n,b}^j  \ge T^{\mathrm{off}} \Bigm | \tau_{n,b}^j \Bigr) \leq e^{-\tilde{T}^{\mathrm{off}}/2}.
	\end{align}
	Combining~\eqref{Eq:NoffUnionBound} and \eqref{Eq:LMtail}, we have
	\begin{align}  \label{Eq:Noffprob_d}
	\mathbb{P}_0(N^{\mathrm{off}}\le m )\le |\mathcal{B}|mpe^{-\tilde{T}^{\mathrm{off}}/2} \le 1/4.
	\end{align}
	
	(b) For $j \in [p]$ and $b \in \mathcal{B}\cup\mathcal{B}_0$, denote $N_b^j := \inf\{n: R_{n,b}^j \ge  T^{\mathrm{diag}}\}$, where $R_{n, b}^j$ is defined by~\eqref{Eq:MaxLR}. By Lemma~\ref{Lemma:DefRProcess}, we have that $R_{n,b}^j=\{R_{n-1,b}^j + b(X_{n}^j-b/2)\} \vee 0$, and that this process is always non-negative. Then $N^{\mathrm{diag}} = \min\bigl\{N_b^j: j\in[p], b\in\mathcal{B}\cup\mathcal{B}_0\bigr\}$. 
	
	Now, fix some $j\in[p]$ and $b\in\mathcal{B}\cup\mathcal{B}_0$. Define $U_0 := 0$ and  $U_h := \inf\{n> U_{h-1}: R^j_{n,b} \notin (0, T^{\mathrm{diag}}) \}$ for $h\in\mathbb{N}$, and let $H := \inf\{ h: R^{j}_{U_h,b}  \ge T^{\mathrm{diag}} \}$. Then
	\[   
	N_b^j = U_{H} \geq H.  
	\]
	To study the distribution of $H$, consider the one-sided sequential probability ratio test of $H_{0,Z}: Z_1,Z_2,\ldots\stackrel{\mathrm{iid}}{\sim} \mathcal{N}(0, 1)$ against $H_{1,Z}: Z_1,Z_2,\ldots\stackrel{\mathrm{iid}}{\sim}  \mathcal{N}(b, 1)$ with log-boundaries $T^{\mathrm{diag}}$ and $-\infty$. The associated stopping time for this test is
	\[
	N_{\mathrm{os}} := \inf\biggl\{n\in\mathbb{N}: b\sum_{t=1}^n (Z_t - b/2) \geq T^{\mathrm{diag}}\biggr\}.
	\]
	Since $(R_{n,b}^j)_n$ is a Markov process that renews itself every time it hits 0, $H$ follows a geometric distribution with success probability 
	\begin{align*}
	\mathbb{P}_0(R_{U_1,b}^j \geq T^{\mathrm{diag}}) \le \mathbb{P}_{H_{0,Z}}(N_{\mathrm{os}} < \infty) \le e^{-T^{\mathrm{diag}}},
	\end{align*}
	where the last inequality follows from Lemma~\ref{Lemma:OneSidedHitting}. Consequently,
	\[
	\mathbb{P}_0(N_b^j \leq m) \leq \mathbb{P}_0(H \leq m) \leq 1 - \Bigl(1-e^{-T^{\mathrm{diag}}}\Bigr)^m .
	\]
	As the above inequality holds for all $j\in[p]$ and $b\in \mathcal{B}\cup\mathcal{B}_0$, we have that
	\begin{align}
	\mathbb{P}_0(N^{\mathrm{diag}}>m) & = \mathbb{P}_0 \biggl(\bigcap_{j\in[p], b\in\mathcal{B}\cup\mathcal{B}_0} \{N_b^j > m\}\biggr) = \prod_{j\in[p]} \biggl\{1 - \mathbb{P}_0\biggl(\bigcup_{b\in\mathcal{B}\cup\mathcal{B}_0} \{N_b^j \leq  m\}\biggr)\biggr\} \nonumber\\
	&\geq \Bigl[1 - |\mathcal{B}\cup \mathcal{B}_0|\Bigl\{1 - \bigl(1 - e^{-T^\mathrm{diag}}\bigr)^m\Bigr\}\Bigr]^p \geq 1 - mp|\mathcal{B}\cup\mathcal{B}_0|e^{-T^{\mathrm{diag}}} \geq 3/4,
	\label{Eq:NdiagProb}
	\end{align}
	as desired, where in the penultimate inequality, we twice used the fact that $(1-x)^\alpha \geq 1-\alpha x$ for all $\alpha \geq 1$ and $x\in[0,1]$. 
\end{proof}
\bigskip
The proof of Theorem~\ref{Thm:DenseMain} is quite involved.  We first define some relevant quantities, and then state and prove some preliminary results. For $\theta\in\mathbb{R}^p$ with effective sparsity $s(\theta)$, there is at most one coordinate in $\theta$ of magnitude larger than $\vartheta/\sqrt{2}$, so there exists $b_* \in \bigl\{ \beta/\sqrt{s(\theta)\log_2(2p)}, -\beta/\sqrt{s(\theta)\log_2(2p)}\bigr\} \subseteq \mathcal{B}$ such that 
\begin{equation}
\label{Eq:mathcalJ}
\mathcal{J}:=\Bigl\{j \in [p]: \theta^j  / b_* \ge 1 \text{ and } |\theta^j|\le \vartheta/\sqrt{2} \Bigr\} 
\end{equation}
has cardinality at least $s(\theta)/2$ (note that the condition $\theta^j/b_*\geq 1$ above ensures that $\{\theta^j:j\in\mathcal{J}\}$ all have the same sign as $b_*$). Both $b_*$ and $\mathcal{J}$ can be chosen as functions of $\theta$. Now, given any sequence $X_1,X_2,\ldots \in\mathbb{R}^p$ and $\theta \in \mathbb{R}^p$, define for any $\alpha\in(0,1]$ the function
\begin{equation}
\label{Eq:ResidualTail}
q(\alpha) = q(\alpha; X_1,\ldots,X_z,\theta) := 
\inf\bigl\{y \in \mathbb{R}: \bigl|\{j\in\mathcal{J}: t_{z,b_*}^j \leq y\}\bigr| \geq  \alpha |\mathcal{J}|\bigr\},
\end{equation}
where $t_{z,b_*}^j$ is obtained by running Algorithm~\ref{Alg:changepoint3} up to time $z$ with $a=0$ and $T^{\mathrm{diag}} = T^{\mathrm{off}} = \infty$. In other words, $q(\alpha)$ is the empirical $\alpha$-quantile of the tail lengths $(t^j_{z,b_*}:j\in\mathcal{J})$ when we run the algorithm without declaring any change up to time $z$. Recall the definition of the function $\psi$ in~\eqref{Eq:Toff}.  

\begin{prop}
	\label{Thm:Alternative}
	Assume that $X_1,X_2,\ldots$ are generated according to $\mathbb{P}_{z,\theta}$ for some $z$ and $\theta$ such that $\|\theta\|_2 = \vartheta \geq \beta > 0$ and that $\theta$ has an effective sparsity of $s:=s(\theta) \geq 2$. Then the output $N$ from Algorithm~\ref{Alg:changepoint3}, with input $(X_t)_{t\in\mathbb{N}}$, $\beta \in \mathbb{R}^p$,  $a=0$, $T^{\mathrm{diag}}\geq 1$ and $T^{\mathrm{off}}=\psi(\tilde T^{\mathrm{off}})$ for
	$\tilde T^{\mathrm{off}}\geq \log(ep)$, satisfies
	\begin{equation} \label{Eq:AlternativeResult}
	\mathbb{E}_{z,\theta}\bigl\{(N-z)\vee 0\bigm| X_1,\ldots,X_z\bigr\} \leq \frac{396\tilde T^{\mathrm{off}} + 65\sqrt{p \tilde T^{\mathrm{off}}}}{\vartheta^2} +  \frac{24 \log_2(2p)}{\alpha\beta^2} + 3q(\alpha) + 2,
	\end{equation}
	for any $\alpha\in(0,1]$. 
\end{prop}
\begin{proof}
	Since the bound in~\eqref{Eq:AlternativeResult} is positive, we may, throughout the proof and for arbitrary $z \in \mathbb{N}$, restrict attention to realisations $X_1=x_1,\ldots,X_z=x_z$ for which we have not declared a change by time $z$. In other words, we have $N > z$. This restriction also ensures that $q(\alpha)$ defined in~\eqref{Eq:ResidualTail} is now indeed the empirical $\alpha$-quantile of the tail lengths $(t_{z, b^*}^j: j \in \mathcal{J})$ at the changepoint. Denote $\mathcal{J}_{\alpha}:=\bigl\{ j \in \mathcal{J}: t_{z, b_*}^j \leq q(\alpha) \bigr\}$. Then we have $|\mathcal{J}_\alpha| \geq \alpha |\mathcal{J}| \geq \alpha s/2$.
	
	We now fix some
	\begin{equation} \label{Eq:AssumptionOnr}
	r \ge \Biggl\{\frac{12\bigl(\tilde{T}^{\mathrm{off}}+\sqrt{2(p-1)\tilde{T}^{\mathrm{off}}}\bigr)}{\vartheta^2} \vee 3q(\alpha)\Biggr\} +2:=r_0.
	\end{equation}
	Note that $r_0 > 3q(\alpha) \geq 3t_{z,b_*}^j$ for all $j \in \mathcal{J}_\alpha$ . For $j \in \mathcal{J}_\alpha$, we define the event
	\[
	\Omega_r^j := \Bigl\{t_{z + \lfloor r \rfloor, b_*}^{j} > 2\lfloor r \rfloor/3\Bigr\}.
	\]
	Under $\mathbb{P}_{z, \theta}$, conditional on $X_1=x_1,\ldots,X_z=x_z$, we know that $X_{z+1}, X_{z+2}, \ldots \stackrel{\mathrm{iid}}{\sim} \mathcal{N}_p(\theta, I_p)$. Hence, by using Lemma~\ref{Lemma:AltDefTailLength} and applying Lemma~\ref{Lemma:DriftedRandomWalk}{(b)} to $t_{z + \lfloor r \rfloor, b_*}^j \wedge \lfloor r \rfloor$ for $j \in \mathcal{J}_\alpha$, we obtain
	\begin{align}
	\mathbb{P}_{z, \theta}\biggl( \bigcap_{j \in \mathcal{J}_\alpha}(\Omega_{r}^j)^c\Bigm | X_1=x_1, \ldots, X_z=x_z  \biggr) \le \exp\bigl\{-|\mathcal{J}_\alpha|b_*^2\lfloor r \rfloor/12  \bigr\} \le  \exp\{  -\alpha sb_*^2\lfloor r \rfloor/24\}. \label{Eq:SmallProbShortTail}
	\end{align}
	
	We now work on the event $\Omega_{r}^{j}$, for some $j \in \mathcal{J}_\alpha$. We note that~\eqref{Eq:AssumptionOnr} guarantees that $ r  \ge 2$, and thus $t_{z + \lfloor r \rfloor, b_*}^{j} \geq \bigl\lceil2\lfloor r\rfloor/3\bigr\rceil \ge 2$. Then, by Lemma~\ref{Lemma:AuxiliaryTailStats} and the fact that $r_0>3t_{z, b_*}^j$, we have that
	\[
	\frac{\lfloor r \rfloor} 3 <  \frac{t_{z + \lfloor r \rfloor, b_*}^{j}}2  \leq \tau_{z + \lfloor r \rfloor, b_*}^{j} \le \frac{3t_{z + \lfloor r \rfloor, b_*}^{j}}{4} \le \frac{3\bigl(t_{z, b_*}^{j} +   r \bigr)}{4} <   r.
	\]
	Hence we conclude that on the event $\Omega_{r}^{j}$,
	\begin{equation}  \label{Eq:AuxTailBounds}
	2/3 \le \lfloor r \rfloor/3 < \tau_{z + \lfloor r \rfloor, b_*}^{j} \leq \lfloor r \rfloor.
	\end{equation}
	Recall that $\Lambda_{z+\lfloor r \rfloor, b_*}^{\bdot, j} \in \mathbb{R}^p$ records the tail CUSUM statistics with tail length $\tau_{z + \lfloor r \rfloor, b_*}^{j}$. We observe by~\eqref{Eq:AuxTailBounds} that on $\Omega_{r}^{j}$, only post-change observations are included in $\Lambda_{z+\lfloor r \rfloor, b_*}^{\bdot, j}$. Hence we have that on the event $\Omega_{r}^{j}$,
	\begin{equation} \label{Eq:Lambdarj}
	\Lambda_{z + \lfloor r \rfloor, b_*}^{k, j} \bigm | \bigl\{\tau_{z + \lfloor r \rfloor, b_*}^{j},X_1=x_1,\ldots,X_z=x_z\bigr\} \stackrel{\mathrm{ind}}{\sim} \mathcal{N}\bigl(\theta^k \tau_{z + \lfloor r \rfloor, b_*}^{j},  \tau_{z + \lfloor r \rfloor, b_*}^{j}\bigr)
	\end{equation}
	for $k \in [p]\setminus\{j\}$. Therefore, on the event $\Omega_{r}^{j}$ and conditional on $\tau_{z + \lfloor r \rfloor, b_*}^{j},X_1=x_1,\ldots,X_z=x_z$, the random variable $ \frac{\|\Lambda^{-j, j}_b\|_2^2}{\tau^{j}_{z + \lfloor r \rfloor, b_*} \vee 1} = \frac{\|\Lambda^{-j, j}_b\|_2^2}{\tau^{j}_{z + \lfloor r \rfloor, b_*}}$ follows a non-central chi-squared distribution with $p-1$ degrees of freedom and noncentrality parameter $\|\theta^{-j}\|_2^2\tau_{z + \lfloor r \rfloor, b_*}^{j}$. Since $j \in \mathcal{J}$ and $s \geq 2$, we observe, by~\eqref{Eq:mathcalJ} and~\eqref{Eq:AuxTailBounds} that $\|\theta^{-j}\|_2^2\tau_{z + \lfloor r \rfloor, b_*}^{j} \geq \vartheta^2\lfloor r \rfloor/6$ on $\Omega_{r}^{j}$.  Write
	\begin{equation*}
	E_r^j := \biggl\{\frac{\|\Lambda^{-j, j}_{z+\lfloor r \rfloor, b_*}\|_2^2}{\tau^{j}_{z + \lfloor r \rfloor, b_*} \vee 1} < T^{\mathrm{off}}\biggr\}.
	\end{equation*}
	Then by \citet[][Lemma~8.1]{Birge2001}, we have
	\begin{align}
	\mathbb{P}_{z, \theta} \Bigl(E_r^{j}  \cap \Omega_{r}^{j} \bigm | \tau_{z + \lfloor r \rfloor, b_*}^{j},&X_1=x_1,\ldots,X_z=x_z  \Bigr) \nonumber \\
	&\leq \exp\Biggl\{  - \frac{\Bigl(\vartheta^2\lfloor r \rfloor/6-\tilde{T}^{\mathrm{off}}-\sqrt{2(p-1)\tilde{T}^{\mathrm{off}}}\Bigr)^2}{4\bigl(p-1+\vartheta^2\lfloor r \rfloor/3\bigr)}\Biggr\}.      \label{Eq:SmallProbChiSquared}
	\end{align}
	Combining~\eqref{Eq:SmallProbShortTail} and~\eqref{Eq:SmallProbChiSquared}, we deduce that
	\begin{align*}
	&\mathbb{P}_{z, \theta}\bigl(N > z +  r  \bigm| X_1=x_1 , \ldots, X_z=x_z\bigr) \le \mathbb{P}_{z, \theta}\bigl(N > z + \lfloor r \rfloor \bigm| X_1=x_1 , \ldots, X_z=x_z\bigr) \\
	&\leq  \mathbb{P}_{z, \theta}\biggl( \bigcap_{j \in \mathcal{J}_\alpha}(\Omega_{r}^j)^c  \biggm | X_1=x_1, \ldots, X_z=x_z  \biggr) +  \sum_{j \in \mathcal{J}_\alpha} \mathbb{P}_{z, \theta} \Bigl(E_r^{j}  \cap \Omega_{r}^{j} \Bigm | X_1=x_1,\ldots,X_z=x_z  \Bigr)  \\
	&\leq \exp \biggl \{   -\frac{\alpha sb_*^2(r-1)}{24} \biggr\} + p\exp\Biggl\{  - \frac{\Bigl(\vartheta^2(r-1)/6-\tilde{T}^{\mathrm{off}}-\sqrt{2(p-1)\tilde{T}^{\mathrm{off}}}\Bigr)^2}{4\bigl(p-1+\vartheta^2(r-1)/3\bigr)} \Biggr\} \\
	&\leq \exp \biggl \{   -\frac{\alpha sb_*^2(r-1)}{24} \biggr\} + p\exp\biggl\{  -\frac{\vartheta^4(r-1)^2}{576\bigl(p-1+\vartheta^2(r-1)/3\bigr)} \biggr\},
	\end{align*}
	where the last inequality uses~\eqref{Eq:AssumptionOnr}. Therefore, we have
	\begin{align*}
	\mathbb{E}_{z, \theta}\bigl\{(N-z)\vee 0 &\mid X_1=x_1, \ldots , X_z=x_z\bigr\} = \int_{0}^{\infty} \mathbb{P}_{z, \theta}\bigl(N > z + u \bigm| X_1=x_1 , \ldots, X_z=x_z\bigr)\, du \\
	&\leq r_0 + \int_{r_0-1}^{\infty}\biggl[\exp \biggl \{   -\frac{\alpha sb_*^2u}{24} \biggr\} + p\exp\biggl\{  -\frac{\vartheta^4u^2}{576\bigl(p-1+\vartheta^2u/3\bigr)} \biggr\} \biggr] \wedge 1 \, du \\
	&\leq r_0 + \frac{24}{\alpha sb_*^2} +  \int_{0}^{\infty}  \bigl(pe^{-\vartheta^2 u/384}\bigr) \wedge 1 \, du +  \int_{0}^{\infty}  \Bigl(pe^{-\frac{\vartheta^4 u^2}{1152(p-1)}}\Bigr) \wedge 1 \, du \\
	&\leq r_0 + \frac{24}{\alpha sb_*^2} + \frac{384\log(ep)}{\vartheta^2} + \frac{24\sqrt{2(p-1)\log p}}{\vartheta^2} + \frac{12\sqrt{2\pi(p-1)}}{\vartheta^2} \\
	&\le r_0 + \frac{24}{\alpha sb_*^2} + \frac{384\log(ep)}{\vartheta^2} + \frac{48\sqrt{(p-1)\log(ep)}}{\vartheta^2},
	\end{align*}
	where the penultimate inequality follows from the fact that $1-\Phi(x) \le \frac{1}{2}e^{-x^2/2}$ for $x\ge 0$. The desired bound~\eqref{Eq:AlternativeResult} follows by substituting in the expressions for $r_0$ and $b_*$.
\end{proof}

The following two propositions control the residual tail length quantile term $q(\alpha)$ in \eqref{Eq:AlternativeResult} in the worst-case and average-case scenarios respectively.
\begin{prop}
	\label{Prop:WorstCaseResidualTail}
	Let $X_1,X_2,\ldots$, $z$, $\theta$, $s$, $a$, $p$ and $N$ be defined as in Proposition~\ref{Thm:Alternative}. On the event $\{N>z\}$, we have
	\[
	q(1; X_1,\ldots,X_z,\theta) \leq \frac{8T^{\mathrm{diag}}s\log_2(2p)}{\beta^2}.
	\]
\end{prop}
\begin{proof}
	We will show the stronger result that on the event $\{N>z\}$, we have
	\[
	t_{z,b}^j < \frac{8T^{\mathrm{diag}}}{b^2}
	\]
	for all $b\in\mathcal{B}$ and $j\in[p]$. The desired result then follows immediately by taking $b = b_*$ and restricting  to the subset $\mathcal{J} \subseteq [p]$. 
	
	Fix $b \in \mathcal{B}$ and $j \in [p]$. Recall from~\eqref{Eq:MaxLR} and Lemma~\ref{Lemma:DefRProcess} the definition of $R_{n,b}^j$ and the recursive relation $R_{n,b}^j=\{R_{n-1,b}^j + b(X_{n}^j-b/2)\} \vee 0$. By the update procedure for $t_{n, b}^j$ in Algorithm~\ref{Alg:changepoint3} and Lemma~\ref{Lemma:AltDefTailLength}, we have
	\begin{equation}
	\label{Eq:Rnb}
	R_{n, b}^j  \begin{cases}
	= 0 &\quad \text{when $n=z-t_{z, b}^j$}, \\
	>0 &\quad \text{when $z-t_{z, b}^j<n\le z$.}
	\end{cases}
	\end{equation}
	We claim that
	\begin{equation}   \label{Eq:InductionR}
	R_{n, b/2}^j \ge \frac{R_{n, b}^j}{2} + \frac{b^2 (n - z + t_{z, b}^j)}{8},
	\end{equation}
	for all $n\in\bigl\{z-t_{z,b}^j,\ldots,z\bigr\}$. To see this, the claim is true when $n=z-t_{z, b}^j$ since  the right hand side of~\eqref{Eq:InductionR} is 0 by~\eqref{Eq:Rnb}. Now, assume~\eqref{Eq:InductionR} is true for some $n=m-1$. Then,
	\begin{align*}
	R_{m, b/2}^j \ge R_{m-1, b/2}^j + \frac{b}{2}\Bigl(X_m^j - \frac{b}{4}\Bigr) &\ge \frac{R_{m-1, b}^j}{2} + \frac{b^2 (m-1 - z + t_{z, b}^j)}{8} + \frac{b}{2}\Bigl(X_m^j - \frac{b}{4}\Bigr) \\
	&= \frac{R_{m, b}^j}2 + \frac{b^2 (m - z + t_{z, b}^j)}8.
	\end{align*}
	This proves the claim by induction.	In particular, on the event $\{N > z\}$, we have $T^{\mathrm{diag}} > R_{z, b/2}^j > b^2t_{z, b}^j/8$ as desired.
\end{proof}
\begin{prop}
	\label{Prop:AverageCaseResidualTail}
	Let $X_1,X_2,\ldots$, $z$, $\theta$, $s$, $a$, $p$ and $N$ be defined as in Proposition~\ref{Thm:Alternative}. There exists a universal constant $C$ and $\beta_0(s)>0$, depending only on $s$, such that for all $\beta < \beta_0(s)$, we have 
	\[
	\mathbb{E}_{z,\theta} \bigl\{ q(s^{-1/2}; X_1,\ldots,X_z,\theta)\bigr\} \leq  \frac{Cs^{1/2}\log(16s^2\beta^{-2}\log_2(2p))\log_2(2p)}{\beta^2}.
	\]
\end{prop}
\begin{proof}
	Recall the definition of $b_*$ in~\eqref{Eq:mathcalJ}. We may assume, without loss of generality that $b_*=\beta /\sqrt{s\log_2(2p)}$ (the case $b_*=-\beta /\sqrt{s\log_2(2p)}$ can be proved in essentially the same way). We first prove the result for sufficiently large $s>s_0$. Recall that $t_{z,b_*}^j = \argmax_{0\leq r\leq z} \sum_{i=z-r+1}^z(X_i^j - b_*/2)$. Define $Z_i := X_{z-i+1}$ for $i \in [z]$ and let $Z_{z+1},Z_{z+2},\ldots\stackrel{\mathrm{iid}}{\sim} \mathcal{N}_p(0,I_p)$ be independent from $Z_1,\ldots,Z_z$. For each $j\in[p]$, let 
	\[
	S_r^j :=  \sum_{i=1}^r \bigl(Z_i^j - b_*/2\bigr) \quad \text{and} \quad \tilde{S}_r^j :=  \sum_{i=1}^r Z_i^j 
	\]
	for $r \in \mathbb{N}$ and define $S_0^j:= \tilde{S}_0^j := 0$. Writing $\xi_0^j:=\argmax_{0\leq r \leq \Delta b_*^{-2}} S_r^j$, $\xi^j := \argmax_{r\in\mathbb{N}_0} S_r^j$, and  $\tilde{\xi}_0^j := \argmax_{0\leq r \leq \Delta b_*^{-2}} \tilde{S}_r^j$, where $\Delta := 8\log(2s)$, we note that like $t_{z,b_*}^j$, these three maxima are also uniquely attained almost surely (see the proof of Lemma~\ref{Lemma:DriftedRandomWalk}). By construction, we have for each $j\in[p]$ that
	\[
	t_{z,b_*}^j = \argmax_{0\leq r\leq z}\sum_{i=z-r+1}^z (X_i^j - b_*/2) = \argmax_{0\leq r\leq z} S_r^j \leq \argmax_{r\in\mathbb{N}_0} S_r^j = \xi^j.
	\]
	Writing $q_\xi(\alpha) := \inf\bigl\{y:|\{j\in\mathcal{J}:\xi^j\leq y\}| \geq \alpha|\mathcal{J}|\bigr\}$ as the empirical $\alpha$-quantile of $(\xi^j:j\in\mathcal{J})$, it follows that $q(\alpha) \leq q_\xi(\alpha)$ and so it suffices to control $\mathbb{E}\{q_\xi(s^{-1/2})\}$ instead of $\mathbb{E}\{q(s^{-1/2})\}$. To this end, we observe that $\bigl\{16\Delta s^{-1/2}b_*^2 < \xi^j \leq  \Delta b_*^{-2}\bigr\} \subseteq \bigl\{16\Delta s^{-1/2}b_*^{-2} < \xi^j_0 \leq  \Delta b_*^{-2}\bigr\}$ and $\tilde{\xi}_0^j \geq \xi_0^j$, and thus 
	\begin{align}
	\label{Eq:DifferenceTwoTerms}
	\mathbb{P}\bigl(\xi^j \leq 16\Delta s^{-1/2}b_*^{-2}\bigr) &\geq \mathbb{P}\bigl(\xi_0^j \leq 16\Delta s^{-1/2}b_*^{-2}\bigr) - \mathbb{P}\bigl(\xi^j > \Delta b_*^{-2}\bigr) \nonumber \\
	&\geq  \mathbb{P}\bigl(\tilde{\xi}_0^j \leq 16\Delta s^{-1/2}b_*^{-2}\bigr) - \mathbb{P}\bigl(\xi^j > \Delta b_*^{-2}\bigr).
	\end{align}
	For the first term on the right hand side of~\eqref{Eq:DifferenceTwoTerms}, by Donsker's invariance principle and the continuity of the argmax map \citep[see, e.g.][Lemma~3.2.1 and Theorem~3.2.2]{vanderVaartWellner1996}, we have in the limit $\beta\searrow 0$ that $\Delta b_*^{-2}\to\infty$ and so
	\[
	\frac{\tilde{\xi}_0^j}{\Delta b_*^{-2}} \stackrel{\mathrm{d}}{\to} \argmax_{t\in[0,1]}B_t,
	\]
	where $(B_t)_{t\geq 0}$ denotes a standard Brownian motion. In particular, we can find $\beta_0(s)>0$ depending only on $s$ such that for $\beta\leq\beta_0(s)$ and $s > 256$, we have
	\begin{equation}
	\label{Eq:FirstTerm}
	\mathbb{P}\bigl(\tilde{\xi}_0^j\leq 16\Delta s^{-1/2}b_*^{-2}\bigr) \geq \frac{1}{2}\mathbb{P}\Bigl(\argmax_{t\in[0,1]} B_t \leq 16s^{-1/2}\Bigr) = \frac{1}{\pi}\arcsin(4s^{-1/4}) \geq \frac{4s^{-1/4}}{\pi}.
	\end{equation}
	where in the second step we used the arcsine law for Brownian motion \citep[see, e.g.][Theorem~5.26]{MortersPeres2010}, and in the final step we used the fact that $4s^{-1/4}< 1$.
	
	For the second term on the right-hand side of~\eqref{Eq:DifferenceTwoTerms}, since $\Delta = 8\log(2s)$, for sufficiently large $s \geq s_0$ and sufficiently small $\beta \leq \beta_0(s)$, we have by Lemma~\ref{Lemma:DriftedRandomWalk}(d) that
	\begin{equation}
	\label{Eq:SecondTerm}
	\mathbb{P}(\xi^j > \Delta b_*^{-2}) \leq 2e^{-\Delta/8} = s^{-1}.
	\end{equation}
	Substituting~\eqref{Eq:FirstTerm} and~\eqref{Eq:SecondTerm} into~\eqref{Eq:DifferenceTwoTerms}, we have, for all $j\in\mathcal{J}$, that
	\[
	\mathbb{P}\bigl(\xi^j \leq 16\Delta s^{-1/2}b_*^{-2}\bigr) \geq s^{-1/4}.
	\]
	As a result, $\bigl|\bigl\{j\in\mathcal{J}:\xi^j \leq 16\Delta s^{-1/2}b_*^{-2}\bigr\}\bigr|$ is stochastically larger than $\mathrm{Bin}\bigl(|\mathcal{J}|, s^{-1/4}\bigr)$. Thus, for $s \geq s_0$, we have, 
	\[
	\mathbb{P}_{z,\theta}\bigl\{q_\xi(s^{-1/2}) > 16\Delta s^{-1/2}b_*^{-2}\bigr\}  \leq  \mathbb{P}\Bigl\{\mathrm{Bin}\bigl(|\mathcal{J}|, s^{-1/4}\bigr) \leq s^{-1/2} |\mathcal{J}|\Bigr\} \leq e^{-s^{1/2}/2},
	\]
	where we have used Hoeffding's inequality and the fact that $|\mathcal{J}|\geq s/2$ in the last step. On the other hand, for sufficiently large $s \geq s_0$ and sufficiently small $\beta \leq \beta_0(s)$, we have,
	\begin{align*}
	\mathbb{E}_{z,\theta}&\Bigl\{q_\xi(s^{-1/2}) \Bigm | q_\xi(s^{-1/2}) > 16\Delta s^{-1/2}b_*^{-2}\Bigr\} \leq \mathbb{E}_{z,\theta}\Bigl\{q_\xi(s^{-1/2}) \Bigm | q_\xi(s^{-1/2}) \geq \Delta b_*^{-2}\Bigr\} \\
	&\leq \mathbb{E}_{z,\theta}\Bigl\{q_\xi(1) \Bigm | q_\xi(|\mathcal{J}|^{-1}) \geq \Delta b_*^{-2}\Bigr\} = \mathbb{E}_{z,\theta}\Bigl\{\max_{j\in\mathcal{J}} \xi^j\bigm | \min_{j\in\mathcal{J}} \xi^j \geq \Delta b_*^{-2}\Bigr\} \leq \frac{61\bigl(\Delta+4\log(2/b_*)\bigr)}{b_*^2},
	\end{align*}
	where we have used Lemma~\ref{Lemma:Monotone}(b) in the second inequality and Lemma~\ref{Lemma:DriftedRandomWalk}(d) (with $\Delta/4$ taking the role of $k$ and $b_*/2$ taking the role of $b$ there) in the final inequality. As a result,
	\begin{align*}
	\mathbb{E}_{z,\theta}\bigl\{q(s^{-1/2})\bigr\}&\leq \mathbb{E}_{z,\theta}\bigl\{q_\xi(s^{-1/2})\bigr\} \leq 16\Delta s^{-1/2}b_*^{-2}  + 61e^{-s^{1/2}/2}\bigl(\Delta+4\log(2/b_*)\bigr)b_*^{-2} \\
	&\leq \frac{Cs^{1/2}\log(16s^2\beta^{-2}\log_2(2p))\log_2(2p)}{\beta^{2}},
	\end{align*}
	where we have used in the final step the fact that $e^{-s^{1/2}/2} \leq s^{-1/2}/100$ for sufficiently large $s$. This proves the desired result for $s\geq s_0$. 
	
	Finally, for $s\leq 256$, we have by Lemma~\ref{Lemma:DriftedRandomWalk}(c) that, for $\beta < \sqrt{s}/2$,
	\begin{align*}
	\mathbb{E}_{z,\theta}\bigl\{q(s^{-1/2})\bigr\} \leq \mathbb{E}_{z,\theta}\Bigl\{\max_{j\in\mathcal{J}} \xi^j \Bigr\} &\leq \frac{32s\log(s^{3/2}\beta^{-1}\log_2^{1/2}(2p))\log_2(2p)}{\beta^2} \\
	&\leq \frac{Cs^{1/2}\log(16s^2\beta^{-2}\log_2(2p))\log_2(2p)}{\beta^2},
	\end{align*}
	and the desired bound then follows.
\end{proof}
We are now in a position to prove Theorem~\ref{Thm:DenseMain}.
\begin{proof}[of Theorem~\ref{Thm:DenseMain}]
	The proof proceeds with different arguments for the case  $s\geq 2$ and the case $s=1$.
	
	\underline{Case 1: $s\geq 2$}. Combining Propositions~\ref{Thm:Alternative} (applied with $\alpha=1$) and~\ref{Prop:WorstCaseResidualTail}, we have
	\[
	\bar{\mathbb{E}}_\theta^{\mathrm{wc}}(N) \leq \frac{396\tilde T^{\mathrm{off}} + 65\sqrt{p \tilde T^{\mathrm{off}}}}{\vartheta^2} +  \frac{24 \log_2(2p)}{\beta^2} + \frac{24T^{\mathrm{diag}}s\log_2(2p)}{\beta^2} + 2.
	\]
	The desired bound~\eqref{Eq:DenseMain1} then follows by substituting in the expression for $\tilde{T}^{\mathrm{off}}$. On the other hand, combining Propositions~\ref{Thm:Alternative} (applied with $\alpha=s^{-1/2}$) and~\ref{Prop:AverageCaseResidualTail}, we have
	\[
	\bar{\mathbb{E}}_\theta(N) \leq \frac{396\tilde T^{\mathrm{off}} + 65\sqrt{p \tilde T^{\mathrm{off}}}}{\vartheta^2} +  \frac{24\sqrt{s} \log_2(2p)}{\beta^2} + \frac{3Cs^{1/2}\log(16s^2\beta^{-2}\log_2(2p))\log_2(2p)}{\beta^2}+2,
	\]
	which proves~\eqref{Eq:DenseMain2}.
	
	\underline{Case 2: $s=1$}. There exists $j_*\in [p]$ such that $|\theta^{j_*}| \ge \vartheta/\sqrt{\log_2(2p)}$, and recall from~\eqref{Eq:mathcalJ} that $b_* := \mathrm{sgn}(\theta^{j_*})\beta/\sqrt{\log_2(2p)} \in \mathcal{B}$. Note that $S_{n, 1}^{\mathrm{diag}} = \max_{(j, b)\in [p]\times (\mathcal{B}\cup \mathcal{B}_0)}R_{n, b}^j \ge R_{n, b_*}^{j_*}$.  We define $\bar R_n := \sum_{i=z+1}^{z+n}b_*(X_i^{j_*}-b_*/2)$ for $n\in\mathbb{N}_0$. Since $R_{z,b_*}^{j_*}\geq 0 = \bar{R}_0$ and $R_n - R_{n-1} = b_*(X_{z+n}^{j_*}-b_*/2) \leq R_{z+n-1,b_*}^{j_*} - R_{z+n-1,b_*}^{j_*}$, it follows by induction that $R_{z+n,b_*}^{j_*}\geq \bar R_n$ for all $n\in\mathbb{N}_0$.   Then, for $n \geq \lceil 4T^{\mathrm{diag}}/(b_*\theta^{j_*})\rceil =:n_0$, we have
	\begin{align*}
	\mathbb{P}_{z,\theta}(N > z +n &\mid X_1=x_1,\ldots,X_z=x_z) \leq \mathbb{P}_{z,\theta}\bigl(R_{z+n, b_*}^{j_*} \leq T^{\mathrm{diag}} \bigm| X_1=x_1,\ldots,X_z=x_z\bigr) \\
	&\leq \mathbb{P}_{z,\theta}\bigl(\bar{R}_n \leq T^{\mathrm{diag}}\bigr) = \Phi\biggl(-\frac{b_*n(\theta^{j_*} - b_*/2) - T^{\mathrm{diag}}}{n^{1/2}b_*}\biggr)  \\
	&\leq \frac{1}{2}\exp\biggl\{-\frac{(b_*n\theta^{j_*}/2 - T^{\mathrm{diag}})^2}{2nb_*^2}\biggr\} \leq \frac{1}{2}e^{-n(\theta^{j_*})^2/32}.
	\end{align*}
	Therefore, 
	\begin{align} \label{Eq:s=1result}
	\mathbb{E}_{z, \theta}&\bigl\{(N-z)\vee 0 \mid X_1=x_1, \ldots , X_z=x_z\bigr\} = \sum_{n=0}^\infty \mathbb{P}_{z, \theta}(N > z + n \mid X_1=x_1 , \ldots, X_z=x_z) \nonumber \\
	&\leq n_0 + \frac{1}{2}\sum_{n=n_0}^\infty e^{-n(\theta^{j_*})^2/32} \leq n_0 + \frac{1}{2}
	\int_{0}^{\infty} e^{ -u(\theta^{j_*})^2/32} \, du  \leq 1+ \frac{4T^{\mathrm{diag}}}{b_*\theta^{j_*}} + \frac{16}{(\theta^{j_*})^2}.
	\end{align}
	After substituting in the expressions for $b_*$, $\theta^{j_*}$ and $T^{\mathrm{diag}}$, we see that
	\[ 
	\bar{\mathbb{E}}_\theta(N) \leq \bar{\mathbb{E}}_\theta^{\mathrm{wc}}(N) \leq 1 + \frac{4\log(16p\gamma \log_2(4p))\log_2(2p)}{\beta\vartheta}+ \frac{16\log_2(2p)}{\vartheta^2},
	\]  
	which proves both~\eqref{Eq:DenseMain1} and~\eqref{Eq:DenseMain3}.
\end{proof}

\subsection{Proofs from Sections~\ref{SubSec:SparseAnalysis} and~\ref{SubSec:AdaptiveTheory}}

\begin{proof}[of Theorem~\ref{Thm:NullThresh}]
	It suffices to only prove $\mathbb{P}_0(N^{\mathrm{off}}\le m) \leq  1/4$, since the remaining proof is identical to that of Theorem~\ref{Thm:Null}.
	
	Since $\Lambda^{k, j}_b \mid \tau_b^j \stackrel{\mathrm{iid}}{\sim} \mathcal{N}(0, \tau_b^j)$ for all $b \in \mathcal{B}, j \in [p]$ and $k \in [p]\backslash \{j\}$ under the null, by the fact that $T^{\mathrm{off}} \geq 12$ and Lemma~\ref{Lemma:ThreshChi}, we have
	\begin{align*}
	\mathbb{P}_0\bigl(Q_{n,b}^j  \geq T^{\mathrm{off}} \bigm | \tau_{n,b}^j \bigr) \leq \mathbb{P}_0\bigl(Q_{n,b}^j  \geq 6 + T^{\mathrm{off}}/2 \bigm | \tau_{n,b}^j \bigr) \leq \exp(-T^{\mathrm{off}}/8).
	\end{align*}
	Hence, it follows that 
	\begin{equation} \label{Eq:Noffprob_s}
	\mathbb{P}_0(N^{\mathrm{off}}\le m )  \leq |\mathcal{B}|mpe^{-T^{\mathrm{off}}/8} \leq 1/4,
	\end{equation}
	as desired.
\end{proof}

\begin{proof}[of Theorem~\ref{Thm:AlternativeThresh}]
	We note that the case $s=1$ in the proof of Theorem~\ref{Thm:DenseMain} does not rely on the off-diagonal statistics. Hence~\eqref{Eq:s=1result} is still valid here with $a=\sqrt{8\log(p-1)}$ and the last expression in~\eqref{Eq:s=1result} again proves the desired bound~\eqref{Eq:SparseMain}. For the case $s \geq 2$, we follow exactly the proof of Proposition~\ref{Thm:Alternative} until~\eqref{Eq:Lambdarj}, with the only exception that we now fix, instead of~\eqref{Eq:AssumptionOnr}, 
	\begin{equation} \label{Eq:AssumptiononrNEW}
	r \ge \biggl\{\frac{24T^{\mathrm{off}}\log_2(2p)}{\vartheta^2} \vee \frac{96s \log_2(2p)\log p}{\vartheta^2}\vee 3q(\alpha) \biggr\} + 2 :=\tilde{r}_0.
	\end{equation}
	By the definition of the effective sparsity of $\theta$, for a fixed $j \in \mathcal{J}_\alpha$,
	\begin{equation*}
	{\mathcal{L}}^{j}:=\biggl\{  j^\prime \in [p]: |\theta^{j^\prime}| \geq \frac{\vartheta}{\sqrt{s\log_2(2p)}}\text{ and } j^\prime \neq j  \biggr \}
	\end{equation*}
	has cardinality at least $s-1$. On the event $\Omega_{r}^{j}$, we have, by~\eqref{Eq:AuxTailBounds}, that for all $k \in	{\mathcal{L}}^{j} $
	\[
	|\theta^k| \sqrt{\tau_{z + \lfloor r \rfloor, b_*}^{j}} \geq \sqrt{\frac{\vartheta^2\lfloor r \rfloor}{3s\log_2(2p)}}=:\tilde{a}_r.
	\]
	We then observe, by~\eqref{Eq:AssumptiononrNEW}, that 
	\begin{equation} \label{Eq:tildea_r}
	\tilde{a}_r \geq \sqrt{32 \log p} > 2a.
	\end{equation}
	Now, from~\eqref{Eq:Lambdarj} we have on the event $\Omega_{r}^{j}$ that, 
	for all $k \in {\mathcal{L}}^{j}$,
	\begin{align*}
	\mathbb{P}_{z, \theta} \Biggl(\Omega_{r}^{j}  \cap \biggl\{|\Lambda^{k, j}_{z+\lfloor r \rfloor, b_*}| < \frac{1}{2} \tilde{a}_r \sqrt{\tau^{j}_{z+\lfloor r \rfloor, b_*}}\biggr\}  \Biggm | &\tau_{z + \lfloor r \rfloor, b_*}^{j}, X_1=x_1,\ldots,X_z=x_z  \Biggr) \leq \frac{1}{2}e^{-\tilde{a}_r^2/8} =: q_r.
	\end{align*}
	We denote
	\[
	U^{j}:=\Biggl| \Biggl\{ k\in {\mathcal{L}}^{j}: \biggl\{|\Lambda^{k, j}_{z+\lfloor r \rfloor, b_*}| < \frac{1}{2} \tilde{a}_r\sqrt{\tau^{j}_{z+\lfloor r \rfloor, b_*}}\biggr\}  \Biggr\}  \Biggr|.
	\]
	Then, by the Chernoff--Hoeffding binomial tail bound \citep[][Equation~(2.1)]{Hoeffding1963}, we have
	\begin{align} \label{Eq:HalfBinomial}
	\mathbb{P}_{z, \theta} \Bigl(\Omega_{r}^{j}  &\cap \bigl\{ U^{j} \geq |{\mathcal{L}}^{j}|/2 \bigr\}  \Bigm | \tau_{z + \lfloor r \rfloor, b_*}^{j}, X_1=x_1,\ldots,X_z=x_z  \Bigr) \leq \exp\biggl\{ -\frac{|{\mathcal{L}}^{j}|}{2} \log\biggl(\frac{1}{4q_r(1-q_r)}\biggr)\biggr\} \nonumber\\
	&\leq \exp\biggl\{ |{\mathcal{L}}^{j}| \biggl( \frac{\log 2 }{2} -  \frac{\tilde{a}_r^2}{16}\biggr) \biggr \} \leq \exp\biggl\{ - \frac{ 3|{\mathcal{L}}^{j}| \tilde{a}_r^2 }{64}\biggr\} \leq \exp \biggl\{   -\frac{\vartheta^2\lfloor r \rfloor}{128 \log_2(2p)}   \biggr\}, 	
	\end{align}
	where the penultimate inequality follows from~\eqref{Eq:tildea_r}. Now, on the event $\Omega_{r}^{j}  \cap \bigl\{ U^{j} < |{\mathcal{L}}^{j}|/2 \bigr\}$, we have
	\begin{align}  \label{Eq:AboveThreshold}
	\sum_{j' \in [p]: j' \neq j}\frac{\bigl(\Lambda^{j', j}_{z+\lfloor r \rfloor, b_*}\bigr)^2}{\tau^j_{z+\lfloor r \rfloor, b_*} \vee 1} &\mathbbm{1}_{\bigl\{|\Lambda^{j', j}_{z+\lfloor r \rfloor, b_*}| \geq a\sqrt{\tau^j_{z+\lfloor r \rfloor, b_*}}\bigr\}}  \geq  \sum_{j' \in [p]: j' \neq j}\frac{\bigl(\Lambda^{j', j}_{z+\lfloor r \rfloor, b_*}\bigr)^2}{\tau^j_{z+\lfloor r \rfloor, b_*} \vee 1} \mathbbm{1}_{\bigl\{|\Lambda^{j', j}_{z+\lfloor r \rfloor, b_*}| \geq \frac{\tilde{a}_r}{2}\sqrt{\tau^j_{z+\lfloor r \rfloor, b_*}}\bigr\}}  	  \nonumber \\
	&\geq \frac{\tilde{a}_r^2}{4} \biggl\{|{\mathcal{L}}^{j}| - \Bigl(  \Bigl\lceil \frac{ |{\mathcal{L}}^{j}|}{2} \Bigr\rceil - 1\Bigr)\biggr\} =  \frac{\tilde{a}_r^2}{4} \biggl\lceil \frac{|{\mathcal{L}}^{j}|+1}{2} \biggr\rceil \geq \frac{\vartheta^2\lfloor r \rfloor}{24\log_2(2p)} \geq T^{\mathrm{off}},
	\end{align}
	where the penultimate inequality uses the fact that $|{\mathcal{L}}^{j}| \geq s-1$ and the last inequality follows from~\eqref{Eq:AssumptiononrNEW}. We now denote
	\[
	\tilde{E}_r^j := \Biggl\{    \sum_{j' \in [p]: j' \neq j}\frac{\bigl(\Lambda^{j', j}_{z+\lfloor r \rfloor, b_*}\bigr)^2}{\tau^j_{z+\lfloor r \rfloor, b_*} \vee 1} \mathbbm{1}_{\bigl\{|\Lambda^{j', j}_{z+\lfloor r \rfloor, b_*}| \geq a\sqrt{\tau^j_{z+\lfloor r \rfloor, b_*}}\bigr\}}   <  T^{\mathrm{off}}   \Biggr \}.
	\]
	Combining~\eqref{Eq:SmallProbShortTail}, \eqref{Eq:HalfBinomial} and~\eqref{Eq:AboveThreshold}, we deduce that
	\begin{align*}
	\mathbb{P}_{z, \theta}&\bigl(N > z +  r  \bigm| X_1=x_1 , \ldots, X_z=x_z\bigr) \le \mathbb{P}_{z, \theta}\bigl(N > z + \lfloor r \rfloor \bigm| X_1=x_1 , \ldots, X_z=x_z\bigr) \\
	&\leq  \mathbb{P}_{z, \theta}\biggl( \bigcap_{j \in \mathcal{J}_\alpha}(\Omega_{r}^j)^c  \biggm | X_1=x_1, \ldots, X_z=x_z  \biggr) +  \sum_{j \in \mathcal{J}_\alpha} \mathbb{P}_{z, \theta} \Bigl(\tilde{E}_r^{j}  \cap \Omega_{r}^{j} \Bigm | X_1=x_1,\ldots,X_z=x_z  \Bigr)  \\
	&\leq \mathbb{P}_{z, \theta}\biggl( \bigcap_{j \in \mathcal{J}_\alpha}(\Omega_{r}^j)^c  \biggm | X_1=x_1, \ldots, X_z=x_z  \biggr) +  \\
	&\qquad \qquad \qquad \qquad \sum_{j \in \mathcal{J}_\alpha} \mathbb{P}_{z, \theta} \Bigl(\Omega_{r}^{j}  \cap \bigl\{ U^{j} \geq |{\mathcal{L}}^{j}|/2 \bigr\}  \Bigm |X_1=x_1,\ldots,X_z=x_z  \Bigr)  \\
	&\leq \exp \biggl \{   -\frac{\alpha sb_*^2(r-1)}{24} \biggr\} + p\exp\biggl\{ - \frac{\vartheta^2(r-1) }{128\log_2(2p)}\biggr\}.
	\end{align*}
	Therefore we have
	\begin{align*}
	&\mathbb{E}_{z, \theta}\bigl\{(N-z) \vee 0 \bigm| X_1=x_1, \ldots , X_z=x_z\bigr\} = \int_{0}^{\infty} \mathbb{P}_{z, \theta}\bigl(N > z + u \bigm| X_1=x_1 , \ldots, X_z=x_z\bigr)\, du \\
	&\leq \tilde{r}_0 + \int_{\tilde{r}_0-1}^{\infty}\biggl[\exp \biggl \{   -\frac{\alpha sb_*^2u}{24} \biggr\} +p\exp\biggl\{ - \frac{\vartheta^2 u }{128\log_2(2p)}\biggr\}\biggr] \wedge 1 \, du \\
	&\leq \tilde{r}_0 + \frac{24}{\alpha sb_*^2} +  \int_{0}^{\infty}  \Bigl(pe^{-\frac{\vartheta^2 u}{128\log_2(2p)}}\Bigr) \wedge 1 \, du \leq \tilde{r}_0 + \frac{24}{\alpha sb_*^2} + \frac{128\log_2(2p)\log(ep)}{\vartheta^2} \\
	&\leq \frac{24T^{\mathrm{off}}\log_2(2p)+96s \log_2(2p)\log p}{\vartheta^2} + 3q(\alpha)+ \frac{24\log_2(2p)}{\alpha\beta^2} + \frac{128\log_2(2p)\log(ep)}{\vartheta^2} +2.
	\end{align*}
	Combining this with Proposition~\ref{Prop:WorstCaseResidualTail} (applied with $\alpha=1$), we have, by  substituting in the expression for $T^{\mathrm{off}}$, that
	\[
	\bar{\mathbb{E}}_\theta(N) \leq \bar{\mathbb{E}}_\theta^{\mathrm{wc}}(N) \leq C\biggl\{\frac{s\log (ep\gamma)\log(ep)}{\beta^2} \vee 1\biggr\},
	\]
	for some universal constant $C>0$, as desired.
\end{proof}
\begin{proof}[of Theorem~\ref{Thm:AdaptiveNull}]
	Let $T^{\mathrm{off,d}}=\psi(\tilde{T}^{\mathrm{off,d}})$. Then, similar to~\eqref{Eq:Noffprob_d}, \eqref{Eq:NdiagProb} and~\eqref{Eq:Noffprob_s}, we have 
	\begin{align*}
	\mathbb{P}_0(N^{\mathrm{diag}}\leq m) &\leq mp|\mathcal{B} \cup \mathcal{B}_0|e^{-T^{\mathrm{diag}}} \leq 1/6, \\
	\mathbb{P}_0(N^{\mathrm{off,d}}\leq m) &\leq mp|\mathcal{B}|e^{-\tilde{T}^{\mathrm{off,d}}/2} \leq 1/6, \\
	\mathbb{P}_0(N^{\mathrm{off,s}}\leq m) &\leq mp|\mathcal{B}|e^{-T^{\mathrm{off,s}}/8} \leq 1/6.
	\end{align*}	
	and hence,
	\begin{align*}
	\mathbb{E}_0(N) &= \mathbb{E}_0(N^{\mathrm{diag}}\wedge N^{\mathrm{off,d}} \wedge  N^{\mathrm{off,s}}) \geq 2\gamma\mathbb{P}_0(N^{\mathrm{diag}}  \wedge N^{\mathrm{off,d}}\wedge N^{\mathrm{off,s}}> 2\gamma) \\
	&\geq 2\gamma\big\{1 - \mathbb{P}_0(N^{\mathrm{diag}}\le m) - \mathbb{P}_0(N^{\mathrm{off,d}}\le m) - \mathbb{P}_0(N^{\mathrm{off,s}}\le m) \bigr\} \geq \gamma,
	\end{align*}
	as desired.
\end{proof}

\begin{proof}[of Theorem~\ref{Thm:AdaptiveAlt}]
	We observe that
	\begin{align*}
	\bar{\mathbb{E}}_\theta^{\mathrm{wc}}(N) & = \bar{\mathbb{E}}_\theta^{\mathrm{wc}}\bigl[(N^{\mathrm{diag}} \wedge N^{\mathrm{off,d}}) \wedge (N^{\mathrm{diag}} \wedge N^{\mathrm{off,s}}) \bigr] \\
	&\leq 
	\bar{\mathbb{E}}_\theta^{\mathrm{wc}}\bigl[(N^{\mathrm{diag}} \wedge N^{\mathrm{off,d}})\bigl] \wedge \bar{\mathbb{E}}_\theta^{\mathrm{wc}}\bigl[(N^{\mathrm{diag}} \wedge N^{\mathrm{off,s}})\bigl],
	\end{align*}
	and similarly for $\bar{\mathbb{E}}_\theta(N)$. The desired bounds~\eqref{Eq:AdaptiveMain1}, \eqref{Eq:AdaptiveMain2} and~\eqref{Eq:AdaptiveMain3} are therefore direct consequences of Theorems~\ref{Thm:DenseMain} and~\ref{Thm:AlternativeThresh} (note that the different constants in the thresholds only affect the value of the universal constant).
\end{proof}
	
\section{Auxiliary results}
\label{Sec:Auxiliary}

\begin{lemma} \label{Lemma:DefRProcess}
	For $n \in \mathbb{N}_0$, $b \in \mathcal{B}\cup\mathcal{B}_0$ and $j \in [p]$, we define $R_{n,b}^j := b A_{n,b}^{j,j} - b^2 t_{n,b}^j/2$, where $A_{n,b}$ and $t_{n,b}$ are taken from Algorithm~\ref{Alg:changepoint3} in the main text. Then 
	\begin{equation} \label{Eq:AltDefRProcess}
		R_{n, b}^j = \max_{0 \le h \le n} \sum_{i=n-h+1}^{n} {b(X_i^j - b/2)}.
	\end{equation}
\end{lemma}
\begin{proof}
	We prove the claim by induction on $n$. The base case $n=0$ is true since, by definition, $R_{0,b}^j = 0$ and the sum on the right-hand side of~\eqref{Eq:AltDefRProcess} is empty. Assume~\eqref{Eq:AltDefRProcess} is true for $n=m-1$. Then, by the update procedure in Algorithm~\ref{Alg:changepoint3} in the main text, we have
	\begin{align*}
		R_{m, b}^j &= \bigl\{R_{m-1,b}^j + b(X_{m}^j-b/2)\bigr\} \vee 0 = \biggl\{\max_{0 \le h \le m-1} \sum_{i=m-h}^{m-1} {b(X_i^j - b/2)}+ b(X_{m}^j-b/2)\biggr\} \vee 0 \\
		&= \biggl\{ \max_{0\le h \le m-1} \sum_{i=m-h}^{m} {b(X_i^j - b/2)} \biggr\} \vee 0 =  \max_{0 \le h \le m} \sum_{i=m-h+1}^{m} {b(X_i^j - b/2)},
	\end{align*}
	and the desired result follows.
\end{proof}

\begin{lemma} \label{Lemma:AltDefTailLength}
	For $n \in \mathbb{N}_0$, $b \in \mathcal{B}\cup\mathcal{B}_0$ and $j \in [p]$, let $t_{n, b}^j$ be defined as in Algorithm~\ref{Alg:changepoint3} in the main text and $R_{n,b}^j $ as in Lemma~\ref{Lemma:DefRProcess}. Then 
	\begin{equation}
		\label{Eq:tnbj}
		t_{n, b}^j = \min\bigl\{0 \le i \le n: R_{n-i, b}^j =0\bigr\} = \sargmax_{0 \le h \le n} \sum_{i=n-h+1}^{n} {b(X_i^j - b/2)}.
	\end{equation}
\end{lemma}
\begin{proof}
	We observe from the procedure in Algorithm~\ref{Alg:changepoint3} in the main text that $R_{n, b}^j=0$ if and only if  $t_{n, b}^j=0$ and that $R_{n, b}^j>0$ if and only if  $t_{n, b}^j=t_{n-1, b}^j+1$. Hence, 
	\[
	t_{n, b}^j =  n- \max\bigl\{0 \le i \le n: R_{i, b}^j =0\bigr\} = \min\bigl\{0 \le i \le n: R_{n-i, b}^j =0\bigr\}.
	\]
	We now prove that $t_{n, b}^j = \sargmax_{0 \le h \le n} \sum_{i=n-h+1}^{n} {b(X_i^j - b/2)}$ by induction on $n$. The base case $n=0$ is true because $t_{n,b}^j = 0$, and the sum on the right-hand side of~\eqref{Eq:tnbj} is empty. Assume the claim is true for $n=m-1$. Then, by the inductive hypothesis and Lemma~\ref{Lemma:DefRProcess},
	\begin{align*}
		t_{m, b}^j &= (t_{m-1, b}^j+1) \mathbbm{1}_{\{R_{m, b}^j > 0\}} =\biggl(\sargmax_{0 \le h \le m-1} \sum_{i=m-h}^{m-1} {b(X_i^j - b/2)}+1\biggr) \mathbbm{1}_{\{R_{m, b}^j > 0\}}  \\
		&= \biggl(\sargmax_{1 \le h \le m} \sum_{i=m-h+1}^{m} {b(X_i^j - b/2)}\biggr) \mathbbm{1}_{ \bigl\{\max_{0 \le h \le m} \sum_{i=m-h+1}^{m} {b(X_i^j - b/2)} > 0\bigr\}} \\
		&= \sargmax_{0 \le h \le m} \sum_{i=m-h+1}^{m} {b(X_i^j - b/2)},
	\end{align*}
	and the desired result follows.
\end{proof}
For two distributions $P_0$ and $P_1$ on the same measurable space, the sequential probability ratio test of $H_0: X_1,X_2, \ldots \stackrel{\mathrm{iid}}{\sim} P_0$ against $ H_1: X_1,X_2,\ldots \stackrel{\mathrm{iid}}{\sim} P_1$ with log-boundaries $a>0$ and $b<0$ is defined as the (extended) stopping time
\[
N := \inf\biggl\{n: \sum_{i=1}^n \log\frac{dP_1}{dP_0}(X_i) \not\in (b,a)\biggr\},
\]
together with the decision rule after stopping that accepts $H_0$ if $\sum_{i=1}^N\log\{(dP_1/dP_0)(X)\} \leq b$ and accepts $H_1$ if $\sum_{i=1}^N\log\{(dP_1/dP_0)(X)\} \geq a$.
\begin{lemma}
	\label{Lemma:OneSidedHitting}
	Suppose $N$ is the stopping time associated with the (one-sided) sequential probability ratio test of $H_0: X_1,X_2, \ldots \stackrel{\mathrm{iid}}{\sim} P_0$ against $ H_1: X_1,X_2,\ldots \stackrel{\mathrm{iid}}{\sim} P_1$ with log-boundaries $a > 0$ and $b = -\infty$. Then
	\[
	\mathbb{P}_0(N < \infty) \leq e^{-a}.
	\]
\end{lemma}
\begin{proof}
	Let $L_n := \prod_{i=1}^n (dP_1/dP_0)(X_i)$. On the event $\{N<\infty\}$, we have $L_N \geq e^a$. Therefore,
	\[
	\mathbb{P}_0(N<\infty) = \sum_{n=1}^\infty \mathbb{P}_0(N=n) \leq  e^{-a}\sum_{n=1}^\infty \mathbb{E}_0(L_n\mathbbm{1}_{\{N=n\}}) = e^{-a} \sum_{n=0}^\infty \mathbb{P}_1(N=n) \leq e^{-a},
	\]
	which proves the desired result.
\end{proof}
\begin{lemma}
	\label{Lemma:CondionalProb}
	Let $X$, $Y$ and $Z$ be real-valued random variables. Assume that $(X,Y)$ and $Z$ are independent. Let $P_{Z|Z\leq Y}$ be the conditional distribution of $Z$ given $Z \leq Y$. Then
	\[ \mathbb{P}(X \geq Z \mid  Y\geq Z) = \int_{\mathbb{R}} \mathbb{P}(X\geq u \mid Y \geq u)\, dP_{Z|Z\leq Y}(u). \]  
\end{lemma}
\begin{proof}
	Let $P_Y$ and $P_Z$ denote the marginal distribution of $Y$ and $Z$ respectively. Then, by the definition of $P_{Z|Z\leq Y}$, we have
	\begin{align*}
	P_{Z|Z\leq Y}\bigl([u, \infty)\bigr) &= \mathbb{P}(Z \geq u \mid Z \leq Y) = \frac{\mathbb{P}(Y \geq Z\geq u)}{\mathbb{P}(Y \geq Z)}=\frac{\int_{\mathbb{R}^2}\mathbbm{1}_{\{y\geq z\geq u\}}\,dP_Y(y)dP_Z(z)}{\mathbb{P}(Y \geq Z)} \\
	&=\frac{\int_u^{\infty}\mathbb{P}(Y\geq z)\,dP_Z(z)}{\mathbb{P}(Y \geq Z)},
	\end{align*}
	where we have used the assumption that $Y$ and $Z$ are independent in the penultimate equality. Hence,
	\begin{align*}
	\int_{\mathbb{R}} \mathbb{P}(X\geq u \mid Y \geq u)\, dP_{Z|Z\leq Y}(u) &= \int_{\mathbb{R}} \mathbb{P}(X\geq u \mid Y \geq u)\frac{\mathbb{P}(Y \geq u)}{\mathbb{P}(Y\geq Z)}\,dP_Z(u) \\
	&= \frac{\int_{\mathbb{R}}\mathbb{P}(X\geq u, Y\geq u)\,dP_Z(u) }{\mathbb{P}(Y\geq Z)}=\frac{\mathbb{P}(X\geq Z, Y\geq Z)}{\mathbb{P}(Y \geq Z)}\\
	&=\mathbb{P}(X\geq Z\mid Y\geq Z),
	\end{align*}
	where we have used the assumption that $(X,Y)$ and $Z$ are independent in the penultimate equality.
\end{proof}
The proof of Lemma~\ref{Lemma:DriftedRandomWalk} below relies on the following result, due to \citet{Groeneboom1989}.  It involves the Airy function $\mathrm{Ai}$, defined for $x \in \mathbb{R}$ by
\[
\mathrm{Ai}(x) := \frac{1}{\pi} \lim_{b \rightarrow \infty} \int_0^b \cos\biggl(\frac{t^3}{3} + xt\biggr) \, dt.
\]
\begin{lemma}[Corollary~3.4 of \citealt{Groeneboom1989}]
	\label{Lemma:Groeneboom}
	Let $(W_t)_{t \in \mathbb{R}}$ be a two-sided standard Brownian motion and $Z:=\argmax_{t\in \mathbb{R}} (W_t - t^2)$. Then $Z$ has a density $f_Z$ on $\mathbb{R}$ which is symmetric about zero, and which satisfies
	\[   f_Z(z) = \frac{1}{2}\frac{4^{4/3}|z|}{\mathrm{Ai}'(\tilde{a}_1)} \exp\Bigl( -\frac{2}{3} |z|^3 + 2^{1/3} \tilde{a}_1 |z|  \Bigr)\{1 + o(1)\}  \]
	as $z\rightarrow \infty$, where $\tilde{a}_1 \approx -2.3381$ is the largest zero of the Airy function $\mathrm{Ai}$ and where $\mathrm{Ai}'(\tilde{a}_1) \approx 0.7022$.
	
	In particular, there exists a universal constant $K \geq 1$ such that $f_Z(z) \geq ze^{-z^3}$ for $z \geq K^{1/3}$.
\end{lemma}

We collect in the following lemma some useful bounds on both the maximum and the argmax of a Brownian motion and a Gaussian random walk with a negative drift.
\begin{lemma}
	\label{Lemma:DriftedRandomWalk}
Fix $b > 0$, and let $(Z_t)_{t\geq 0}$ be given by $Z_t = W_t - bt$ for $t \geq 0$, where $(W_t)_{t \geq 0}$ is a standard Brownian motion. Define $\hat{M}:=\sup_{t \geq 0} Z_t$ and $M:=\sup_{r\in \mathbb{N}_0} Z_r$.   
	\begin{enumerate}[label=(\alph*)]
		\item For any $a\geq 0$, we have 
		\[
		  \frac{2\sqrt{ab}}{\sqrt{2\pi}(4ab+1)}e^{-2ab} \leq \mathbb{P}(\hat{M} \geq a) \leq e^{-2ab}
		\]
		and
		\[
		\frac{3\sqrt{ab/2}}{\sqrt{2\pi}(9ab/2+1)}e^{-9ab/4} \mathbbm{1}_{\{a\geq b\}} +  \frac{2b}{\sqrt{2\pi}(4b^2+1)}e^{-2b^2} \mathbbm{1}_{\{a < b\}}\leq \mathbb{P}(M\geq a) \leq e^{-2ab}.
		\]
		\item If $c \geq 0$ satisfies $bc\geq a \geq 0$, then
		\[
		\mathbb{P}\Bigl(\sup_{r \in \mathbb{N}: r \geq c} Z_r \geq a\Bigr) \leq \mathbb{P}\Bigl(\sup_{t \geq c} Z_t \geq a\Bigr)  \leq  \exp\biggl\{\frac{-(bc+a)^2}{2c}\biggr\}.
		\]
		Now let $\hat{\xi}:=\argmax_{t \geq 0} Z_t$ and $\xi:=\argmax_{r\in\mathbb{N}_0} Z_r$. Then $\hat{\xi}$ and $\xi$ are both almost surely unique. Moreover, letting $\xi^1,\ldots,\xi^s$ denote independent copies of $\xi$, we have the following results:
		\item If $b\leq 1/2$, then
		\[
		\mathbb{E}\Bigl(\max_{j\in[s]} \xi^j\Bigr) \leq \frac{8\log(s/b)}{b^2}.
		\]
		\item Taking $K \geq 1$ from Lemma~\ref{Lemma:Groeneboom}, for all $k \geq K$ we have
		\[ e^{-2k}\leq \mathbb{P}(\hat{\xi} \geq kb^{-2}) \leq e^{-k/2}. \]
		Moreover, for each $k \geq K$, there exists $b_0 >0$, depending only on $k$, such that for all $b \leq b_0$ we have
		\begin{equation}
        \label{Eq:TailBoundArgmax}
        \frac{1}{2}e^{-2k}\leq \mathbb{P}(\xi \geq kb^{-2}) \leq 2e^{-k/2}
        \end{equation} 
		and
		\[
		\mathbb{E}\Bigl(\max_{j\in[s]} \xi^j \Bigm| \min_{j\in[s]}\xi^j \geq kb^{-2}\Bigr) \leq \frac{60}{b^2}\bigl\{k+\log(1/b)\bigr\} + sb^5.
		\]
	\end{enumerate}
\end{lemma}
\begin{proof}
		\noindent (a) Since $M \leq \hat{M}$, we have
	\[
	\mathbb{P}(M\geq a) \leq \mathbb{P}(\hat{M} \geq a) = \mathbb{P}\Bigl(\sup_{t \geq 0} (W_t - bt) \geq a\Bigr) = e^{-2ab},
	\]
	where the calculation for the final equality can be found in, e.g. \citet[][Proposition 2.4 and Equation (2.5)]{Siegmund1986}. For the lower bounds, we note that
	\[
	\mathbb{P}(\hat{M} \geq a) \geq \sup_{t\geq 0} \mathbb{P}(Z_t \geq a) = \sup_{t\geq 0} \Phi\biggl(- \frac{a+bt}{\sqrt{t}}\biggr) = \Phi\bigl(- 2\sqrt{ab}\bigr).
	\]
	Similarly, assuming without loss of generality that $a > 0$ (since otherwise the result is clear),
	\[
	\mathbb{P}(M \geq a) \geq \sup_{r\in\mathbb{N}}\mathbb{P}(Z_r\geq a) = \sup_{r\in\mathbb{N}} \Phi\biggl(- \frac{a+br}{\sqrt{r}}\biggr) \geq \Phi\biggl(- \frac{a+br_0}{\sqrt{r_0}}\biggr),
	\]
	where $r_0 = \lceil a/b\rceil \vee 1$.  If $a\geq b$, then using the fact that the function $x\mapsto (a+bx)/\sqrt{x}$ is increasing on $[\sqrt{a/b},\infty)$, we have
	\[
	\frac{a+br_0}{\sqrt{r_0}} \leq \frac{a+b(a/b+1)}{\sqrt{a/b+1}} = 2\sqrt{b} \cdot \frac{a+b/2}{\sqrt{a+b}} \leq 2\sqrt{b}\biggl(a + \frac{b^2/4}{a+b}\biggr)^{1/2} \leq 3\sqrt{ab/2}.
	\]
	On the other hand, if $a<b$, then 
	\[
	\frac{a+br_0}{\sqrt{r_0}} = a+b < 2b.
	\]
	The desired results follow from the bound $\Phi(-x) \geq \frac{x}{\sqrt{2\pi}(x^2+1)}e^{-x^2/2}$ for all $x>0$.
	
	\medskip
	
	\noindent (b) By part (a), we have
	\begin{align*}
		\mathbb{P}\Bigl(\sup_{r \in \mathbb{N}: r \geq c} Z_r \geq a\Bigr) &\leq \mathbb{P}\Bigl(\sup_{t \geq c} Z_t \geq a\Bigr) = \int_{-\infty}^\infty \mathbb{P}\Bigl(\sup_{t\geq c} Z_t \geq a \mid Z_c = x\Bigr) \frac{1}{\sqrt{2\pi c}}e^{-(x+bc)^2/(2c)}\,dx\\
		&\leq \int_{-\infty}^a e^{-2(a-x)b} \frac{1}{\sqrt{2\pi c}}e^{-(x+bc)^2/(2c)}\,dx + \int_{a}^\infty \frac{1}{\sqrt{2\pi c}}e^{-(x+bc)^2/(2c)}\,dx\\
		&= e^{-2ab}\Phi\biggl(-\frac{bc-a}{\sqrt{c}}\biggr) + \Phi\biggl(-\frac{bc+a}{\sqrt{c}}\biggr) \leq \exp\biggl\{\frac{-(bc+a)^2}{2c}\biggr\},
	\end{align*}
	where in the final step we have used the fact that $bc\geq a$ and $\Phi(-x)\leq e^{-x^2/2}/2$ for $x\geq 0$.
	
	\medskip
	
	To prove that $\xi$ is almost surely unique, it suffices to note that
	\[
	\mathbb{P}(\xi \ \text{not unique}) \leq \mathbb{P}\biggl(\bigcup_{r_1, r_2 \in \mathbb{N}_0: r_1 < r_2} \{Z_{r_1} = Z_{r_2}\}\biggr) \leq \sum_{r_1, r_2 \in \mathbb{N}_0: r_1 < r_2} \mathbb{P}\Bigl(Z_{r_2}-Z_{r_1} = 0\Bigr) = 0,
	\]
	since $Z_{r_2}-Z_{r_1} \sim \mathcal{N}\bigl(-b(r_2-r_1),r_2-r_1\bigr)$. To prove that $\hat{\xi}$ is almost surely unique, note that
	\begin{align*}
	&\mathbb{P}(\hat{\xi} \ \text{not unique}) \leq \mathbb{P}\biggl(\bigcup_{q_1, q_2 \in \mathbb{Q}: 0 < q_1 < q_2} \Bigl\{ \max_{t \in [0, q_1]} Z_t = \max_{t \in [q_2, \infty)} Z_t \Bigr\}\biggr) \\
	&\leq \sum_{q_1, q_2 \in \mathbb{Q}: 0 < q_1 < q_2} \mathbb{P}\Bigl(\max_{t \in [0, q_1]} Z_t = \max_{t \in [q_2, \infty)} Z_t \Bigr) \\
	&= \sum_{q_1, q_2 \in \mathbb{Q}: 0 < q_1 < q_2} \mathbb{P}\biggl(\Bigl(\max_{t \in [q_2, \infty]} Z_t - Z_{q_2}\Bigr) =  \bigl( Z_{q_2}-Z_{q_1} \bigr) - \Bigl(\max_{t \in [0, q_1]} Z_t - Z_{q_1}\Bigr) \biggr) = 0,
	\end{align*}
	where we have used the Markov property of $(Z_t)_{t\geq 0}$ for the final equality.
	
	\noindent (c) For any $x\in\mathbb{N}$, we have by two union bounds that
	\begin{align*}
		\mathbb{P}\Bigl(\max_{j\in[s]} \xi^j \geq x\Bigr) &\leq s \sum_{r=x}^\infty \mathbb{P}(\xi=r) \leq s \sum_{r=x}^\infty \mathbb{P}(S_r\geq 0)\\
		& = s\sum_{r=x}^\infty \Phi(-b\sqrt{r}) \leq \frac{s}{2}\sum_{r=x}^\infty e^{-rb^2/2} = \frac{se^{-xb^2/2}}{2(1-e^{-b^2/2})}.
	\end{align*}
	Now define $x_0 := \lceil 4b^{-2}\log(s/b)\rceil$. Then for $b \in (0,1/2]$,
	\begin{align*}
		\mathbb{E}\Bigl(\max_{j\in[s]}\xi^j\Bigr) &= \sum_{x=1}^\infty \mathbb{P}\Bigl(\max_{j\in[s]}\xi^j \geq x\Bigr) \leq x_0-1 + \sum_{x=x_0}^\infty \frac{se^{-xb^2/2}}{2(1-e^{-b^2/2})} \\
		&\leq \frac{4\log(s/b)}{b^2} + \frac{se^{-x_0b^2/2}}{2(1-e^{-b^2/2})^2} \leq \frac{4\log(s/b)}{b^2} + \frac{2}{b^2(1-1/16)^2} \\
		&\leq \frac{8\log(s/b)}{b^2},
	\end{align*}
	where we have used the fact that $1 - e^{-x} \geq 15x/16$ for $x \in [0,1/8]$.
	\medskip
	
	\noindent (d) First note that $W_t- b^3t^2/k \leq W_t-bt$ for $t \geq kb^{-2}$ and $W_t- b^3t^2/k > W_t-bt$ for $t < kb^{-2}$.  Thus, using the fact that $(W_t)_{t \geq 0} \stackrel{\mathrm{d}}{=} \bigl(a^{-1}W_{a^2t}\bigr)_{t \geq 0}$ for every $a > 0$, and taking $K \geq 1$ from Lemma~\ref{Lemma:Groeneboom}, we have for $k \geq K$ that 
	\begin{align}
	\mathbb{P}(\hat{\xi} \geq kb^{-2}) &= \mathbb{P}\biggl(\argmax_{t\geq 0} (W_t-bt) \geq \frac{k}{b^2}\biggr) \geq \mathbb{P}\Biggl(\argmax_{t\geq 0} \biggl(W_t-\frac{b^3t^2}{k}\biggr) \geq \frac{k}{b^2}\Biggr) \nonumber \\
	&=\mathbb{P}\Biggl(\argmax_{t\geq 0} \biggl(\frac{W_{b^2 k^{-2/3}t}}{bk^{-1/3}}-\frac{b^3t^2}{k}\biggr) \geq \frac{k}{b^2}\Biggr)=\mathbb{P}\biggl(\argmax_{t\geq 0} (W_t-t^2) \geq {k}^{1/3}\biggr) \nonumber\\
	&\geq \int_{k^{1/3}}^{\infty} 2z\exp(-z^3)\, dz \geq \int_{k^{1/3}}^{\infty} \biggl(\frac{3z}{2} + \frac{1}{2z^2}\biggr)\exp(-z^3)\, dz = \frac{e^{-k}}{2{k}^{1/3}} \geq e^{-2k}, \label{Eq:BMargmaxLemmac1}
	\end{align}
	where the second inequality follows from Lemma~\ref{Lemma:Groeneboom}. We also have, by part (b), that
	\begin{equation} \label{Eq:BMargmaxLemmac2}
	 \mathbb{P}(\hat{\xi} \geq kb^{-2}) = \mathbb{P}\biggl(\argmax_{t\geq 0} (W_t-bt) \geq kb^{-2}\biggr)\leq \mathbb{P}\biggl(\sup_{t\geq kb^{-2}} (W_t-bt) \geq 0\biggr) \leq e^{-k/2}.
	\end{equation}
	
We now compute upper and lower bounds on the tail probabilities for $\xi$.  By Donsker's invariance principle \citep[][Theorem~5.22]{MortersPeres2010} and the continuity of the argmax map \citep[e.g.][Theorem~3.2.2]{vanderVaartWellner1996}, we have, as $b\rightarrow 0$, that
\[
  b^2\xi \stackrel{\mathrm{d}}{=} b^2 \argmax_{r \in \mathbb{N}_0} \biggl(\frac{W_{rb^2} - rb^2}{b}\biggr) \stackrel{\mathrm{d}}{=} \argmax_{r \in b^2\mathbb{N}_0}(W_r-r) \stackrel{\mathrm{d}}{\rightarrow} \argmax_{t \geq 0} (W_t-t).
\]
	Thus there exists $b_0 > 0$, depending only on $k$, such that for $b<b_0$, we have by~\eqref{Eq:BMargmaxLemmac1} and~\eqref{Eq:BMargmaxLemmac2} that
	\[\mathbb{P}(\xi \geq kb^{-2}) \geq \frac{1}{2}\mathbb{P}\biggl(\argmax_{t\geq 0} (W_t-t) \geq k\biggr) \geq \frac{1}{2}e^{-2k}, \]
	and that
	\[ \mathbb{P}(\xi \geq kb^{-2}) \leq 2\mathbb{P}\biggl(\argmax_{t\geq 0} (W_t-t)\geq k\biggr) \leq 2e^{-k/2}. \]
        We now move on to the final claim of Lemma~\ref{Lemma:DriftedRandomWalk}(d).  For $r \in \mathbb{N}_0$, define $M_r := \max_{r'\in \{0,1,\ldots,r\}} Z_{r'}$ and let $P_r$ denote the conditional distribution of $Z_r- M_r$ given that $\xi \geq r$. Note that $\{\xi \geq r\} = \{\max_{r' \in \mathbb{N}_0: r'\geq r}Z_{r'} \geq M_r\}$ up to a null set. Denote $x_0:= \lfloor 60\{k+\log(1/b)\}/b^2\rfloor$ and $c:=\lceil kb^{-2}\rceil$.  Without loss of generality, we may assume that $b_0<1/2$.  Then for $b \leq b_0$, we have $c < x_0$, so 
        \begin{align*}
          \mathbb{E}\Bigl(\max_{j\in[s]} \xi^j \Bigm| \min_{j\in[s]} \xi^j \geq c\Bigr) -x_0 &\leq \sum_{x=x_0}^\infty \mathbb{P}\Bigl(\max_{j\in[s]} \xi^j\geq x\Bigm| \min_{j\in[s]}\xi^j\geq c\Bigr)\leq s\sum_{x=x_0}^\infty \mathbb{P}\Bigl(\xi^1 \geq x\Bigm| \min_{j \in [s]}\xi^j\geq c\Bigr) \\
          &= s \sum_{x=x_0}^\infty \frac{\mathbb{P}(\xi \geq x)\mathbb{P}(\xi \geq c)^{s-1}}{\mathbb{P}(\xi \geq c)^s} = s\sum_{x=x_0}^\infty  \mathbb{P}(\xi\geq x\mid \xi\geq c).
        \end{align*}
But, for every $x \in \mathbb{N}$ with $x \geq x_0$,
\begin{align}
  \mathbb{P}(\xi\geq x\mid \xi\geq c) &\leq \mathbb{P}\Bigl(\sup_{r\in\mathbb{N}: r \geq x} Z_r \geq M_c \Bigm| \sup_{r \in \mathbb{N}: r \geq c}Z_r \geq M_c\Bigr)  \nonumber\\
&= \mathbb{P}\Bigl(\sup_{r\in\mathbb{N}: r \geq x} (Z_r-Z_c) \geq M_c-Z_c \Bigm| \sup_{r \in \mathbb{N}: r\geq c}(Z_r-Z_c) \geq M_c-Z_c\Bigr)  \nonumber\\
&= \int_0^\infty \mathbb{P}\Bigl(\sup_{r\in \mathbb{N}: r \geq x} (Z_r-Z_c) \geq u \Bigm| \sup_{r\in \mathbb{N}: r \geq c}(Z_r-Z_c) \geq u \Bigr) \, dP_c(u) \nonumber\\
&= \int_0^\infty \mathbb{P}\Bigl(\sup_{r\in \mathbb{N}: r\geq x-c} Z_r \geq u\Bigm| M\geq u\Bigr)\, dP_c(u), \label{Eq:TBC}
\end{align}
where the second equality follows from Lemma~\ref{Lemma:CondionalProb} and the fact that $M_c-Z_c$ is independent of the sequence $(Z_r-Z_c)_{r\in\mathbb{N}:r\geq c}$.  If $b(x-c)/4\geq u \geq b$, then by Lemma~\ref{Lemma:DriftedRandomWalk}(a) and (b) we have
\begin{align*}
  \mathbb{P}\Bigl(\sup_{r\in \mathbb{N}: r\geq x-c} Z_r \geq u \Bigm| M\geq u\Bigr) &\leq \exp\biggl\{-\frac{(b(x-c)+u)^2}{2(x-c)}\biggr\} \cdot  \frac{9ub/2+1}{3\sqrt{ub/(4\pi)}}e^{9ub/4}\\
&\leq e^{-b^2(x-c)/2+5bu/4}\biggl(3\sqrt{\pi ub}+\frac{2\sqrt{\pi}/3}{\sqrt{ub}}\biggr)\\
&\leq e^{-3b^2(x-c)/16}\biggl(\frac{3\sqrt{\pi}}{2}\sqrt{(x-c)b^2} + \frac{2\sqrt{\pi}}{3b}\biggr).
\end{align*}
Since the function $h\mapsto he^{-h^2/2}$ is decreasing for $h\geq 1$, we have that $3\sqrt{\pi(x-c)b^2/4}+2\sqrt{\pi}/(3b) \leq 3e^{b^2(x-c)/16}/2$ for $x-c \geq 60b^{-2}\log(1/b)$, when $b\leq 1/2$.  Thus,
\begin{equation}
\label{Eq:CaseBig}
\mathbb{P}\Bigl(\sup_{r\in\mathbb{N}:r\geq x-c} Z_r \geq u\Bigm| M\geq u\Bigr) \leq \frac{3}{2}e^{-b^2(x-c)/8}.
\end{equation}
On the other hand, if $b > u$ (note that this implies $b(x-c) \geq u$), then by Lemma~\ref{Lemma:DriftedRandomWalk}(a) and (b) we have that
\begin{align*}
\mathbb{P}\Bigl(\sup_{r\in\mathbb{N}:r\geq x-c} Z_r \geq u\Bigm| M\geq u\Bigr) &\leq \exp\biggl\{-\frac{(b(x-c)+u)^2}{2(x-c)}\biggr\} \cdot \frac{\sqrt{2\pi}(1+4b^2)}{2b}e^{2b^2}\\
&\leq e^{-b^2(x-c)/2+2b^2}\biggl(\frac{\sqrt{2\pi}}{2b}+2\sqrt{2\pi}b\biggr)\\
&\leq \frac{\sqrt{2\pi}}{b} e^{-b^2(x-c)/4} \leq \frac{\sqrt{2\pi}}{2^{13/2}}e^{-b^2(x-c)/8},
\end{align*}
where we have used the fact that $x-c\geq 60b^{-2}\log(1/b) \geq 8$ in the final two bounds. Combining the above display with~\eqref{Eq:CaseBig}, we see that for $b(x-c)/4\geq u$, we have
\begin{equation}
\label{Eq:Integrand}
\mathbb{P}\Bigl(\max_{r\in\mathbb{N}:r\geq x-c} Z_r \geq u\Bigm| M\geq u\Bigr) \leq \frac{3}{2}e^{-b^2(x-c)/8}.
\end{equation}
Thus, by reducing $b_0 > 0$ (still depending only on $k$) if necessary, we have for $b \leq b_0$ that 
\begin{align*}
  \int_{0}^\infty \mathbb{P}\Bigl(\sup_{r\in \mathbb{N}: r\geq x-c} Z_r \geq u&\Bigm| M\geq u\Bigr)\, dP_c(u) \\
  &\leq \int_{0}^{b(x-c)/4} \mathbb{P}\Bigl(\sup_{r\in \mathbb{N}: r\geq x-c} Z_r \geq u\Bigm| M\geq u\Bigr) \, dP_c(u)\\
&\hspace{2cm}+ \mathbb{P}\biggl(M_c\geq \frac{b(x-c)}{8}\biggm | \xi \geq c\biggr) + \mathbb{P}\biggl(Z_c\leq -\frac{b(x-c)}{8}\biggm | \xi \geq c\biggr) \\
&\leq \frac{3}{2}e^{-b^2(x-c)/8} + 2e^{-b^2(x-c)/4}e^{2k} + 2\Phi\biggl(-\frac{b(x-9c)}{8\sqrt{c}}\biggr)e^{2k} \leq 5e^{-(b^2x-k)/8},
\end{align*}
where we have used~\eqref{Eq:Integrand}, Lemma~\ref{Lemma:DriftedRandomWalk}(a) and \eqref{Eq:TailBoundArgmax} in the penultimate inequality, and, in the final step, we have used the fact that $x\geq 60c$, the Gaussian tail bound $\Phi(-x)\leq \frac{1}{2}e^{-x^2/2}$ for $x\geq 0$ and the fact that 
\[
\frac{(x-9c)^2}{c} \geq \frac{(51/59)^2(x-c)^2}{c} \geq {59(51/59)^2}(x-c) \geq 32(x-c). 
\]
Combining with~\eqref{Eq:TBC}, we conclude that 
\[
\mathbb{E}\Bigl(\max_{j\in[s]} \xi^j \Bigm| \min_{j\in[s]} \xi^j\geq kb^{-2}\Bigr)-x_0 \leq 5s\sum_{x= x_0 }^\infty e^{-(b^2x-k)/8} \leq \frac{5se^{-15\log(1/b)/2-7}}{1-e^{-b^2/8}}\leq  sb^{5},
\]
as desired, where we have used again the fact that $b\leq 1/2$ in the final inequality.
\end{proof}

\begin{lemma}
	\label{Lemma:Monotone}
	(a) For any $n \in \mathbb{N}$, $0 < p \leq q < 1$ and $x \in \{0,1,\ldots,n\}$, we have
	\begin{equation}
		\label{Eq:BinRatio}
		\frac{\mathbb{P}\bigl(\mathrm{Bin}(n,p) \geq x\bigr)}{\mathbb{P}\bigl(\mathrm{Bin}(n,p/q) \geq x\bigr)} \leq \mathbb{P}\bigl(\mathrm{Bin}(n,q) \geq x\bigr).
	\end{equation}
	(b) Let $W_1,\ldots,W_n$ be independent and identically distributed, real-valued random variables, with corresponding order statistics $W_{(1)} \leq \ldots \leq W_{(n)}$.  Then for every $s \geq t$ and every $m \in [n]$, we have that 
	\[
	\mathbb{P}(W_{(m)} \geq s |W_{(m)} \geq t) \leq \mathbb{P}(W_{(m)} \geq s |W_{(1)} \geq t).
	\]
	In particular, $\mathbb{E}(W_{(m)} | W_{(m)} \geq t) \leq \mathbb{E}(W_{(m)} | W_{(1)} \geq t)$.
\end{lemma}
\begin{proof}
	\noindent (a) Let $g(p)$ denote the left-hand side of~\eqref{Eq:BinRatio}.  It suffices to prove that $g$ is an increasing function on $(0,q]$.  We may also assume that $x \geq 1$, because otherwise the result is clear.  Now, let
	\[
	h(p) := \mathbb{P}\bigl(\mathrm{Bin}(n,p) \geq x\bigr) = \sum_{r=x}^n \binom{n}{r}p^r(1-p)^{n-r}.
	\]
	Then
	\begin{align*}
		h'(p) &= \sum_{r=x}^n \binom{n}{r}rp^{r-1}(1-p)^{n-r} - \sum_{r=x}^{n-1} \binom{n}{r}(n-r)p^r(1-p)^{n-r-1} \\
		&= \sum_{r=x-1}^{n-1} \frac{n!}{r!(n-r-1)!}p^r(1-p)^{n-r-1} - \sum_{r=x}^{n-1} \frac{n!}{r!(n-r-1)!}p^r(1-p)^{n-r-1} \\
		&= \frac{n!}{(x-1)!(n-x)!}p^{x-1}(1-p)^{n-x}.
	\end{align*}
	We can therefore compute
	\begin{align*}
		g'(p) = \frac{h(p/q)h'(p) - h(p)h'(p/q)/q}{h(p/q)^2},
	\end{align*}
	and we note that
	\begin{align*}
		h(p/q)&h'(p) - h(p)h'(p/q)/q \\
		&= \frac{n!}{(x-1)!(n-x)!}p^{x-1}(1-p)^{n-x} \sum_{r=x}^n \binom{n}{r}\Bigl(\frac{p}{q}\Bigr)^r\Bigl(1-\frac{p}{q}\Bigr)^{n-r} \\
		&\hspace{3cm}- \frac{n!}{(x-1)!(n-x)!}\frac{1}{q}\Bigl(\frac{p}{q}\Bigr)^{x-1}\Bigl(1-\frac{p}{q}\Bigr)^{n-x}\sum_{r=x}^n \binom{n}{r}p^r(1-p)^{n-r} \\
		&= \frac{n!p^{x-1}(1-p)^{n-x}(1-p/q)^{n-x}}{q^x(x-1)!(n-x)!}\sum_{r=x}^n \binom{n}{r}p^r\biggl\{\frac{1}{(q-p)^{r-x}} - \frac{1}{(1-p)^{r-x}}\biggr\} \geq 0,
	\end{align*}
	as required.
	
	\medskip
	
	\noindent (b) Write $F$ for the distribution function of $W_1$, and let $\bar{F} := 1 - F$.  We also write $\bar{F}(x-) := \lim_{y \nearrow x} \bar{F}(x)$.  For a Borel measurable set $A \subseteq \mathbb{R}$, let $N(A) := \sum_{i=1}^n \mathbbm{1}_{\{W_i \in A\}}$.  Then, for $s \geq t$,
	\begin{align*}
		\mathbb{P}(W_{(m)} \geq s |W_{(m)} \geq t) &= \frac{\mathbb{P}(W_{(m)} \geq s)}{\mathbb{P}(W_{(m)} \geq t)} = \frac{\mathbb{P}\bigl\{N\bigl([s,\infty)\bigr) \geq n-m+1\bigr\}}{\mathbb{P}\bigl\{N\bigl([t,\infty)\bigr) \geq n-m+1\bigr\}} \\
		&= \frac{\mathbb{P}\bigl\{\mathrm{Bin}\bigl(n,\bar{F}(s-)\bigr) \geq n-m+1\bigr\}}{\mathbb{P}\bigl\{\mathrm{Bin}\bigl(n,\bar{F}(t-)\bigr) \geq n-m+1\bigr\}}.
	\end{align*}
	On the other hand,
	\begin{align*}
		\mathbb{P}(W_{(m)} \geq s |W_{(1)} \geq t) &= \frac{\mathbb{P}\bigl(W_{(m)} \geq s,W_{(1)} \geq t\bigr)}{\mathbb{P}(W_{(1)} \geq t)} \\
		&= \frac{\mathbb{P}\bigl\{N\bigl((-\infty,t)\bigr) = 0,N\bigl([s,\infty)\bigr) \geq n-m+1\bigr\}}{\mathbb{P}\bigl\{N\bigl((-\infty,t)\bigr) = 0\bigr\}} \\
		&= \frac{\sum_{r=n-m+1}^n \binom{n}{r}\bar{F}(s-)^r\bigl\{\bar{F}(t-) - \bar{F}(s-)\bigr\}^{n-r}}{\bar{F}(t-)^n} \\
		&= \mathbb{P}\bigl\{\mathrm{Bin}\bigl(n,\bar{F}(s-)/\bar{F}(t-)\bigr) \geq n-m+1\bigr\}.
	\end{align*}
	The first conclusion therefore follows immediately from (a), and the second conclusion is an immediate consequence of the first.                                                  %
\end{proof}

\begin{lemma}
	\label{Lemma:LogGridSignal}
	Let $v = (v_1, \ldots, v_p)^\top \in \mathbb{R}^p$ be a unit vector. There exists $\ell \in \{0, \ldots, \lfloor\log_2 p\rfloor\}$ such that 
	\[   
	\biggl|\biggl\{j \in [p]: v_j^2 \geq \frac{1}{2^\ell \log_2(2p)}\biggr\}\biggr| \geq 2^\ell. 
	\]
\end{lemma}
\begin{proof}
	The case $p = 1$ is trivially true, so we may assume without loss of generality that $p\geq 2$. Let $L:=\lfloor \log_2p\rfloor$, $b_\ell := 2^{-\ell} \log_2^{-1} (2p)$ and $n_\ell:=\bigl|\bigl\{j: v_j^2 \geq b_\ell \bigr\}\bigr|$ for $\ell \in \{0,\ldots,L\}$.  Assume for a contradiction that $n_\ell < 2^\ell$ for all $\ell$.   Then by Fubini's theorem we have
	\begin{align*}
		\|v\|_2^2 &= \sum_{j=1}^p \int_{t=0}^1 \mathbbm{1}_{\{v_j^2\geq t\}} \,dt \leq n_0(1-b_0)+ \sum_{\ell=1}^L n_\ell(b_{\ell-1}-b_\ell) + pb_L\\
		&\leq \sum_{\ell=1}^L (2^\ell-1)(b_{\ell-1}-b_{\ell}) + pb_L= \sum_{\ell=0}^{L-1} 2^\ell b_\ell + (p- 2^L+1)b_L \leq \frac{L+1}{\log_2 (2p)} \leq 1.
	\end{align*}
	Note that the penultimate inequality is strict if $p+1$ is not an integer power of 2 and the final inequality is strict if $p$ is not an integer power of 2. Since $p \geq 2$, it cannot be the case that we have equality in both equalities, so $\|v\|_2^2 < 1$, which contradicts the fact that $v$ is a unit vector. 
\end{proof}

\begin{lemma}
	\label{Lemma:AuxiliaryTailStats} Define sequences $(a_n)_{n \in \mathbb{N}_0}$ and $(b_n)_{n \in \mathbb{N}_0}$ as follows: $a_0:=b_0:=0$, $b_{n} := (b_{n-1} + 1)\mathbbm{1}_ { \left\{ n \notin \{2^\xi: \xi \in \mathbb{N}_0\} \right\}}$ and $a_{n} := (a_{n-1} + 1)\mathbbm{1}_ { \left\{ n \notin \{2^\xi: \xi \in \mathbb{N}_0\} \right\}} + (b_{n-1} + 1)\mathbbm{1}_ { \left\{ n \in \{2^\xi: \xi \in \mathbb{N}_0\} \right\}}$ for $n \in \mathbb{N}$. Then, we have
	\[
	n/2 \le a_n < 3n/4,
	\]
	for all $n\ge 2$.
\end{lemma}
\begin{proof}
	The two sequences $(a_n)_{n \in \mathbb{N}_0}$ and $(b_n)_{n \in \mathbb{N}_0}$ are tabulated below.
	\begin{center}
		\begin{tabular}{|c|c|c|c|c|c|c|c|c|c|c|c|c|c|c|c|}
			\hline
			$n$ & 0 & 1 & 2 & 3 & 4 & 5 & 6 & 7 & 8  & \ldots  & $2^\xi$ & $2^\xi + 1$  & \ldots & $2^{\xi +1}-1$ & \ldots  \\ \hline
			$a_n$ & 0 & 1 & 1 & 2 & 2 & 3 & 4 & 5 & 4 & \ldots& $2^{\xi -1}$  & $2^{\xi -1}+1$ & \ldots & $3\cdot 2^{\xi -1}-1$ & \ldots \\ \hline
			$b_n$ & 0 & 0 & 0 & 1 & 0 & 1 & 2 & 3 & 0 & \ldots  & 0  & 1 & \ldots & $2^\xi -1$ & \ldots \\ \hline
		\end{tabular}
	\end{center}
	It is clear from the definition of $(b_n)_n$ that $b_{2^\xi + i} =i$ for $\xi \in \mathbb{N}_0$ and $0 \le i \le 2^\xi -1$. Consequently, we have $a_{2^\xi} = b_{2^\xi -1} + 1 = 2^{\xi-1}$ and $a_{2^\xi + i} = 2^{\xi -1} + i$ for $\xi \in \mathbb{N}$ and $1 \le i \le 2^\xi -1$. Hence, we have
	\[
	\frac{1}{2}=\frac{2^{\xi-1}}{2^\xi}\le \frac{a_{2^\xi + i}}{2^\xi +i} = \frac{2^{\xi-1} + i}{2^\xi +i} \le \frac{2^{\xi-1} + 2^\xi -1}{2^\xi +2^\xi -1}< \frac{3}{4},
	\]
	for all $\xi \in \mathbb{N}$ and $0 \le i \le 2^\xi -1$ and the desired result follows.
\end{proof}
\begin{lemma} \label{Lemma:ThreshChi}
	Let $Z_1, \ldots, Z_p \stackrel{\mathrm{iid}}{\sim} \mathcal{N}(0, 1)$. Then for any $a>0$ and $x >0$, we have
	\[
	\mathbb{P} \biggl(  \sum_{j=1}^{p} Z_j^2 \mathbbm{1}_{\{|Z_j| \geq a\}} \geq 6pe^{-a^2/8}+4x \biggr) \leq e^{-x}.  
	\]
\end{lemma}
\begin{proof}
	This proof has some similarities with that of Lemma~17 of \citet{LGS2019}.  By a Chernoff bound, we have for any $u, \lambda > 0$ that,
	\begin{equation} \label{Eq:Chernoff}
		\mathbb{P} \biggl(  \sum_{j=1}^{p} Z_j^2 \mathbbm{1}_{\{|Z_j| \geq a\}} \geq u \biggr)\leq e^{-\lambda u} \bigl\{ \mathbb{E}e^{\lambda Z_1^2 \mathbbm{1}_{\{|Z_j| \geq a\}}}  \bigr\}^p.
	\end{equation}
	We write $p(x):=(2\pi)^{-1/2}x^{-1/2}e^{-x/2}$ for the density of a $\chi_1^2$ distribution. For $\lambda \in (0, 1/4]$, we bound the moment generating function above as follows:
	\begin{align*}
		\mathbb{E}e^{\lambda Z_1^2 \mathbbm{1}_{\{|Z_j| \geq a\}}} &= \int_{a^2}^{\infty} e^{\lambda x} p(x) \, dx \leq 1 + \int_{a^2}^{\infty} (e^{\lambda x}-1)p(x) \, dx = 1 + \int_{a^2}^{\infty} \sum_{k=1}^{\infty} \frac{\lambda^k x^k}{k!}p(x) \, dx \\
		&\leq 1 + \int_{a^2}^{\infty} \! \lambda x e^{\lambda x} p(x) \, dx \leq 1 + \frac{\lambda}{\sqrt{2\pi}}\int_{a^2}^{\infty} \! x^{1/2} e^{-x/4} \, dx = 1 +  \frac{4\lambda}{\sqrt{\pi}}\int_{a/\sqrt{2}}^{\infty} t^2 e^{-t^2/2} \, dt \\
		&=1 + \sqrt{\frac{8}{\pi}}\lambda ae^{-a^2/4} + 4\sqrt{2}\lambda \Bigl\{ 1- \Phi\Bigl(\frac{a}{\sqrt{2}}\Bigr) \Bigr\} \leq 1 + \sqrt{\frac{8}{\pi}}\lambda ae^{-a^2/4} + 2\sqrt{2}\lambda e^{-a^2/4} \\
		&\leq 1 + \biggl(  2\sqrt{\frac{8}{\pi}} e^{-1/2} + 2\sqrt{2}\biggr) \lambda e^{-a^2/8} \leq 1 + 5\lambda e^{-a^2/8},
	\end{align*}
	where we use the fact that $xe^{-x^2/4} \leq 2e^{-1/2}e^{-x^2/8}$ for $x \in \mathbb{R}$ in the penultimate inequality. Hence, by substituting this bound into~\eqref{Eq:Chernoff}, we have for every $u > 0$, that
	\begin{align*}
		\mathbb{P} \biggl(  \sum_{j=1}^{p} Z_j^2 \mathbbm{1}_{\{|Z_j| \geq a\}} \geq u \biggr)&\leq \exp\bigl\{-\lambda u + p \log( 1 + 5 \lambda e^{-a^2/8} ) \bigr \} \leq \exp\bigl( -\lambda u + 5p\lambda e^{-a^2/8} \bigr).
	\end{align*}
	We set $u = 6pe^{-a^2/8} + 4x$. If $x \leq pe^{-a^2/8}/4$, choose $\lambda =p^{-1}xe^{a^2/8} \leq 1/4$; if $x > pe^{-a^2/8}/4$, choose $\lambda = 1/4$. In both cases, we have
	\[
	\mathbb{P} \biggl(  \sum_{j=1}^{p} Z_j^2 \mathbbm{1}_{\{|Z_j| \geq a\}} \geq u \biggr) \leq e^{-x},
	\]
	as required.	
\end{proof}

\noindent \textbf{Acknowledgements}: The research of TW was supported by EPSRC grant EP/T02772X/1 and that of RJS was supported by EPSRC grants EP/P031447/1 and EP/N031938/1. Data for this study come from the High Resolution Seismic Network (HRSN), doi: 10.7932/HRSN, operated by the UC Berkeley Seismological Laboratory, which is archived at the Northern California Earthquake Data Center (NCEDC), doi: 10.7932/NCEDC.  We are grateful to the anonymous reviewers for constructive feedback, which helped to improve the paper.

\end{document}